\documentclass{article}
\pdfoutput=1
\usepackage[preprint]{neurips_2020}
\usepackage{enumerate}
\usepackage{pdfsync}
\usepackage{times}
\usepackage[OT1]{fontenc}
\usepackage[usenames]{color}
\usepackage{booktabs}
\usepackage{mathrsfs}
\usepackage{fullpage}
\usepackage{hyperref}
\usepackage[protrusion=true,expansion=true]{microtype}
\usepackage{float}
\usepackage{setspace}

\usepackage{balance}
\usepackage{enumitem}

\usepackage{amsmath, amsfonts, amssymb, amsthm}
\usepackage[linesnumbered,ruled]{algorithm2e}
\usepackage{pgfplots}
\usepackage{mathtools}
\usepackage{siunitx}
\usepackage{graphicx, caption, subcaption}
\usepackage{verbatim}
\usepackage{multirow}

\RequirePackage{bm} 

\makeatletter
\def\url@leostyle{%
  \@ifundefined{selectfont}{\def\UrlFont{\small\sf}}{\def\UrlFont{\small\sf}}}
\makeatother
\begin{document}

\theoremstyle{definition}
\newtheorem{defn}{Definition}[section]
\newtheorem{experiment}[defn]{Experiment}
\newtheorem{proc}[defn]{Procedure}

\newcommand{\RR}{\mathbb{R}}
\newcommand{\PP}{\mathbb{P}}
\newcommand{\EE}{\mathbb{E}}

\newcommand{\cA}{\mathcal{A}}
\newcommand{\cD}{\mathcal{D}}
\newcommand{\bD}{\mathbf{D}}
\newcommand{\bX}{\mathbf{X}}
\newcommand{\bY}{\mathbf{Y}}
\newcommand{\cM}{\mathcal{M}}
\newcommand{\bbE}{\mathbb{E}}
\newcommand{\cN}{\mathcal{N}}
\newcommand{\cX}{\mathcal{X}}
\newcommand{\cY}{\mathcal{Y}}
\newcommand{\zb}{\mathbf{z}}
\newcommand{\ub}{\mathbf{u}}
\newcommand{\bxi}{\bm{\xi}}
\newcommand{\Ib}{\mathbf{I}}
\newcommand{\bgamma}{\bm{\gamma}}

\newtheorem{theorem}{Theorem}[section]
\newtheorem{corollary}[theorem]{Corollary}
\newtheorem{lemma}[theorem]{Lemma}
\newtheorem{proposition}[theorem]{Proposition}
\newtheorem{definition}[theorem]{Definition}

\newcommand{\opnorms}[1]{%
  \left|\mkern-1.5mu\left|\mkern-1.5mu\left|
   #1
  \right|\mkern-1.5mu\right|\mkern-1.5mu\right|
}
\makeatletter
\newcommand*{\rom}[1]{\expandafter\@slowromancap\romannumeral #1@}
\makeatother

\renewcommand{\baselinestretch}{1.04}
\renewcommand{\floatpagefraction}{.8}

\newcommand{\bfdataset}[1]{{\small{\textsf{#1}}}}
\newcommand{\bfattack}[1]{{\small{\textsf{#1}}}}

\newcommand{\bnote}[1]{\textcolor{purple}{\bf \emph{Bargav: #1}}}

\newcommand{\dnote}[1]{\textcolor{blue}{\bf \emph{Dave: #1}}}

\newcommand{\gnote}[1]{\textcolor{blue}{\bf \emph{Gu: #1}}}

\newcommand{\lnote}[1]{\textcolor{green}{\bf \emph{Lingxiao: #1}}}

\newcommand{\todo}[1]{\textcolor{red}{\bf \emph{TODO: #1}}}

\newcommand\shortsection[1]{\vspace{6pt}{\noindent\bf #1.}}

\newcommand\shortdsection[1]{\vspace{3pt}{\noindent -- #1:}}

\title{Revisiting Membership Inference \\ Under Realistic Assumptions}

\author{%
  Bargav Jayaraman$^\vee$, 
  Lingxiao Wang$^\gamma$, 
Katherine Knipmeyer$^\vee$, \\
{\bf Quanquan Gu}$^\gamma${\bf ,}
{\bf David Evans}$^\vee$ \\[1.0ex]
  \emph{University of Virginia} ($\vee$) and
  \emph{University of California Los Angeles} ($\gamma$)
  \\
}

\maketitle

\begin{abstract}\noindent
We study membership inference in settings where some of the assumptions typically used in previous research are relaxed. First, we consider skewed priors, to cover cases such as when only a small fraction of the candidate pool targeted by the adversary are actually members and develop a PPV-based metric suitable for this setting. This setting is more realistic than the balanced prior setting typically considered by researchers. Second, we consider adversaries that select inference thresholds according to their attack goals and develop a threshold selection procedure that improves inference attacks. Since previous inference attacks fail in imbalanced prior setting, we develop a new inference attack based on the intuition that inputs corresponding to training set members will be near a local minimum in the loss function, and show that an attack that combines this with thresholds on the per-instance loss can achieve high PPV even in settings where other attacks appear to be ineffective.
\end{abstract}

\section{Introduction}
Differential privacy has become the gold standard for performing any privacy-preserving statistical analysis over sensitive data. 
Its privacy-utility tradeoff is controlled by the privacy loss budget parameter $\epsilon$ (and failure probability $\delta$). 
While it is a well known fact that larger privacy loss budgets lead to more leakage, it is still an open question how low privacy loss budgets should be to provide meaningful privacy in practice. 

Although differential privacy provides strong bounds on the worst-case privacy loss, it does not elucidate what privacy attacks could be realized in practice. Attacks, on the other hand, provide an empirical lower bound on privacy leakage for a particular setting. Many attacks on machine learning algorithms have been proposed that aim to infer private information about the model or the training data. These attacks include membership inference~\citep{shokri2017membership, long2017towards, salem2019ml, yeom2018privacy}, attribute inference~\citep{fredrikson2014privacy, fredrikson2015model, yeom2018privacy}, property inference~\citep{ateniese2015hacking, ganju2018property}, model stealing~\citep{lowd2005adversarial, tramer2016stealing} and hyperparameter stealing~\citep{wang2018stealing, yan2018cache}. Of these, membership inference attacks are most directly connected to the differential privacy definition, and thus are a good basis for evaluating the privacy leakage of differentially private mechanisms. Given a small enough privacy loss budget, a differentially private mechanism should provide a defense against these attacks. But, in practice it is rarely possible to obtain a model with enough utility without increasing the privacy loss budget beyond the minimum needed to establish such guarantees. Instead, models are tested using empirical methods using simulated attacks to understand how much an adversary would be able to infer. Previous works on membership inference attacks only consider balanced priors, however, leading to a skewed understanding of inference risk in cases where models are likely to face adversaries with imbalanced priors. In this work, we develop a metric based on positive predictive value that captures the inference risk even in scenarios where the priors are skewed, and introduce a new attack strategy that shows models are vulnerable to inference attacks even in settings where previous attacks would be unable to infer anything useful.

\shortsection{Theoretical Contributions}
Motivated by recent results~\citep{jayaraman2019evaluating, liu2019investigating}, we aim to develop more useful privacy metrics. Similarly to \cite{liu2019investigating}, we adopt a hypothesis testing perspective on differential privacy in which the adversary uses hypothesis testing on the differentially private mechanism's output to make inferences about its private training data. We use the recently proposed $f$-differential privacy notion (see Section~\ref{sec:fdp}) to bound the privacy leakage of the mechanism. Using this hypothesis testing framework, we tighten the theoretical bound on the advantage metric (Section~\ref{sec:advantage_metric}). Then, we show that this metric alone does not suffice in most realistic scenarios since it does not consider the prior probability of the data distribution from which the adversary chooses records.
We propose using positive predictive value (PPV) in conjunction with the advantage metric as it captures this notion, and provide a theoretical analysis of this metric (Section~\ref{sec:ppv_metric}). 

\shortsection{Empirical Contributions}
We provide a threshold selection procedure that can be used to improve any threshold-based inference attack to better capture how an adversary with a particular goal would use the attack (Section~\ref{sec:threshold_selection}). We use this procedure for the loss-based attack of \cite{yeom2018privacy} and the confidence-based attack of \cite{shokri2017membership}, as well as for two new attacks. We propose a novel inference attack strategy that samples points around the candidate input to gauge if it is near a local minimum in the loss function (Section~\ref{sec:merlin}). The \bfattack{Merlin} attack uses this strategy to decide if an input is a member based on a threshold on the ratio of samples where the loss value increases. Our \bfattack{Morgan} attack (Section~\ref{sec:morgan}) combines this with thresholds on the per-instance loss value. Finally, we use these attacks to empirically evaluate the privacy leakage of neural networks trained both with and without differential privacy on four multi-class data sets considering balanced and imbalanced prior data distribution (Section~\ref{sec:threshold_selection_results}). 
Our main empirical findings include:
\begin{itemize}
    \item Non-private models are vulnerable to high-confidence membership inference attacks in both balanced and imbalanced prior settings.
    \item PPV changes with the prior and hence it is a more reliable metric in imbalanced prior settings.
    \item The \bfattack{Morgan} attack achieves higher PPV than \bfattack{Merlin}, which already outperforms previous attacks.
    \item Private models can be vulnerable to our attacks, but only when privacy loss budgets are well above the theoretical guarantees.
\end{itemize}
\section{Related Work}\label{sec:related_works}

While statistical membership inference attacks were demonstrated on genomic data in the late 2000s~\citep{homer2008resolving, sankararaman2009genomic}, the first membership inference attacks against machine learning models were performed by \cite{shokri2017membership}. In these attacks, the attacker exploits the model confidence reflecting overfitting to infer membership. \cite{shokri2017membership} consider the balanced prior setting and evaluate the attack success with an accuracy metric. The attacker trains shadow models similar to the target model, and uses these shadow models to train a membership inference model. 
\cite{yeom2018privacy} proposed a simpler, but usually more effective, attack based on per-instance loss and proposed using membership advantage metric for attack evaluation as it has theoretical interpretation with differential privacy. 

Yeom et al.'s membership advantage metric is useful for balanced prior settings, but not representative of true privacy leakage in realistic scenarios (as we demonstrate in Section~\ref{sec:privacy_metrics}).
\cite{rahman2018membership} evaluate differentially private mechanisms against membership inference attacks and use accuracy and F-score as privacy leakage metrics. But they do not specify the theoretical relationship between their privacy leakage metrics and the privacy loss budgets (i.e., how the metric would scale with increasing privacy loss budget) necessary to gain insight as to what privacy loss budgets are safe even in the worst case scenarios. \cite{jayaraman2019evaluating} evaluate the private mechanisms against both membership inference and attribute inference attacks using the advantage privacy leakage metric of \cite{yeom2018privacy}. All the above works consider a balanced prior data distribution probability and hence are not applicable to settings where the prior probability is skewed. 

\cite{liu2019investigating} theoretically evaluate differentially private mechanisms using a hypothesis testing framework using precision, recall and F-score metrics. They give a theoretical relationship connecting these metrics to the differential privacy parameters ($\epsilon$ and $\delta$) and give some insights for choosing the parameter values based on the background knowledge of the adversary. Recently, \cite{balle2019hypothesis} provided hypothesis testing framework for analysing the relaxed variants of differential privacy that use R\'{e}nyi divergence. However, neither of the above works provide empirical evaluation of privacy leakage of the private mechanisms. In another recent work, \cite{farokhi2020modelling} propose using conditional mutual information as the privacy leakage metric and derive its upper bound based on Kullback–Leibler divergence. Although they provide a relationship between this upper bound and the privacy loss budget, they do not evaluate the empirical privacy leakage in terms of the proposed metric.  We provide a theoretical analysis of privacy leakage metrics and perform membership inference attacks under the more realistic assumptions of different prior data distribution probabilities and an adversary that can adaptively pick inference thresholds based on specific attack goals.
\section{Differential Privacy}

This section provides background on the differential privacy notions we use. Table~\ref{tab:notations} summarizes the notations we use throughout.

\begin{table}[tb]
    \centering\small
    \begin{tabular}{cp{11.5cm}}
        Notation & \multicolumn{1}{c}{Description} \\ \toprule
        $\cD = \cX \times \cY$ &  Distribution of records with features sampled form $\cX$ and labels sampled from $\cY$\\
        $S \sim \cD^n$ & Data set $S$ consisting of $n$ records, sampled i.i.d. from distribution $\cD$ \\
        $\zb \sim S$ & Record $\zb$ is picked uniformly random from data set $S$ \\
        $\zb \sim \cD$ & Record $\zb$ is chosen according to the distribution $\cD$ \\
        $\cM_S$ & Model obtained by using a learning algorithm $\cM$ over data set $S$ \\
        $\cA$ & Membership inference adversary \\
        $p$ & Probability of sampling a record from training set \\
        $\gamma$ & Test-to-train set ratio, $(1 - p)/p$\\
        $\epsilon$ & Privacy loss budget of DP  mechanism \\
        $\delta$ & Failure probability of DP mechanism \\
        $\alpha$ & False positive rate of inference adversary \\
        $\beta$ & False negative rate of inference adversary \\
        $\phi$ & Decision threshold of inference adversary; also called rejection rule in hypothesis testing \\
        $\Phi$ & Cumulative distribution function of standard normal distribution \\ \bottomrule
    \end{tabular}
    \caption{Notation}
    \label{tab:notations}
\end{table}

\citet{dwork2006calibrating} introduced a formal notion of privacy that provides a probabilistic information-theoretic security guarantee:
\begin{definition}[Differential Privacy]
A randomized algorithm $\cM$ is \emph{$(\epsilon, \delta)$-differentially private} if for any pair of neighbouring data sets $S, S'$ that differ by one record, and any set of outputs $O$,
\begin{equation*}
    Pr[\cM(S) \in O] \le e^\epsilon Pr[\cM(S') \in O] + \delta.
\end{equation*}
\end{definition}
Thus, the ratio of output probabilities across neighbouring data sets is bounded by the $\epsilon$ and $\delta$ parameters. The intuition behind this definition is to make any pairs of neighbouring data sets indistinguishable to the adversary given the information released.

From a hypothesis testing perspective \citep{wasserman2010statistical,kairouz2017composition,liu2019investigating,balle2019hypothesis,dong2019gaussian}, the adversary can be viewed as  
performing the following hypothesis testing problem given the ouput of either $\cM(S)$ or $\cM(S')$:
\begin{align*}
    H_0 &: \text{the underlying data set is $S$}, \\ 
    H_1 &: \text{the underlying data set is $S'$}.
\end{align*}
According to the definition of differential privacy, given the information released by the private algorithm $\cM$, the hardness of this hypothesis testing problem for the adversary is measured by the worst-case likelihood ratio between the distributions of the outputs $\cM(S)$ and $\cM(S')$. Following \cite{wasserman2010statistical}, a more natural way to characterize the hardness of this hypothesis testing problem is its type I and type II errors and can be formulated in terms of finding a rejection rule $\phi$ which  trades off between type I and type II errors in an optimal way. In other words, for a fixed type I error $\alpha$, the adversary tries to find a rejection rule $\phi$ that minimizes the type II error $\beta$. More specifically, recalling the definition of trade-off function from \cite{dong2019gaussian}:
\begin{definition}[Trade-off Function]\label{def:trade-off}
For any two probability distributions $P$ and $Q$ on the same space, the \emph{trade-off function} $T(P, Q) : [0, 1] \rightarrow [0, 1]$ is defined as: 
\begin{equation*}
    T(P, Q)(\alpha) = \inf \{\beta_\phi : \alpha_\phi \le \alpha\},
\end{equation*}
where the infimum is taken over all (measurable) rejection rules, and $\alpha_\phi$ and $\beta_\phi$ are the type I and type II errors for the rejection rule $\phi$.
\end{definition}
This definition suggests that the larger the trade-off function is, the harder the hypothesis testing problem will be. It has been established in \cite{dong2019gaussian} that a function $f : [0, 1] \rightarrow [0, 1]$ is a trade-off function if and only if it is convex, continuous, non-increasing, and $f(x) \le 1 - x$ for $x \in [0,1]$. Thus, differential privacy can be reformulated as finding the trade-off function $f$ that limits the adversary's hypothesis testing power, i.e., it maximizes the adversary's type II error for any given type I error.

\subsection{\emph{f}-Differential Privacy}\label{sec:fdp}

The hypothesis testing formulation of differential privacy described above leads to the notion of $f$-differential privacy~\citep{dong2019gaussian} ($f$-DP)
which aims to find the optimal trade-off between type I and type II errors and will be used to derive the theoretical upper bounds of our proposed metrics for the privacy leakage.
\begin{definition}[$f$-Differential Privacy]
Let $f$ be a trade-off function. A mechanism $\cM$ is \emph{$f$-differentially private} if for all neighbouring data sets $S$ and $S'$:
\begin{equation*}
    T(\cM(S), \cM(S')) \ge f.
\end{equation*}
\end{definition}

Note that in the above definition, we abuse the notations of $\cM(S)$ and $\cM(S')$ to represent their corresponding distributions. For an $(\epsilon, \delta)$-differentially private algorithm, the trade-off function $f_{\epsilon, \delta}$ is given by Lemma~\ref{lemma:trade_off_eps_del} (proved by \cite{wasserman2010statistical} and \cite{kairouz2017composition}): 

\begin{lemma}[\cite{wasserman2010statistical,kairouz2017composition}]\label{lemma:trade_off_eps_del}
Suppose $\cM$ is an $(\epsilon, \delta)$-differentially private algorithm, then for a false positive rate of $\alpha$, the trade-off function is:
\begin{equation*}
    f_{\epsilon, \delta}(\alpha) = \max\{0, 1 - \delta - e^\epsilon \alpha, e^{-\epsilon}(1 - \delta - \alpha)\}.
\end{equation*}
\end{lemma}

This lemma suggests that higher values of $f_{\epsilon, \delta}(\alpha)$ correspond to more privacy and perfect privacy would require $f_{\epsilon,\delta}(\alpha) = 1 - \alpha$. In addition, increasing $\epsilon$ and $\delta$ decreases  $f_{\epsilon, \delta}(\alpha)$, reflecting the expected reduction in privacy.

\subsection{Gaussian Differential Privacy}\label{sec:gdp}
The Gaussian mechanism is one of the most popular and fundamental approaches for achieving differential privacy, especially for differentially private deep learning~\citep{abadi2016deep}.
Noisy stochastic gradient descent (SGD) and noisy Adam \citep{tf-privacy}, i.e., adding Gaussian noise (Gaussian mechanism) to SGD and Adam, are often used as the underlying private optimizers for training neural networks with privacy guarantees.  Precisely characterizing the privacy loss of the composition of Gaussian mechanisms and deriving its sub-sampling amplification results, leads to the notion of Gaussian differential privacy~\citep{dong2019gaussian}, which belongs to the family of $f$-DP with a single parameter $\mu$ that defines the mean of the Gaussian distribution.

\begin{definition}[$\mu$-Gaussian Differential Privacy]
A mechanism $\cM$ is \emph{$\mu$-Gaussian differentially private} if for all neighbouring data sets $S$ and $S'$:
\begin{equation*}
    T(\cM(S), \cM(S')) \ge G_\mu,
\end{equation*}
where $G_\mu = T(\cN(0, 1), \cN(\mu, 1))$.
\end{definition}
In this definition, $G_\mu$ is a trade-off function and hence $\mu$-GDP is identical to $f$-DP where $f = G_\mu$. Lemma~\ref{lemma:trade_off_gdp}, which is established in \cite{dong2019gaussian}, gives the equation for computing $G_\mu$:
\begin{lemma}\label{lemma:trade_off_gdp}
Given that $\cM$ is a $\mu$-Gaussian differentially private algorithm, then for a false positive rate of $\alpha$, the trade-off function is given as:
\begin{equation*}
    G_\mu(\alpha) = \Phi(\Phi^{-1}(1 - \alpha) - \mu),
\end{equation*}
where $\Phi$ is the cumulative distribution function of standard normal distribution.
\end{lemma}

The Gaussian mechanism for $\mu$-GDP is given by the following theorem \citep{dong2019gaussian}.
\begin{theorem}\label{th:gdp_noise}
A mechanism $\cM$ operating on a statistic $\theta$ is $\mu$-GDP if $\cM(S) = \theta(S) + \zeta$, where $\zeta \sim \cN(0, \nabla(\theta)^2 / \mu^2)$ and $\nabla(\theta)$ is the global sensitivity of $\theta$.
\end{theorem}

We also have the relationship between $\mu$-GDP and $(\epsilon, \delta)$-DP as follows (Corollary 2.13 in \cite{dong2019gaussian} and Theorem 8 in \cite{balle2018improving}):

\begin{proposition}\label{prop:gdp_primal_to_dual}
A mechanism is $\mu$-GDP if and only if it is $(\epsilon, \delta(\epsilon))$-DP for all $\epsilon \ge 0$, where
\[\delta(\epsilon) = \Phi\bigg(-\frac{\epsilon}{\mu} + \frac{\mu}{2}\bigg) - e^\epsilon \Phi\bigg(-\frac{\epsilon}{\mu} - \frac{\mu}{2}\bigg).\]
\end{proposition}

Gaussian differential privacy supports lossless composition of mechanisms (Corollary 3.3 in \cite{dong2019gaussian}) and privacy amplification due to sub-sampling (Theorem 4.2 in \cite{dong2019gaussian}).  
\begin{theorem}[Composition]\label{th:gdp_composition}
The $T$-fold composition of $\mu_i$-GDP mechanisms is $\sqrt{\mu_1^2 + \cdots + \mu_T^2}-GDP$.
\end{theorem}
\begin{theorem}[Sub-sampling]\label{th:gdp_subsampling}
Suppose $\cM$ is $f$-DP on $\cD^m$, if we apply $\cM$ to a subset of samples with sampling ratio $ \tau = m/n\in[0,1]$, then $\cM$ is $\min \{f_\tau, f_\tau^{-1}\}^{**}$-DP, where $f_\tau = \tau f + (1 - \tau)Id$. 
\end{theorem}
The function $Id$ is the identity function defined as $Id(x) = 1 - x$, and $\min \{f_\tau, f_\tau^{-1}\}^{**}$ is the double conjugate of $\min \{f_\tau, f_\tau^{-1}\}$ function. 
Theorems~\ref{th:gdp_noise}, \ref{th:gdp_composition} and \ref{th:gdp_subsampling} can be combined to achieve GDP for gradient perturbation based machine learning algorithms such as noisy SGD and noisy Adam. For instance, adding standard Gaussian noise with standard deviation $\sigma$ to each batch of gradients sampled with probability $\tau$ would lead to $\tau\sqrt{T(e^{1/\sigma^2} - 1)}$-GDP for $T$ compositions~\citep{bu2019deep}. In our experiments, we use such result to characterize the privacy loss of our private optimizers for training differentially private neural networks, and use Proposition \ref{prop:gdp_primal_to_dual} to convert it into $(\epsilon,\delta)$-DP for the purpose of comparison. 
\section{Measuring Privacy Leakage}\label{sec:privacy_metrics}
To evaluate privacy leakage, we define an adversarial game inspired by \cite{yeom2018privacy}. Unlike their game which assumes a balanced prior, our game factors in the prior membership distribution probability. The adversarial game models the scenario where an adversary has access to a model, $\cM_S$, trained over a data set $S$, knowledge of the training procedure and data distribution, and wishes to infer whether a given input is a member of that training set.
\begin{experiment}[Membership Experiment]\label{membership_experiment}
Assume a membership adversary, $\cA$, who has information about the training data set size $n$, the distribution $\cD$ from which the data set is sampled, and the prior membership probability $p$. The adversary runs this experiment:
\begin{enumerate}
    \item Sample a training set $S \sim \cD^n$ and train a model $\cM_S$ over the training set $S$.
    \item Randomly sample $b \in \{0, 1\}$, such that $b = 1$ with probability $p$.
    \item If $b = 1$, then sample $\zb \sim S$; else sample $\zb \sim \cD$.
    \item Output 1 if $\cA(\zb, \cM_S, n, \cD) = b$; otherwise output 0. 
\end{enumerate}
\end{experiment}
Note that our experiment incorporates the prior probability $p$ of sampling a record, compared to the setting of Yeom et al.\ that assumes balanced prior probability ($p = 0.5$). We consider skewed prior $p$ as inferring membership is more important than inferring non-membership in our problem setting. This is different from the semantic security analogue where all messages are treated equally regardless of the skewness of the message distribution. For most practical scenarios (that is, where being exposed as a member carries meaningful risk to an individual), $p$ is much smaller than 0.5. For instance, for a scenario of an epidemic outbreak, the training set could be the list of patients with the disease symptoms admitted at a hospital. The non-members can be the remaining population of the city or a district. Hence, assuming a balanced prior of $p=0.5$ is not a realistic assumption, and it is important to develop a privacy metric that can be used to evaluate scenarios with lower (or higher) priors.

\begin{figure}[tb]
    \centering
    \includegraphics[width=0.5\linewidth]{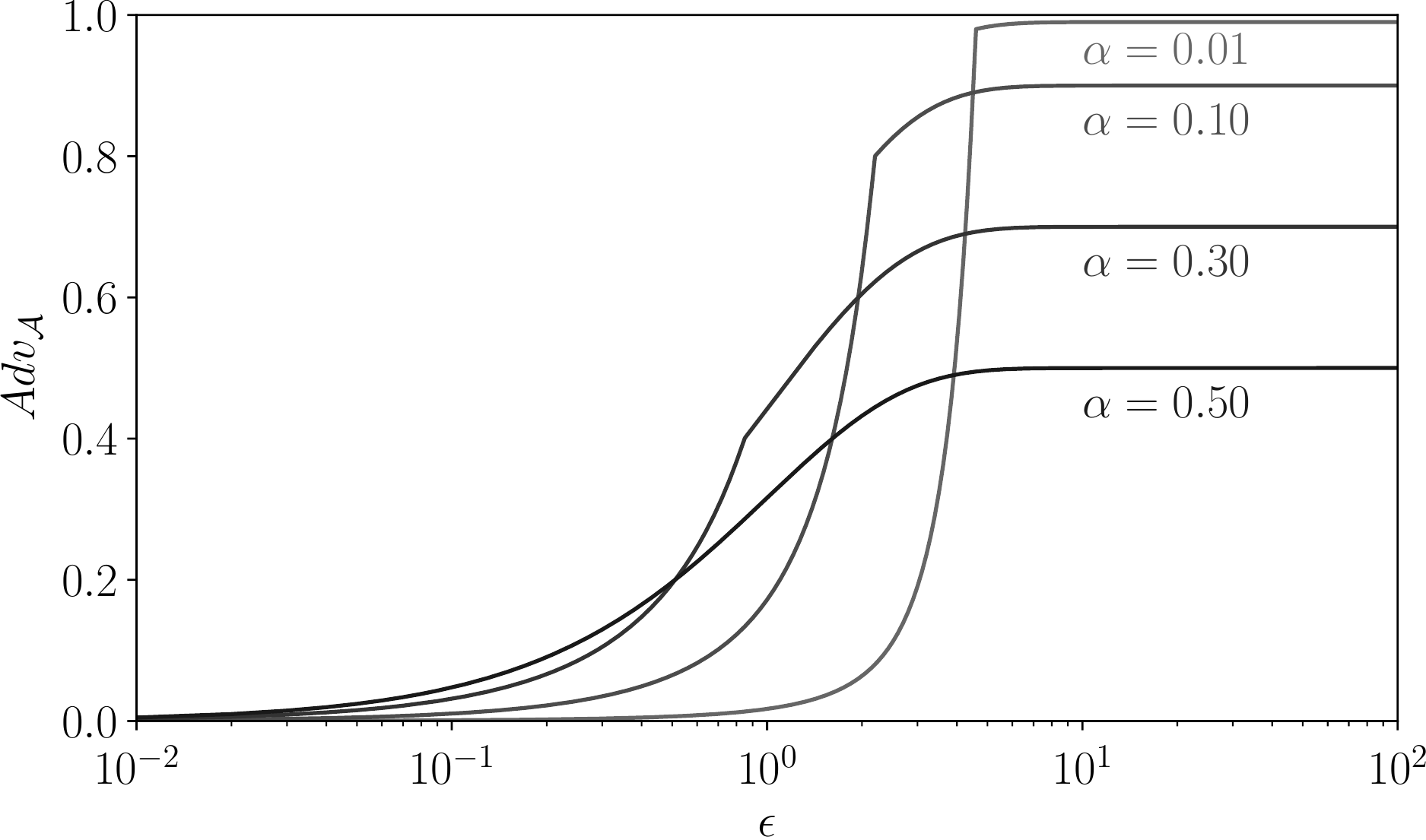}
    \caption{Theoretical upper bounds on $\mathit{Adv}_\cA(\alpha)$ metric for various privacy loss budgets with varying $\alpha$ ($\delta = 10^{-5}$).}
    \label{fig:adv_bound}
\end{figure}

\begin{figure}[tb]
    \centering
    \includegraphics[width=0.5\linewidth]{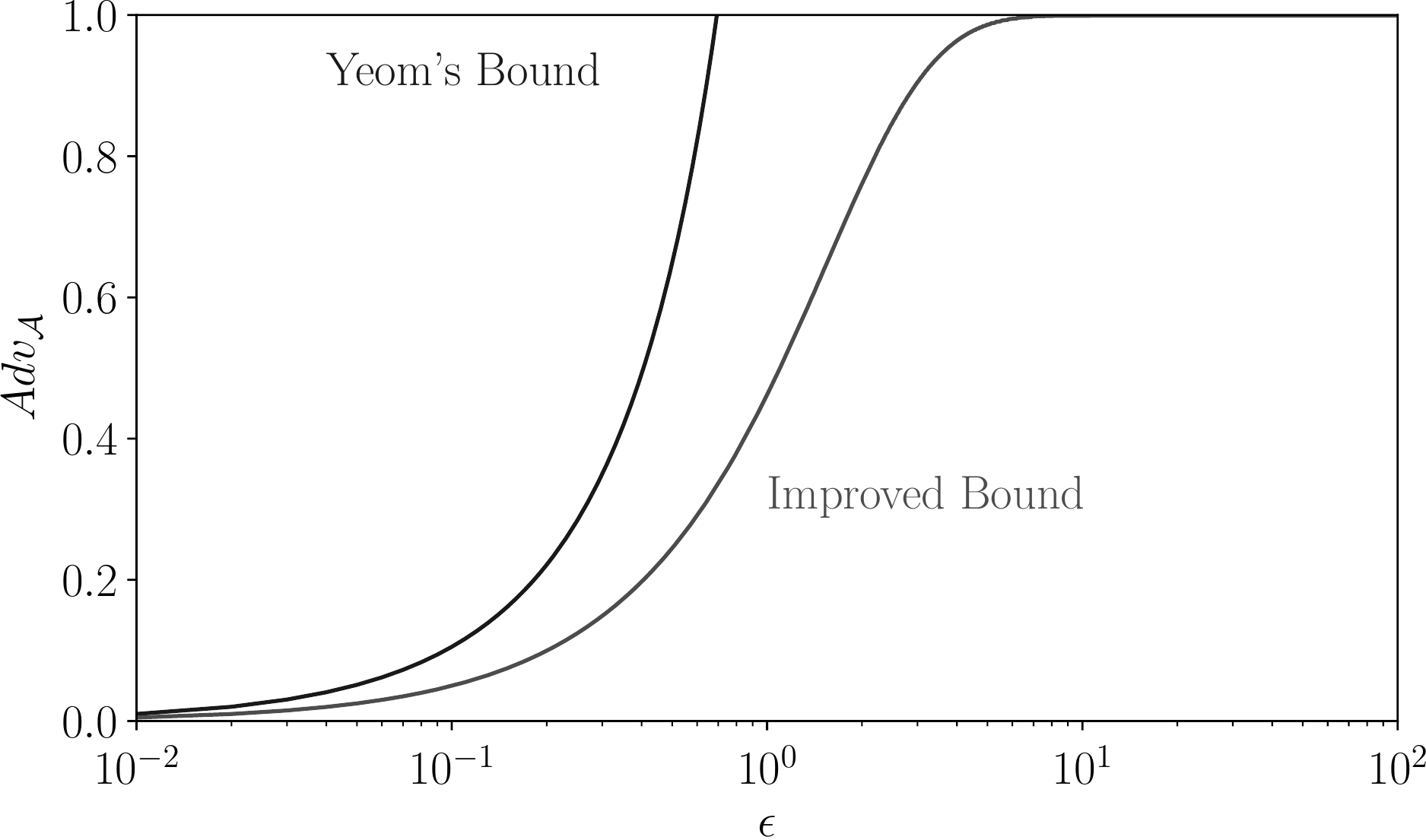}
    \caption{Comparing theoretical bounds on membership advantage ($\delta = 0$). \rm Improved bound uses Theorem~\ref{thm:adv} to get maximum advantage across all $0 < \alpha \le 1$.}
    \label{fig:max_adv_bound}
\end{figure}

\subsection{Membership Advantage}\label{sec:advantage_metric}
The membership advantage metric, $\mathit{Adv}$, was defined by \cite{yeom2018privacy} as the difference between the true positive rate and the false positive rate for the membership inference adversary provided that $p=0.5$ (i.e., balanced prior membership distribution). Yeom et al.\ showed that for an $\epsilon$-differentially private mechanism, the theoretical upper bound for membership advantage is $e^\epsilon - 1$, which can be quite loose for higher $\epsilon$ values and is not defined for $e^\epsilon - 1 > 1$ since the membership advantage metric proposed by Yeom et al.\ is only defined between $0$ and $1$.
Moreover, the bound is not valid for $(\epsilon, \delta)$-differentially private algorithms which are more commonly used for private deep learning. 

We derive a tighter bound for the membership advantage metric that is applicable to $(\epsilon, \delta)$-differentially private algorithms based on the notion of $f$-DP:
\begin{theorem}\label{thm:adv}
Let $\cM$ be an $(\epsilon, \delta)$-differentially private algorithm. For any randomly chosen record $\zb$ and fixed false positive rate $\alpha$, the membership advantage of a membership inference adversary $\cA$ is bounded by:
\begin{equation*}
    \mathit{Adv}_{\cA}(\alpha) \le 1 - f_{\epsilon, \delta}(\alpha) - \alpha,
\end{equation*}
where $f_{\epsilon, \delta}(\alpha) = \max\big\{0, 1 - \delta - e^\epsilon \alpha, e^{-\epsilon}(1 - \delta - \alpha)\big\}$.
\end{theorem}

\begin{proof}[Proof of Theorem~\ref{thm:adv}] The proof follows directly from Yeom at al.'s definition, $\mathit{Adv}_{\cA}(\alpha) = \mathit{TPR} - \mathit{FPR}$, when we have balanced prior membership distribution, $p=0.5$. For a given $FPR = \alpha$, we have $1-TPR \geq f_{\epsilon, \delta}(\alpha)$ according to the definition of trade-off function (Definition \ref{def:trade-off} and Lemma \ref{lemma:trade_off_eps_del}). Therefore, $\mathit{Adv}_{\cA}(\alpha) \le 1 - f_{\epsilon, \delta}(\alpha) - \alpha.$
\end{proof}

Figure~\ref{fig:adv_bound} shows the relationship between the false positive rate $\alpha$ of a given membership inference adversary and the upper bound of the advantage given by Theorem~\ref{thm:adv}. This bound lies strictly between 0 and 1 and is tighter than the bound of \cite{yeom2018privacy}, as shown in Figure~\ref{fig:max_adv_bound}. However, this metric is limited to balanced prior distribution of data and hence can overestimate (or underestimate) the privacy threat in any scenario where the prior probability is not $0.5$. Thus, membership advantage alone is not a reliable way to measure the privacy leakage. Hence, we next propose the positive predictive value metric that considers the prior distribution of data.

\begin{figure*}[tb]
\centering
\begin{subfigure}[b]{0.45\textwidth}
\centering
\includegraphics[width=\linewidth]{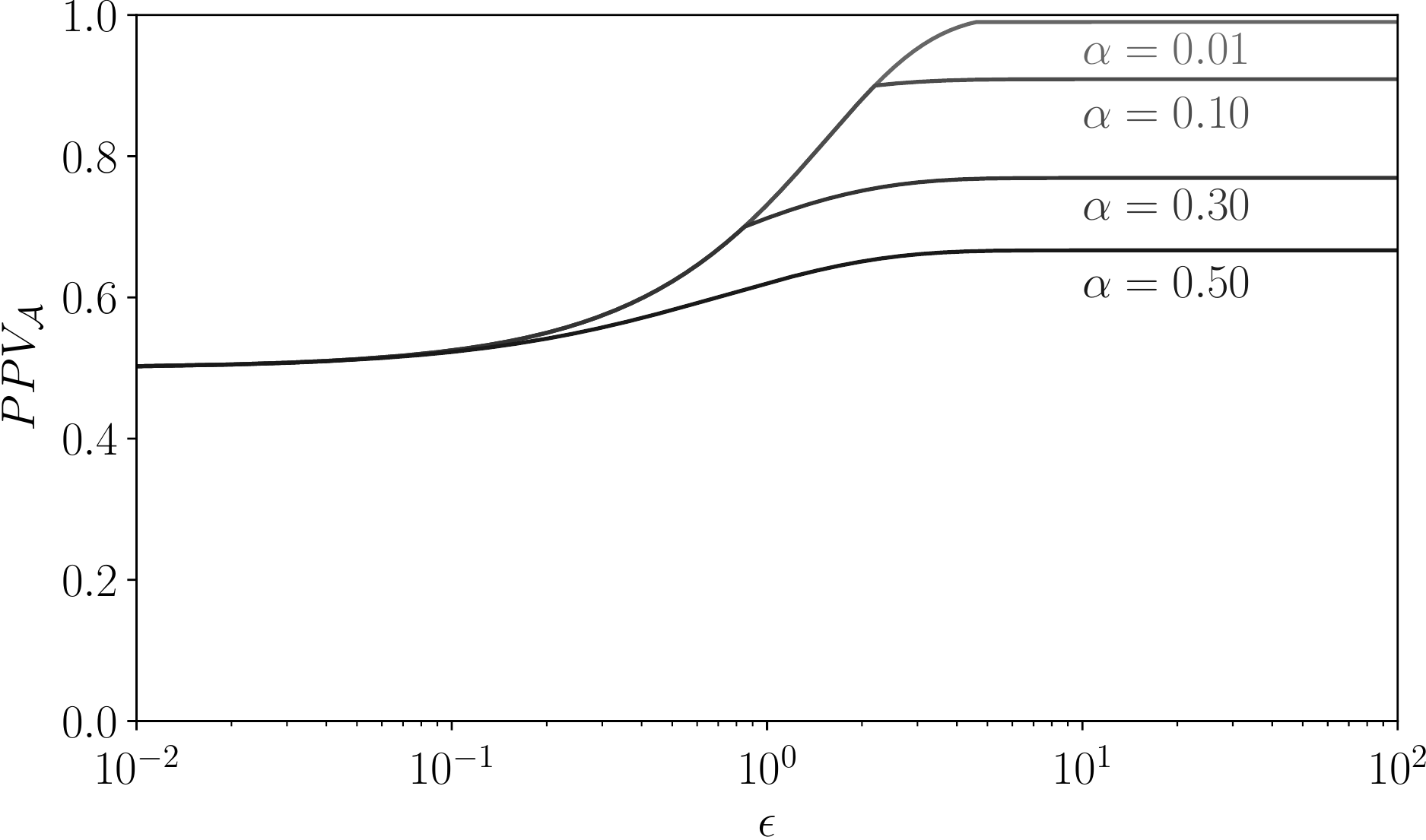}
\caption{$\mathit{PPV}_\cA$ with varying $\alpha$ ($\gamma = 1$)}
\label{fig:ppv_bound_alpha}
\end{subfigure} \qquad
\begin{subfigure}[b]{0.45\textwidth}
\centering
\includegraphics[width=\linewidth]{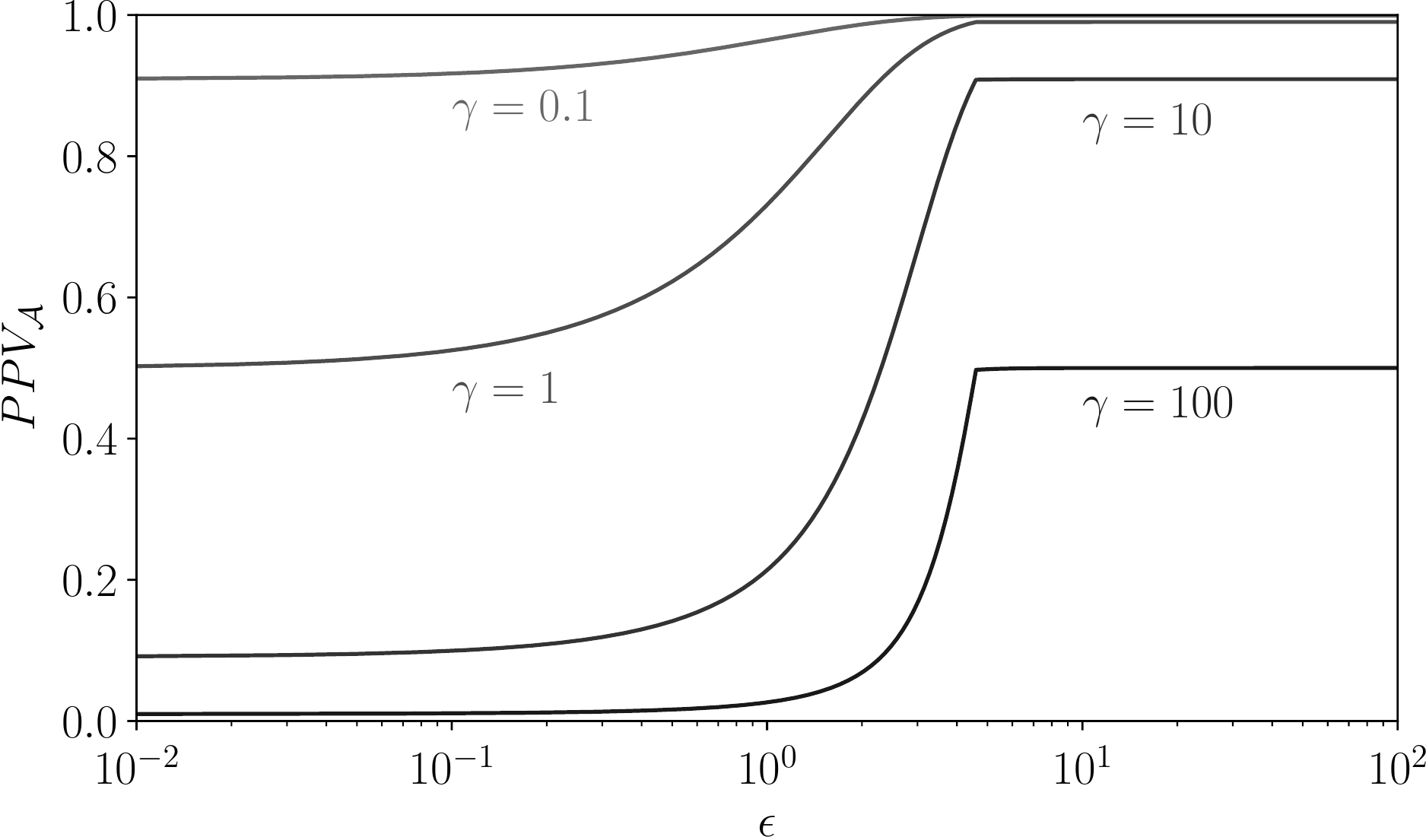}
\caption{$\mathit{PPV}_\cA$ with varying $\gamma$ ($\alpha = 0.01$)}
\label{fig:ppv_bound_gamma}
\end{subfigure}
\caption{Theoretical upper bounds on PPV metric for various privacy budgets ($\delta = 10^{-5}$).}
\end{figure*}

\subsection{Positive Predictive Value}\label{sec:ppv_metric}
Positive predictive value (PPV) gives the ratio of true members predicted among all the positive membership predictions made by an adversary (the precision of the adversary). For an $(\epsilon, \delta)$-differentially private algorithm, the PPV is bounded by the following theorem:
\begin{theorem}\label{thm:ppv}
Let $\cM$ be an $(\epsilon, \delta)$-differentially private algorithm and $\cA$ be a membership inference adversary. For any randomly chosen record $\zb$ and a fixed false positive rate of $\alpha$, the positive predictive value of $\cA$ is bounded by
\begin{equation*}
    \mathit{PPV}_{\cA}(\alpha, \gamma) \le \frac{1 - f_{\epsilon, \delta}(\alpha)}{1 - f_{\epsilon, \delta}(\alpha) + \gamma \alpha},
\end{equation*}
where 
$f_{\epsilon, \delta}(\alpha) = \max\big\{0, 1 - \delta - e^\epsilon \alpha, e^{-\epsilon}(1 - \delta - \alpha)\big\}$, $\gamma = (1-p)/p$, and $p$ is the prior membership probability defined in Membership Experiment \ref{membership_experiment}.
\end{theorem}
\begin{proof}[Proof of Theorem~\ref{thm:ppv}]
According to the trade-off function definition (Definition \ref{def:trade-off} and Lemma \ref{lemma:trade_off_eps_del}), for a given $FPR = \alpha$, we have $1-TPR \geq f_{\epsilon, \delta}(\alpha)$. Since $\mathit{PPV}_{\cA}(\alpha, \gamma) = TP/(TP + FP)$, we can obtain:
\begin{equation*}
    \mathit{PPV}_{\cA}(\alpha, \gamma) = \frac{TPR}{TPR + \gamma \cdot FPR} \le \frac{1 - f_{\epsilon, \delta}(\alpha)}{1 - f_{\epsilon, \delta}(\alpha) + \gamma \alpha}.
\end{equation*}
\end{proof}
Like membership advantage, the PPV metric is strictly bounded between 0 and 1. Moreover, the bound on PPV metric considers the prior distribution via $\gamma$, which gives the ratio of probability of selecting a non-member to a member. This allows the PPV metric to better capture the privacy threat across different settings. Figure~\ref{fig:ppv_bound_alpha} shows the effect of varying the false positive rate $\alpha$ and Figure~\ref{fig:ppv_bound_gamma} shows the effect of varying the prior distribution probability $\gamma$ on the PPV metric. For example, for $\epsilon = 5, \delta = 10^{-5}, \alpha = 0.01, \gamma = 100$, the advantage metric can be as high as 0.98, while the PPV metric is close to 0.5 (i.e., coin toss probability). Thus, in such cases, advantage grossly overestimates the privacy threat. 
\section{Inference Attacks}\label{sec:attacks}

While the previous section covers the metrics to evaluate privacy leakage, here we discuss about the membership inference attack procedures. In Section~\ref{sec:threshold_selection}, we describe our threshold selection procedure for threshold-based inference attacks. Section~\ref{sec:proposed_mi} presents our threshold-based inference attack that perturbs a query record and uses the direction of change in per-instance loss of the record for membership inference. Section~\ref{sec:morgan} presents our second attack that combines our first attack with the threshold-based attack of \cite{yeom2018privacy}.

\subsection{Setting the Decision Threshold}\label{sec:threshold_selection}

The membership inference attacks we consider need to output a Boolean result for each test, converting a real number measure from a test into a Boolean that indicates whether or not a given input is considered a member. The effectiveness of an attack depends critically on the value of this decision threshold. 

We introduce a simple procedure to select the decision threshold for any threshold-based attack where the adversary's goal is to maximize leakage for a given expected maximum false positive rate:

\begin{proc}[Finding Decision Threshold]\label{algo:inference_training}
Given an adversary, $\cA$, that knows information about a target model including the training data distribution $\cD$, training set size $n$, training procedure, and model architecture, as well as knowing the prior distribution probability $p$ for the suspected membership set, this procedure finds a threshold $\phi$ that maximizes the privacy leakage of the sampled data points for a given maximum false positive rate $\alpha$. 
\begin{enumerate}
    \item Sample a training data set $\bar S \sim \cD^n$ for training a model $\cM_{\bar S}$.
    \item Randomly sample $b \in \{0, 1\}$, such that $b = 1$ with probability $p$. 
    \item Sample record $\zb \sim \bar S$ if $b = 1$, otherwise $\zb \sim \cD$.
    \item Output the decision threshold, $\phi$, that maximizes its true positive rate constrained to a maximum false positive rate of $\alpha$ for the inference attack, $\cA(\zb, \cM_{\bar S}, n, \cD, \phi)$. 
\end{enumerate}
\end{proc}

Note that in comparison to Experiment~\ref{membership_experiment}, the adversary $\cA$ takes an additional parameter $\phi$ here. With this $\phi$, the adversary can query the target model $\cM_S$ to perform membership inference. Procedure~\ref{algo:inference_training} works for any threshold-based inference attack where an adversary knows the data distribution and model training process well enough to train its own models similar to the target model. 

\shortsection{Application to Yeom's Attack} 
The membership inference attack of \cite{yeom2018privacy} uses per-instance loss information for membership inference. Given a loss $\ell(\zb, \cM_S)$ on the query record $\zb$, their approach classifies it as a member if the loss is less than the expected training loss. Using our Procedure~\ref{algo:inference_training}, we can instead use a threshold $\phi$ for membership inference that corresponds to an expected maximum false positive rate $\alpha$. In other words, if the per-instance loss $\ell(\zb, \cM_S) \le \phi$, then $\zb$ is classified as a member of the target model's training set $S$, otherwise it is classified as a non-member. We refer to this membership inference adversary as \bfattack{Yeom}.

\shortsection{Application to Shokri's Attack}
In the membership inference attack of \cite{shokri2017membership}, the attacker first trains multiple shadow models similar to the target model, and then uses these shadow models to train an inference model for binary classification. We modify this attack by taking the softmax output of the inference model that indicates the model's prediction confidence, and use our threshold selection procedure on the model confidence. By default, the model predicts the input is a member if the confidence is above 0.5, which is equivalent to Shokri et al.'s original version. We vary this threshold between 0 and 1 according to Procedure~\ref{algo:inference_training}, and refer to this inference adversary as \bfattack{Shokri}.

\subsection{Merlin}\label{sec:proposed_mi}\label{sec:merlin}
Procedure~\ref{algo:inference_training} can be used on any threshold-based inference attack. Here, we introduce a new threshold-based membership inference attack called \bfattack{Merlin}\footnote{Backronym for {\bf ME}asuring {\bf R}elative {\bf L}oss {\bf I}n {\bf N}eighborhood.} that uses a different approach to infer membership. Instead of the per-instance loss of a record, this method uses the direction of change in per-instance loss of the record when it is perturbed with a small amount of noise. The intuition here is that due to overfitting, the target model's loss on a training set record will tend to be close to a local minimum, so the loss at perturbed points near the original input will be higher. On the other hand, the loss is equally likely to either increase or decrease for a non-member record.

Algorithm~\ref{algo:proposed_mi_2} describes the attack procedure. For a query record $\zb$, random Gaussian noise with zero mean and standard deviation $\sigma$ is added and the change of loss direction is recorded. This step is repeated $T$ times and the $count$ is incremented each time the per-instance loss of the perturbed record increases. Though we use Gaussian noise, the algorithm works for other noise distributions as well. We also tried uniform distribution and observed similar results, but with different $\sigma$ values. Both the parameters $T$ and $\sigma$ can be pre-tuned on a hold-out set to maximize the attacker's distinguishing power and fixed for the entire attack process. In our experiments, we find $T = 100$ and $\sigma = 0.01$ work well across all data sets. Finally, the query record $\zb$ is classified as a member when $count / T \ge \phi$, where $\phi$ is a threshold that could be set by Procedure~\ref{algo:inference_training}.

\begin{algorithm}[tb]
 \SetKwInOut{Input}{Input}
 \SetKwInOut{Output}{Output}
 \underline{$\cA(\zb, \cM_S, n, \cD, \phi)$}:\\
 \Input{$\zb$: input record, $\cM_S$: model trained on data set $S$ of size $n$, $\cD$: data distribution, $\phi$: decision function, $T$: number of repeat, $\sigma$: standard deviation parameter }
 \Output{membership prediction of $\zb$ (0 or 1)}
    $count \leftarrow 0$ \;
    \For{$T$ runs} {
        $\bxi \sim \cN(0, \sigma^2\Ib)$ \tcp*{Sample Gaussian noise}
        \If{$\ell(\zb + \bxi, \cM_S) > \ell(\zb, \cM_S)$} {
            $count \leftarrow count + 1$ \;
        }
    }
    \Return {$count / T \ge \phi$} \tcp*{1 if `member'}
 \caption{Inference Using Direction of Change in Per-Instance Loss}
 \label{algo:proposed_mi_2}
\end{algorithm}

\shortsection{Comparison with Related Attacks}
Although the intuition behind the \bfattack{Merlin} is new, it has similarities with previous attacks that also involve sampling. \cite{fredrikson2015model} proposed a white-box attack for model inversion problem, which is different from the membership inference problem we consider, where the attacker has \emph{count} information of all training instances and uses it to guess the most probable value for the sensitive attribute of the query training instance. This `count' is different from the count used in \bfattack{Merlin} attack. \cite{long2018understanding} proposed a black-box model inversion attack that is similar to \bfattack{Merlin}. While the \bfattack{Merlin} attack considers the target point's environment in the input space, the attacks in \cite{long2018understanding} consider the target point's environment in the logit-space, i.e., the output of the target network before the softmax is applied. As the logit-space is much more dense than the input space, \bfattack{Merlin} is much more fine-grained, enabling it to detect membership where the logit-space attacks would not. \cite{choo2020label} recently proposed a label-only membership inference attack which is similar to \bfattack{Merlin} in the sense that they also use the model's behavior on \emph{neighboring points} as part of a membership inference attack. The key difference is that they assume the neighboring points, which in their case are data augmentations of the target record, are also present in the training set, while we do not have any such assumptions for \bfattack{Merlin}.

\subsection{Morgan}\label{sec:morgan}
Both \bfattack{Yeom} and \bfattack{Merlin} use different information for membership inference and hence do not necessarily identify the same member records. Some members are more vulnerable to one attack than the other, and different inputs produce false positives for each attack. Our observations of the distribution of the values from the \bfattack{Yeom} and \bfattack{Merlin} attacks (see Figure~\ref{fig:morgan_purchase_plots}) motivate combining the attacks in a way that can maximize PPV by excluding points with very low per-instance loss. The intuition is that if the per-instance loss is extremely low, the \bfattack{Merlin} attack will suggest a local minimum, but in fact it is a near-global minimum, which is not as strongly correlated with being a member. 
Hence, we introduce a combination of the \bfattack{Yeom} and \bfattack{Merlin} attacks, called \bfattack{Morgan}\footnote{{\bf M}easuring l{\bf O}ss, {\bf R}elatively {\bf G}reater {\bf A}round {\bf N}eighborhood.}, that combines both attacks to identify inputs that are most likely to be members. 

The \bfattack{Morgan} attack uses three thresholds: a lower threshold on per-instance loss $\phi_L$, an upper threshold on per-instance loss $\phi_U$, and a threshold on the ratio as used by \bfattack{Merlin}, $\phi_M$. \bfattack{Morgan} classifies a record as member if the per-instance loss of the record is between $\phi_L$ and $\phi_U$, both inclusive, and has a \bfattack{Merlin} ratio of at least $\phi_M$. The $\phi_U$ and $\phi_M$ thresholds are set using the standard threshold selection procedure for the \bfattack{Yeom} and \bfattack{Merlin} attacks respectively, by varying their $\alpha$ values. A value for $\phi_L$ is found using a grid search to find the maximum PPV possible in conjunction with $\phi_U$ and $\phi_M$ thresholds, and selecting the lowest value for $\phi_L$ that achieves that PPV to maximize the number of members identified. Note that all three thresholds are selected together to maximize the PPV on a separate holdout set that is disjoint from the target training set, as is done in our threshold selection procedure~\ref{algo:inference_training}. As reported in Table~\ref{tab:key_findings}, this exposes some members with 100\% PPV for both \bfdataset{RCV1X} and \bfdataset{CIFAR-100}. Section~\ref{sec:threshold_selection_results} reports on \bfattack{Morgan}'s success on identifying the most vulnerable records with $>95\%$ PPV at balanced prior and with $>90\%$ PPV in skewed prior cases ($\gamma > 1$).
\section{Experimental Setup}\label{sec:exp_setup}

This section describes the data sets and models used, along with the training procedure. We evaluate our methods on both standard (non-private) models and models trained using differential privacy mechanisms. We focus on differentially private models since our theoretical bounds apply to these models.  Although several other defenses have been proposed, such as dropout, model stacking or MemGuard~\citep{jia2019memguard}, our theoretical bounds do not apply to them and we do not include them in our evaluation.\footnote{Our attacks and experimental tests do, however, and it will be interesting to see how effective non-DP defenses are against our attacks, so we do plan to include evaluations of other defenses in future work.}  

Table~\ref{tab:key_findings} summarizes the data sets used and the performance of non-private models trained over each data set, and the maximum PPV of the most effective membership inference attack (\bfattack{Morgan}). In the balanced prior setting ($\gamma = 1$), some members are exposed with very high confidence ($>95\%$ PPV) for all the test data sets. The membership inference is significant even in the imbalanced prior case, when $\gamma = 10$. We defer discussion of these results to Section~\ref{sec:threshold_selection_results}.

\sisetup{
  input-decimal-markers = .,input-ignore = {,},table-number-alignment = right,
  group-separator={,}, group-four-digits = true
}

\begin{table*}[tp]
    \centering
    \small
    \begin{tabular}{cS[table-format=4.0]S[table-format=4.0]S[table-format=1.2]S[table-format=1.2]S[table-format=3.1,separate-uncertainty,table-figures-uncertainty=1]S[table-format=2.1,separate-uncertainty,table-figures-uncertainty=1]}
        Data set & {\#Features} & {\#Classes} & {Train Acc} & {Test Acc} & {PPV at $\gamma = 1$} & {PPV at $\gamma = 10$} \\ \hline
        \bfdataset{Purchase-100X} & 600 & 100 & 1.00 & 0.71 & 98.0 \pm 4.0 & 97.5 \pm 5.0 \\
        \bfdataset{Texas-100} & 6,000 & 100 & 1.00 & 0.53 & 95.7 \pm 4.6 & {\rm\em (insufficient  data)} \\ 
        \bfdataset{RCV1X} & 2,000 & 52 & 1.00 & 0.84 & 100.0 \pm 0.0 & 93.0 \pm 9.8 \\
        \bfdataset{CIFAR-100} & 3,072 & 100 & 0.48 & 0.18 & 100.0 \pm 0.0 & {\rm\em (insufficient data)} \\
    \end{tabular}
    \caption{Summary of data sets and results for non-private models. \rm Maximum PPV achieved by \bfattack{Morgan} is reported, averaged across five runs.}
    \label{tab:key_findings}
\end{table*}

\shortsection{Data Sets}
Multi-class classification tasks are more vulnerable to membership inference, as shown in prior works on both black-box~\citep{shokri2017membership, yeom2018privacy} and white-box~\citep{nasr2019comprehensive} attacks.
Hence, we select four multi-class classification tasks for our experiments. Although these data sets are public, they are representative of data sets that contain potentially sensitive information about individuals. 

\shortdsection{\bfdataset{Purchase-100X}} 
\cite{shokri2017membership} created \bfdataset{Purchase-100} data set by extracting customer transactions from Kaggle's acquire valued customers challenge~\citep{purchase}. The authors \emph{arbitrarily} selected 600 items from the transactions data and considered only those customers who purchased at least one of the 600 items. Their resulting data set consisted of 197,000 customer records with 600 binary features representing the customer purchase history. The records are clustered into 100 classes, each representing a unique purchase style, such that the goal is to predict a customer's purchase style.
Since we needed more records for our experiments with the $\gamma = 10$ setting, we curated our own data set by following the same procedure but instead of 600 arbitrary items taking the 600 \emph{most frequently} purchased items. This resulted in an expanded, but similar, data set with around 300,000 customer records which we call \bfdataset{Purchase-100X}.

\shortdsection{\bfdataset{Texas-100}} The Texas hospital data set, also used by \cite{shokri2017membership}, consists of 67,000 patient records with 6,000 binary features where each feature represents a patient's medical attribute. This data set also has 100 output classes where the task is to identify the main procedure that was performed on the patient. This data set is too small for tests with high $\gamma$ settings, but a useful benchmark for the other settings.
    
\shortdsection{\bfdataset{RCV1X}} The Reuters RCV1 corpus data set \citep{lewis2004rcv1} is a collection of Reuters newswire articles with more than 800,000 documents, a 47,000-word vocabulary and 103 classes. The original 103 classes are arranged in a hierarchical manner, and each article can belong to more than one class. We follow data pre-processing procedures similar to \cite{srivastava2014dropout} to obtain a data set such that each article only belongs to a single class.  The final data set we use has 420,000 articles, 2,000 most frequent words represented by their term frequency–inverse document frequency (TFIDF) which are used as features and 52 classes. We call our expanded data set \bfdataset{RCV1X}.

\shortdsection{\bfdataset{CIFAR-100}} We use the standard CIFAR-100~\citep{krizhevsky2009learning} data set used in machine learning which consists of 60,000 images of 100 common world objects. The task is to identify an object based on the input RGB image consisting of $32 \times 32$ pixels. Although the privacy issue here is not clear, we include this data set in our experiments because it is used as a benchmark in many privacy works.

\vspace{1ex}\noindent
All the above data sets are pre-processed such that the $\ell_2$ norm of each record is bounded by 1. This is a standard pre-processing procedure that improves model performance that is used by many prior works~\citep{chaudhuri2011differentially, jayaraman2018distributed}.

\shortsection{Model Architecture} We train neural networks with two hidden layers using ReLU activation. Each hidden layer has 256 neurons and the output layer is a softmax layer. Several previous works used similar multi-layer ReLU network architectures to analyze privacy-preserving machine learning~\citep{shokri2015privacy, abadi2016deep, shokri2017membership}. Details on hyperparameters can be found in Appendix~\ref{appendix:hyperparameters}. Table~\ref{tab:key_findings} includes the training and test accuracy of non-private models across the four data sets.\footnote{As with all of the experimental results we report in this paper, the results are averaged over five runs in which the target model is trained from the scratch for each run.} Although we tuned the model hyperparameters to maximize the test accuracy for each data set, there is a considerable gap between the training and test accuracy. This generalization gap indicates that the model overfits the training data, and hence, there is information in the model that could be exploited by inference attacks.

\begin{figure*}[tb]
    \centering
    \begin{subfigure}[b]{0.245\textwidth}
    \centering
    \includegraphics[width=\linewidth]{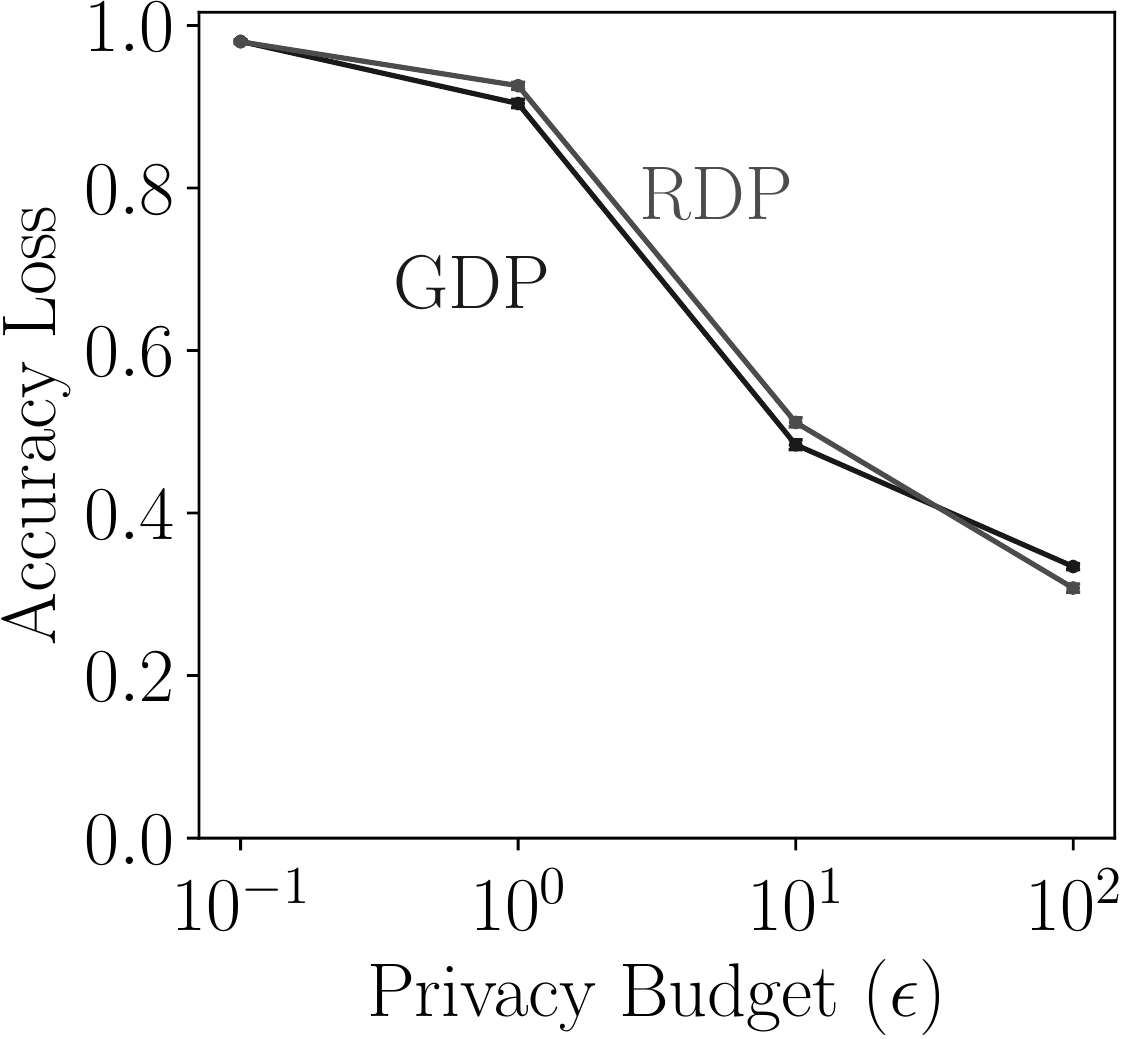}
    \caption{\bfdataset{Purchase-100X}}\label{fig:acc_loss_purchase_100}
    \end{subfigure}
    \begin{subfigure}[b]{0.245\textwidth}
    \centering
    \includegraphics[width=\linewidth]{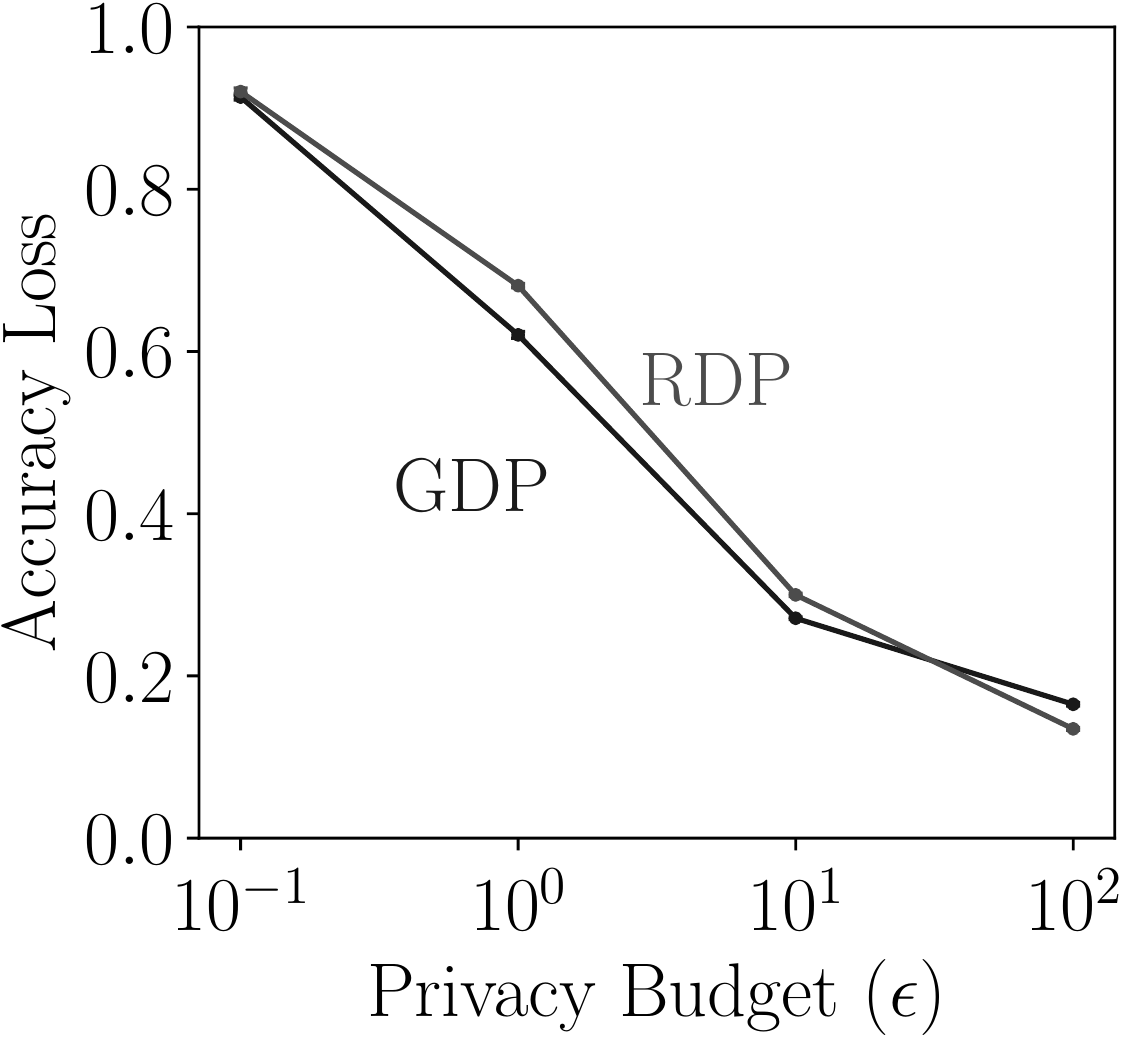}
    \caption{\bfdataset{Texas-100}}\label{fig:acc_loss_texas_100}
    \end{subfigure}
    \begin{subfigure}[b]{0.245\textwidth}
    \centering
    \includegraphics[width=\linewidth]{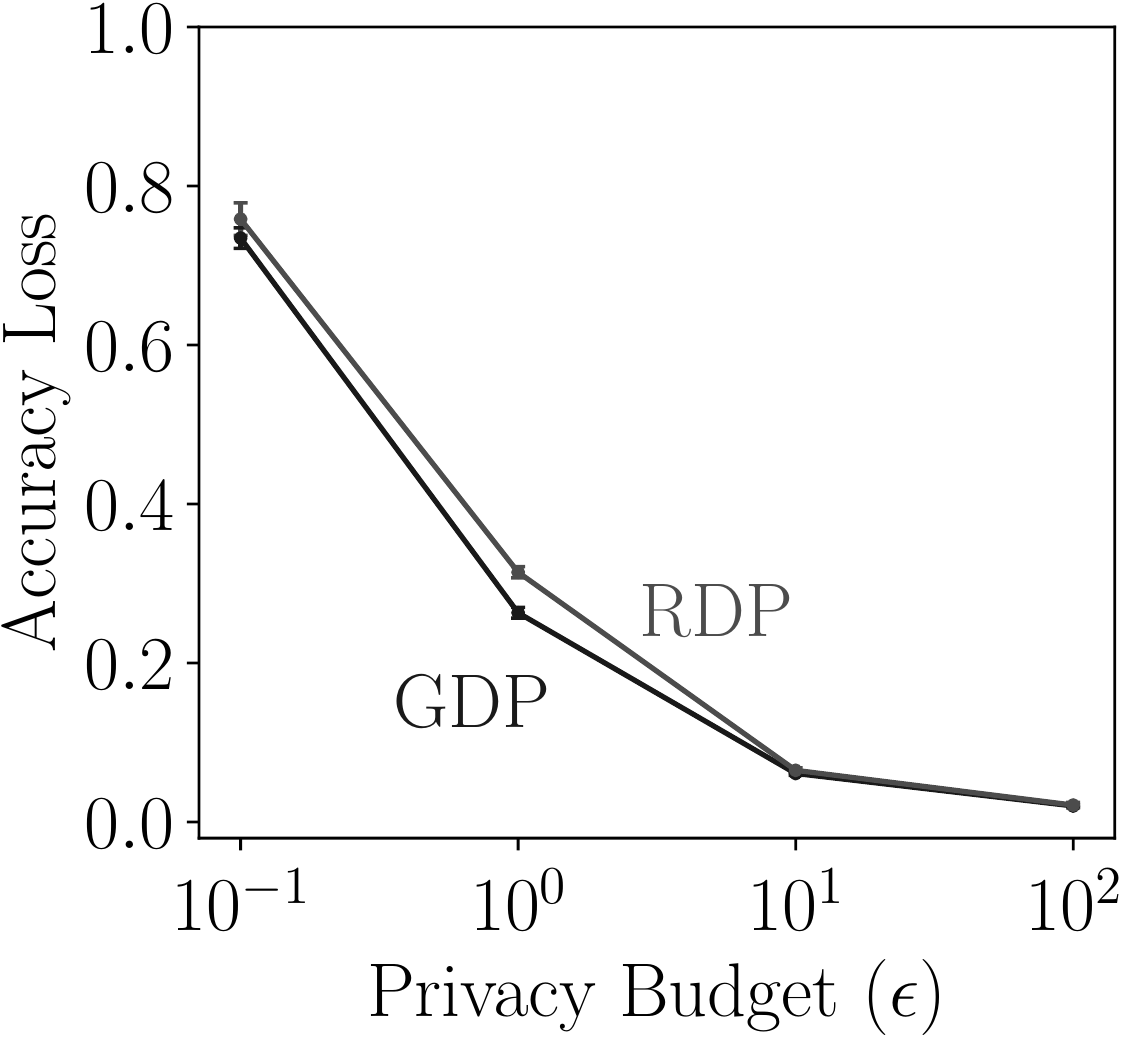}
    \caption{\bfdataset{RCV1X}}\label{fig:acc_loss_rcv1}
    \end{subfigure}
    \begin{subfigure}[b]{0.245\textwidth}
    \centering
    \includegraphics[width=\linewidth]{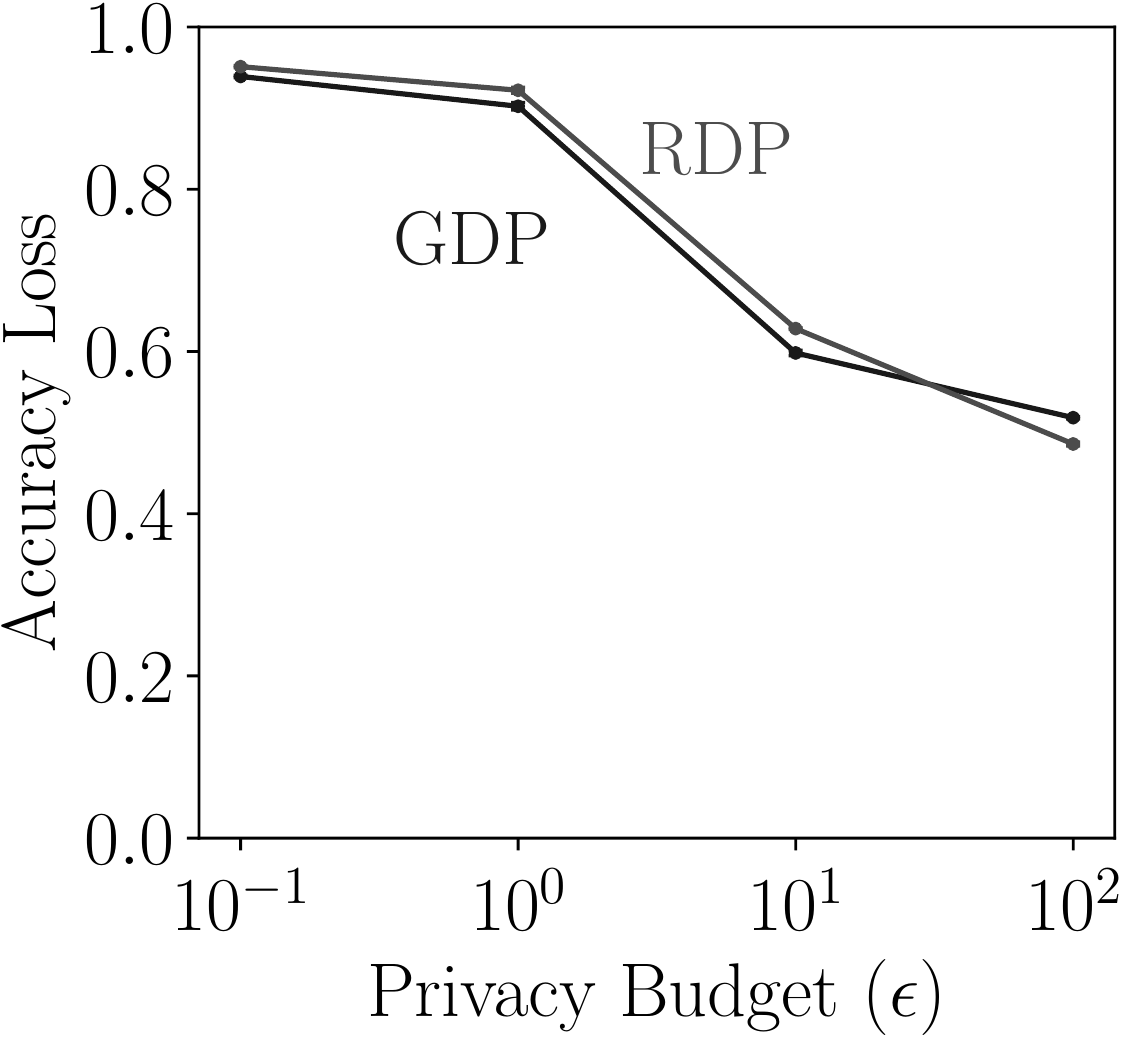}
    \caption{\bfdataset{CIFAR-100}}\label{fig:acc_loss_cifar_100}
    \end{subfigure}
    \caption{Accuracy loss comparison of private models trained with different privacy analyses.}
    \label{fig:acc_loss}
\end{figure*} 

\shortsection{Private Model Training}\label{sec:model_accuracy}
We evaluate the model accuracy of private neural network models trained on different data sets. We vary the privacy loss budget $\epsilon$ between 0.1 and 100 for differentially private training and repeat the experiments five times for all the settings to report the average results.

We report the \emph{accuracy loss}, which gives the relative loss in test accuracy of private models with respect to non-private baseline:
\[\text{\em Accuracy Loss} = 1 - \frac{\text{\em Accuracy of Private Model}}{\text{\em Accuracy of Non-Private Model}}\]

\noindent
Figure~\ref{fig:acc_loss} gives the accuracy loss of differentially private models trained on different data sets with varying privacy loss budgets. The private models are trained using the gradient perturbation mechanism where the gradients at each epoch are clipped and Gaussian noise is added to preserve privacy. The privacy accounting for composition of mechanisms is done via both Gaussian differential privacy (GDP)~\citep{dong2019gaussian} and the prior state-of-the-art R\'{e}nyi differential privacy (RDP)~\citep{mironov2017renyi}. As shown in the figure, the GDP mechanism has a lower accuracy loss for $\epsilon \le 10$ due to its tighter privacy analysis. The GDP composition theorem requires that the individual mechanisms be highly private, and hence it is hard to reduce noise for $\epsilon > 10$ without increasing the failure probability $\delta$. For all the data sets, GDP performs better than RDP, hence we only report the results for GDP in the remaining experiments.

\begin{figure}[tb]
    \centering
    \begin{subfigure}[b]{0.35\textwidth}
    \centering
    \includegraphics[width=.95\linewidth]{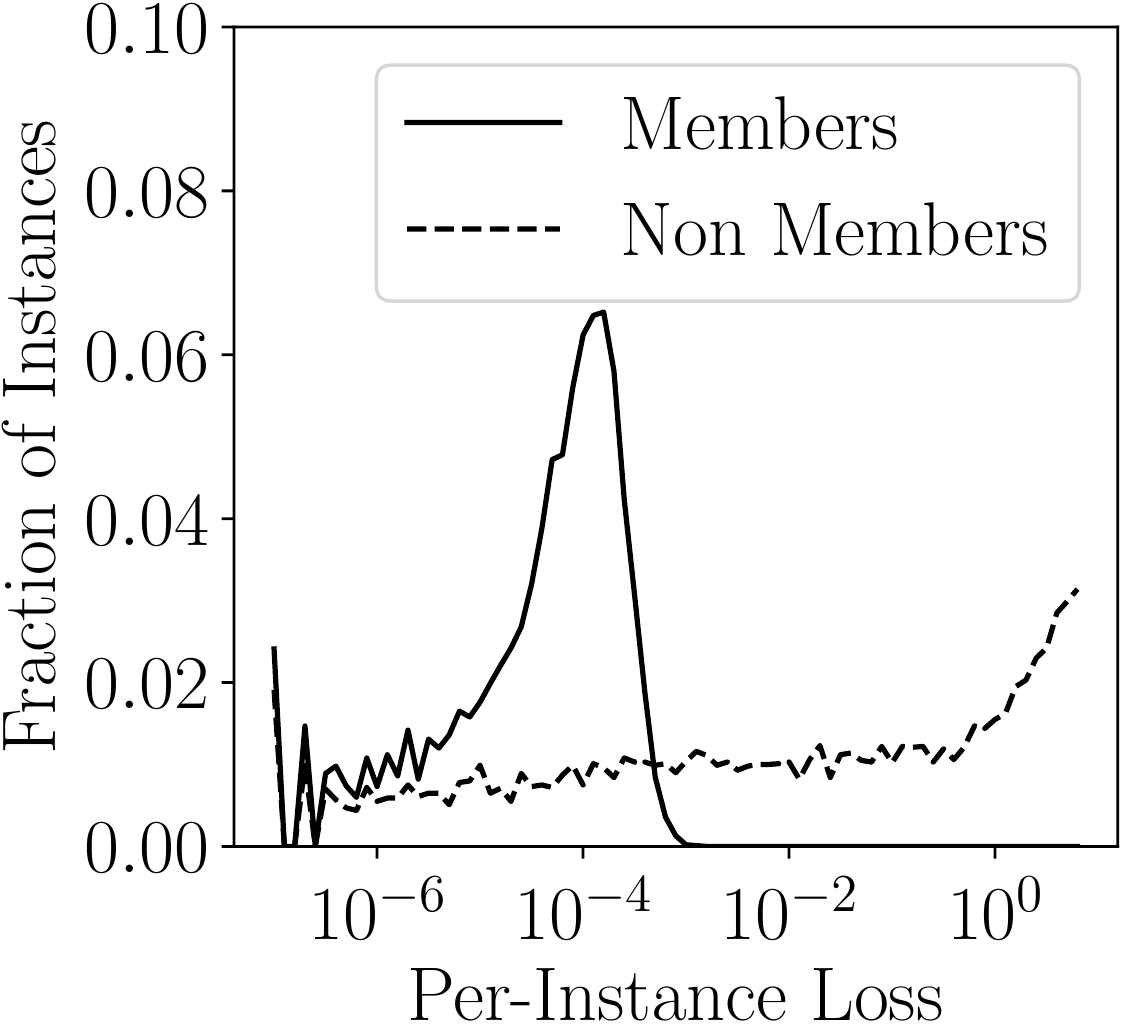}
    \caption{Loss Distribution.}
    \label{fig:purchase_100_mi_1_analysis_a}
    \end{subfigure} \qquad \qquad
    \begin{subfigure}[b]{0.35\textwidth}
    \centering
    \includegraphics[width=.95\linewidth]{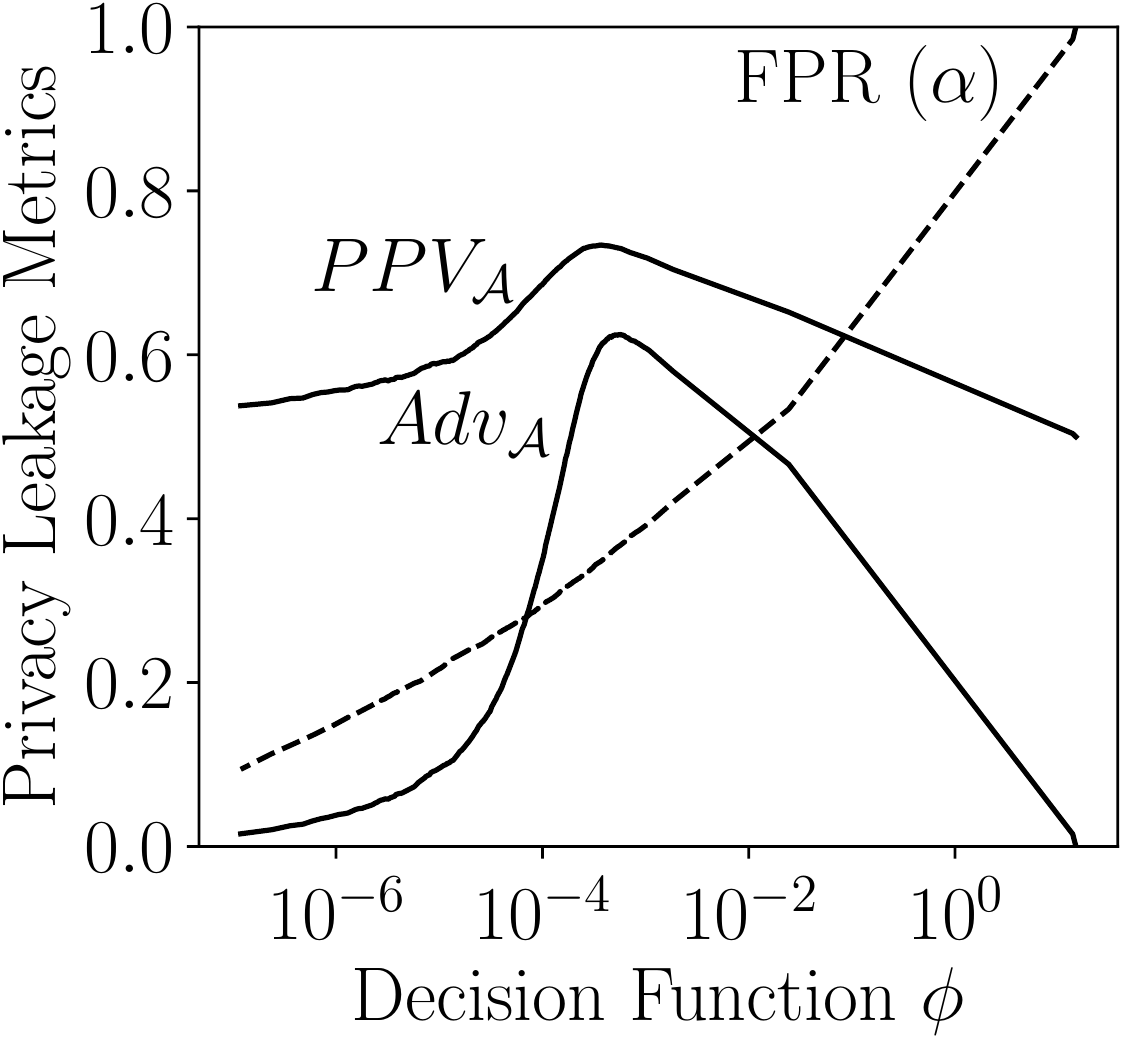}
    \caption{\bfattack{Yeom} Performance.}
    \label{fig:purchase_100_mi_1_analysis_b}
    \end{subfigure}
    \caption{Analysis of \bfattack{Yeom} on non-private model trained on \bfdataset{Purchase-100X} with balanced prior. {\rm The x-axis shows the per-instance loss on a logarithmic scale from $10^{-7}$ to $10^1$ where the buckets are in the range $(10^{-7}, 10^{-6.9})$, $(10^{-6.9}, 10^{-6.8})$, and so on up to $(10^{0.9}, 10^{1})$. }}
    \label{fig:purchase_100_mi_1_analysis}
\end{figure}
\begin{table*}[tb]
    \centering
    \begin{tabular}{llS[table-format=2.2]cS[table-format=2.1,separate-uncertainty,table-figures-uncertainty=1]S[table-format=2.1,separate-uncertainty,table-figures-uncertainty=3]}
        & & {$\alpha$} & $\phi$ & {$Adv_\cA$} & {$PPV_\cA$} \\ \hline
        \multirow{5}{*}{\parbox{1.5cm}{\bfattack{Yeom}}} & Fixed FPR & 1.00 & - & {-} & {-} \\
        & Min FPR & 10.00 & \num{0} & 1.0 \pm 0.8 & 32.5 \pm 26.5 \\
        & Fixed $\phi$ & {-} & \num{1.0 \pm 0.0 e-4} & 33.8 \pm 0.3 & 68.1 \pm 0.2 \\
        & Max $PPV_\cA$ & 35.00 & \num{3.7 \pm 0.3 e-4} & 60.2 \pm 0.8 & 73.0 \pm 0.2 \\
        & Max $Adv_\cA$ & 37.00 & \num{6.0 \pm 0.4 e-4} & 61.9 \pm 0.3 & 72.7 \pm 0.2 \\ \hline
        \multirow{5}{*}{\parbox{1.5cm}{\bfattack{Yeom CBT}}} & Min FPR & 0.01 & \rm \num{0}, \num{0}, \num{6.7e-6} & 0.2 \pm 0.1 & 73.4 \pm 5.0 \\
        & Max $PPV_\cA$ & 0.01 & \rm \num{0}, \num{0}, \num{6.7e-6} & 0.2 \pm 0.1 & 73.4 \pm 5.0 \\
        & Fixed FPR & 1.00 & \rm \num{0}, \num{0}, \num{6.7e-6} & 0.2 \pm 0.1 & 73.4 \pm 5.0 \\
        & Fixed $\phi$ & {-} & \rm(\num{0.2}, \num{0.9}, \num{2.4})$\times$\num{e-4} & 32.4 \pm 1.6 & 67.6 \pm 0.5 \\
        & Max $Adv_\cA$ & 55.00 & \rm(\num{1.1}, \num{4.6}, \num{18.7})$\times$\num{e-4} & 61.6 \pm 0.5 & 72.7 \pm 0.2 \\ \hline
        \multirow{5}{*}{\parbox{1.5cm}{\bfattack{Shokri}}} & Min FPR & 0.02 & \num{1.06 \pm 0.38} & 0.0 \pm 0.1 & 33.2 \pm 31.7 \\
        & Fixed FPR & 1.00 & \num{0.80 \pm 0.02} & 1.8 \pm 0.5 & 71.9 \pm 4.2 \\
        & Max $PPV_\cA$ & 1.92 & \num{0.78 \pm 0.01} & 3.8 \pm 0.5 & 73.4 \pm 1.6 \\
        & Fixed $\phi$ & {-} & \num{0.50 \pm 0.00} & 50.6 \pm 0.5 & 67.1 \pm 0.3 \\
        & Max $Adv_\cA$ & 47.30 & \num{0.50 \pm 0.04} & 50.6 \pm 0.7 & 67.1 \pm 0.2 \\ \hline
        \multirow{5}{*}{\parbox{1.6cm}{\bfattack{Shokri CBT}}} & Fixed FPR & 1.00 & {-} & {-} & {-} \\
        & Min FPR & 2.00 & \num{0.5}, \num{0.8}, \num{1.8} & 0.1 \pm 0.1 & 53.8 \pm 5.0 \\
        & Max $PPV_\cA$ & 40.00 & \num{0.4}, \num{0.7}, \num{1.1} & 57.5 \pm 0.4 & 72.0 \pm 0.2 \\
        & Max $Adv_\cA$ & 50.00 & \num{0}, \num{0.7}, \num{1.1} & 59.6 \pm 0.4 & 71.7 \pm 0.1 \\ 
        & Fixed $\phi$ & {-} & \num{0.5}, \num{0.5}, \num{0.5} & 50.6 \pm 0.5 & 67.1 \pm 0.3 \\
        \hline
        \multirow{4}{*}{\parbox{1.5cm}{\bfattack{Merlin}}} & Min FPR & 0.01 & \num{0.88 \pm 0.01} & 0.1 \pm 0.0 & 93.4 \pm 6.3 \\
        & Max $PPV_\cA$ & 0.01 & \num{0.88 \pm 0.01} & 0.1 \pm 0.0 & 93.4 \pm 6.3 \\
        & Fixed FPR & 1.00 & \num{0.78 \pm 0.00} & 3.4 \pm 0.2 & 85.0 \pm 2.0 \\
        & Max $Adv_\cA$ & 31.00 & \num{0.60 \pm 0.00} & 20.6 \pm 0.2 & 62.6 \pm 0.2 \\ \hline
        \bfattack{Morgan} & Max $PPV_\cA$ & {-} & \rm \num{3.4e-5}, \num{6.0e-4}, \num{0.88} & 0.1 \pm 0.0 & 98.0 \pm 4.0 \\
    \end{tabular}
    \caption{Thresholds selected against non-private models trained on \bfdataset{Purchase-100X} with balanced prior. \rm 
    The results are averaged over five runs such that the target model is trained from the scratch for each run. \bfattack{Yeom CBT} and \bfattack{Shokri CBT} use class-based thresholds, where $\phi$ shows the triplet of minimum, median and maximum thresholds across all classes. All values, except $\phi$, are in percentage.}
    \label{tab:purchase_100_loss_selection_results_uniform_prior}
\end{table*}

\section{Empirical Results}\label{sec:threshold_selection_results}

In this section, we evaluate our threshold selection procedure (Procedure~\ref{algo:inference_training}) across the four inference attacks. We first consider the \bfattack{Yeom} attack, and show that our threshold selection procedure can be used to obtain thresholds that achieve particular attacker goals, such as maximizing the PPV or membership advantage metric, or minimizing the false positive rate. Next, we use our threshold selection procedure on the \bfattack{Shokri} attack and discuss the results in Section~\ref{sec:evalshokri}. In Section~\ref{sec:evalmerlin} we evaluate the \bfattack{Merlin} attack using the same threshold selection procedure, and find that it achieves higher PPV metric compared to both \bfattack{Yeom} and \bfattack{Shokri}. Then, Section~\ref{sec:evalmorgan} shows how the \bfattack{Morgan} attack achieves higher PPV by combining aspects of both \bfattack{Yeom} and \bfattack{Merlin}. Results in the first four subsections focus on non-private models and balanced prior scenarios. In Section~\ref{sec:threshprivacy} we evaluate the attacks on differentially private models. Section~\ref{sec:imbalanced_prior_case} presents results for scenarios with imbalanced priors. The results show that non-private models are vulnerable to our proposed attacks, especially \bfattack{Morgan}, even in the skewed prior settings. Private models are vulnerable in the balanced prior setting if the privacy loss budget is set beyond theoretical guarantees.

\subsection{Yeom Attack}\label{sec:evalyeom}
The \bfattack{Yeom} attack uses a fixed threshold on per-instance loss for its membership inference test. A query record is classified as a member if its per-instance loss is less than the selected threshold. We show that the adversary can achieve better privacy leakage, specific to particular attack goals, by using our threshold selection procedure.

\begin{table*}[tb]
    \centering
    \small
    \begin{tabular}{llS[table-format=2.2]S[table-format=2.1,separate-uncertainty,table-figures-uncertainty=1]S[table-format=2.1,separate-uncertainty,table-figures-uncertainty=1]S[table-format=2.2]S[table-format=2.1,separate-uncertainty,table-figures-uncertainty=1]S[table-format=2.1,separate-uncertainty,table-figures-uncertainty=1]}
        & & \multicolumn{3}{c}{\bfdataset{Texas-100}} & \multicolumn{3}{c}{\bfdataset{RCV1X}} \\
        & & {$\alpha$} & {$Adv_\cA$} & {$PPV_\cA$} & {$\alpha$} & {$Adv_\cA$} & {$PPV_\cA$} \\ \hline
        \multirow{5}{*}{\parbox{1cm}{\bfattack{Yeom}}} & Fixed FPR & 1.00 & {-} & {-} & 1.00 & {-} & {-} \\
        & Min FPR & 3.00 & 0.4 \pm 0.9 & 12.0 \pm 24.0 & 33.00 & 0.4 \pm 0.8 & 10.3 \pm 20.5 \\
        & Fixed $\phi$ & {-} & 51.3 \pm 2.6 & 75.0 \pm 1.6 & {-} & 26.9 \pm 2.7 & 58.0 \pm 0.8 \\
        & Max $PPV_\cA$ & 26.00 & 59.2 \pm 11.7 & 76.1 \pm 1.6 & 67.00 & 24.8 \pm 5.0 & 57.9 \pm 1.0 \\
        & Max $Adv_\cA$ & 31.00 & 62.9 \pm 7.7 & 75.0 \pm 0.6 & 70.00 & 25.1 \pm 3.2 & 57.7 \pm 0.6 \\ \hline
        \multirow{5}{*}{\parbox{0.8cm}{\bfattack{Shokri}}} & Min FPR & 0.01 & 1.0 \pm 0.5 & 72.6 \pm 8.6 & 0.01 & 0.6 \pm 0.2 &  91.7 \pm 4.2 \\
        & Max $PPV_\cA$ & 0.70 & 13.8 \pm 1.1 & 89.4 \pm 1.5 & 0.01 & 0.6 \pm 0.2 &  91.7 \pm 4.2 \\
        & Fixed FPR & 1.00 & 16.0 \pm 1.3 & 88.9 \pm 1.7 & 1.00 &  4.6 \pm 0.6& 84.5 \pm 1.8 \\
        & Fixed $\phi$ & {-} & 64.0 \pm 1.4 & 74.1 \pm 1.3 & {-} & 24.0 \pm 0.8 & 57.3 \pm 0.4 \\
        & Max $Adv_\cA$ & 31.00 & 64.1 \pm 1.2 & 74.7 \pm 0.9 &  75.00 & 24.2\pm 0.5 & 58.0\pm 0.4 \\ \hline
        \multirow{4}{*}{\parbox{1cm}{\bfattack{Merlin}}} & Min FPR & 0.01 & 0.1 \pm 0.1 & 51.9 \pm 42.4 & 0.01 & 0.2 \pm 0.0 & 98.8 \pm 2.4 \\
        & Max $PPV_\cA$ & 0.06 & 0.3 \pm 0.2 & 92.0 \pm 4.5 & 0.01 & 0.2 \pm 0.0 & 98.8 \pm 2.4 \\
        & Fixed FPR & 1.00 & 4.9 \pm 1.3 & 87.8 \pm 2.7 & 1.00 & 2.6 \pm 0.7 & 81.7 \pm 4.3 \\
        & Max $Adv_\cA$ & 36.00 & 37.8 \pm 1.5 & 68.0 \pm 0.8 & 26.00 & 11.6 \pm 2.3 & 59.5 \pm 2.0 \\ \hline
        \bfattack{Morgan} & Max $PPV_\cA$ & {-} & 0.5 \pm 0.2 & 95.7 \pm 4.6 & {-} & 0.4 \pm 0.3 & 100.0 \pm 0.0 \\
    \end{tabular}
    \caption{Comparing attacks on non-private models trained on \bfdataset{Texas-100} and \bfdataset{RCV1X} data sets for balanced prior. {\rm All values are percentages ($\alpha=0.01$ means 1 out of 10,000).}} 
    \label{tab:leakage_comparison_balanced_prior}
\end{table*}

\shortsection{Results on \bfdataset{Purchase-100X}}
Figure~\ref{fig:purchase_100_mi_1_analysis_a} shows the distribution of per-instance loss of members and non-members for a non-private model trained on \bfdataset{Purchase-100X}. Per-instance losses of members are concentrated close to zero, and most of the loss values are less than 0.001. Whereas for non-members, the loss values are spread across the range. This suggests that a larger fraction of members will be identified by the attacker with high precision (PPV) for loss thresholds less than 0.001, and hence the privacy leakage will be high. 

Another notable observation is that out of the 10,000 test records there are $959.2 \pm 23.5$ non-members (average across five runs) with zero loss, and hence the minimum achievable false positive rate is around 10\%. This is reflected in Figure~\ref{fig:purchase_100_mi_1_analysis_b}, which shows the effect of selecting different loss thresholds on the privacy leakage metrics. An attacker can use our threshold selection procedure to choose a loss threshold to meet specific attack goals, such as minimizing the false positive rate (Min FPR), or achieving a fixed false positive rate (Fixed FPR), or maximizing either of the privacy leakage metrics (Max $PPV_\cA$ and Max $Adv_\cA$). Table~\ref{tab:purchase_100_loss_selection_results_uniform_prior} summarizes these scenarios and compares their thresholds with the threshold selected by the method of Yeom et al.\ (Fixed $\phi$). For Fixed FPR, we consider an attacker with a false positive rate of 1\% ($\alpha = 1\%$). 

The attacker uses Procedure~\ref{algo:inference_training} to find the loss threshold, $\phi$, corresponding to $\alpha = 1\%$, which it uses for membership inference on the target set. However, since the minimum achievable false positive rate for \bfattack{Yeom} on \bfdataset{Purchase-100X} is 10\%, this attack fails to find a suitable threshold. For maximizing PPV or advantage, the attacker can use the threshold selection procedure with varying $\alpha$ values and choose the threshold $\phi$ that maximizes the required privacy metric. In comparison, Fixed~$\phi$ uses expected training loss as threshold which does not necessarily maximize the privacy leakage. As the results in the table demonstrate, an attacker can accomplish different attack goals, and achieve increased privacy leakage, using the \bfattack{Yeom} attack with thresholds chosen using our threshold selection procedure.

\shortsection{Results on Other Data Sets}
Table~\ref{tab:leakage_comparison_balanced_prior} compares the performance of \bfattack{Yeom} against non-private models across 
the \bfdataset{Texas-100} and \bfdataset{RCV1X} data sets. We observe similar trends of privacy leakage corresponding to the selected thresholds for these data sets as we did for \bfdataset{Purchase-100X} so present most of the results for these data sets in Appendix~\ref{appendix:non_private_results}, and only discuss some notable differences here. Results for \bfdataset{CIFAR-100} can be found in Appendix~\ref{appendix:non_private_results}.

For \bfdataset{Texas-100}, \bfattack{Yeom} can achieve false positive rates as low as 3\%. The attack performance on this data set is comparable to that of \bfdataset{Purchase-100X}. For \bfdataset{RCV1X}, the attack success rate is substantially lower than that for the other data sets. This is because, unlike the other data sets which have 100 classes, \bfdataset{RCV1X} is a 52-class classification task. As reported in prior works~\citep{membershipinference, yeom2018privacy}, success of membership inference attack is proportional to the complexity of classification task. We further note that the maximum PPV that can be achieved by \bfattack{Yeom} on \bfdataset{RCV1X} is only around 58\%, at which point the membership advantage is close to 27\%. This gives credence to our claim that membership advantage should not be solely relied on as a measure of inference risk. While membership advantage can be high, the privacy leakage is negligible for balanced priors when the PPV is close to 50\%. Later in Section~\ref{sec:imbalanced_prior_case} we show that this phenomenon is prevalent across all data sets when the prior is imbalanced. 

\bfattack{Yeom}'s performance on \bfdataset{CIFAR-100} is similar to that on \bfdataset{Purchase-100X} and \bfdataset{Texas-100} data sets. Since the model does not completely overfit on \bfdataset{CIFAR-100}, the distribution of loss values for both members and non-members are not far apart, and as a consequence \bfattack{Yeom} is able to achieve much lower false positive rates. 

\shortsection{Using Class-Based Thresholds}
Recently, \cite{song2020systematic} demonstrated that the approach of \cite{yeom2018privacy} can be further improved by using class-based thresholds instead of one global threshold on loss values. We implement this approach, using our threshold setting algorithm to independently set the threshold for each class (referred as \bfattack{Yeom CBT}). This enables finding class-based thresholds corresponding to smaller $\alpha$ values, as seen for the minimum FPR ($\alpha = 0.01$) and fixed FPR ($\alpha = 1$) cases for \bfdataset{Purchase-100X} in Table~\ref{tab:purchase_100_loss_selection_results_uniform_prior}. Nonetheless, the maximum PPV still does not increase much beyond \bfattack{Yeom} on \bfdataset{Purchase-100X}, with the largest increase being from 73.0\% to 73.4\%. For other data sets, though, this technique improves the maximum PPV of \bfattack{Yeom}. For \bfdataset{Texas-100}, the PPV increases from 76\% to 92\%, for \bfattack{RCV1X}, the PPV increases from 58\% to 93\% and for \bfdataset{CIFAR-100}, the PPV increases from 73\% to 81\% (see Appendix~\ref{appendix:non_private_results}). However, the maximum PPV never exceeds beyond \bfattack{Merlin} or \bfattack{Morgan}. While \cite{song2020systematic} also showed the application of their class-based thresholds on other metrics such as model confidence and modified entropy, their experimental results show that these approaches achieve similar attack performance to the CBT on per-instance loss metric. Hence, we do not include the CBT results for other metrics.

\subsection{Shokri Attack}\label{sec:evalshokri}
The \bfattack{Shokri} attack~\citep{shokri2017membership} requires training multiple shadow models on hold-out data sets similar to the target model. These shadow models are used to train an inference model that outputs a confidence value between 0 and 1 for membership inference, where 1 indicates member.  We use the experimental setting of \cite{jayaraman2019evaluating} to train five shadow models with the same architecture and hyperparameter settings of the target model. The inference model is a two-layer neural network with 64 neurons in each hidden layer. As with the \bfattack{Yeom} attack, our threshold selection procedure can be used to increase privacy leakage for \bfattack{Shokri}.

\shortsection{Results on \bfdataset{Purchase-100X}}
Table~\ref{tab:purchase_100_loss_selection_results_uniform_prior} shows the privacy leakage of \bfattack{Shokri} for different attack goals. The original attack of Shokri et al.\ (Fixed $\phi$) uses a threshold of 0.5 on the inference model confidence and achieves close to 50\% membership advantage, but has a PPV of around 67\%. Using our threshold setting procedure to maximize PPV, \bfattack{Shokri} achieves PPV of over 73\%, which is comparable to the \bfattack{Yeom} attack.

\shortsection{Results on Other Data Sets}
Table~\ref{tab:leakage_comparison_balanced_prior} shows the results of \bfattack{Shokri} across multiple data sets. The \bfattack{Shokri} attack performance varies considerably across different data sets when compared to the \bfattack{Yeom} attack. While \bfattack{Shokri} achieves higher PPV than \bfattack{Yeom} on \bfdataset{Texas-100} and \bfdataset{RCV1X}, reflecting significant privacy risk on these data sets, \bfattack{Yeom} outperforms \bfattack{Shokri} on \bfdataset{CIFAR-100}. However, \bfattack{Merlin} and \bfattack{Morgan} consistently achieve higher PPV than both \bfattack{Yeom} and \bfattack{Shokri} (see Sections~\ref{sec:evalmerlin} and~\ref{sec:evalmorgan}).

\shortsection{Using Class-Based Thresholds}
We also use class-based thresholds for \bfattack{Shokri} attack and include the results for \bfdataset{Purchase-100X} in Table~\ref{tab:purchase_100_loss_selection_results_uniform_prior} (called \bfattack{Shokri CBT}). However, we do not observe any significant improvement in privacy leakage over the \bfattack{Shokri} attack. While the maximum membership advantage increases from 50\% to around 60\%, the maximum PPV is still close to 72\%. We observe similar behaviour across other data sets.

\begin{figure}[tb]
    \centering
    \begin{subfigure}[b]{0.35\textwidth}
    \centering
    \includegraphics[width=.95\linewidth]{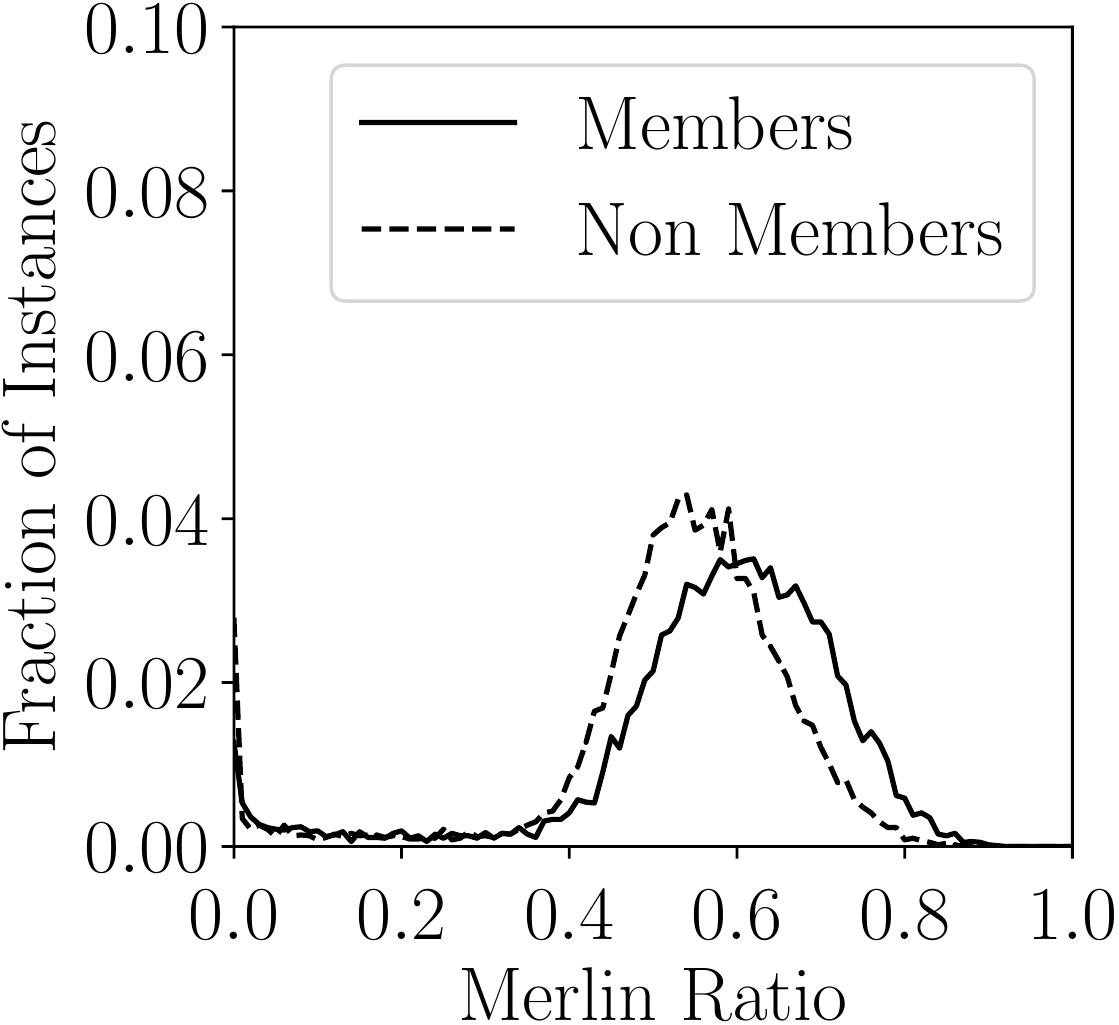}
    \caption{\bfattack{Merlin} Ratio Distribution.}
    \label{fig:purchase_100_mi_2_analysis_a}
    \end{subfigure} \qquad \qquad
    \begin{subfigure}[b]{0.35\textwidth}
    \centering
    \includegraphics[width=.95\linewidth]{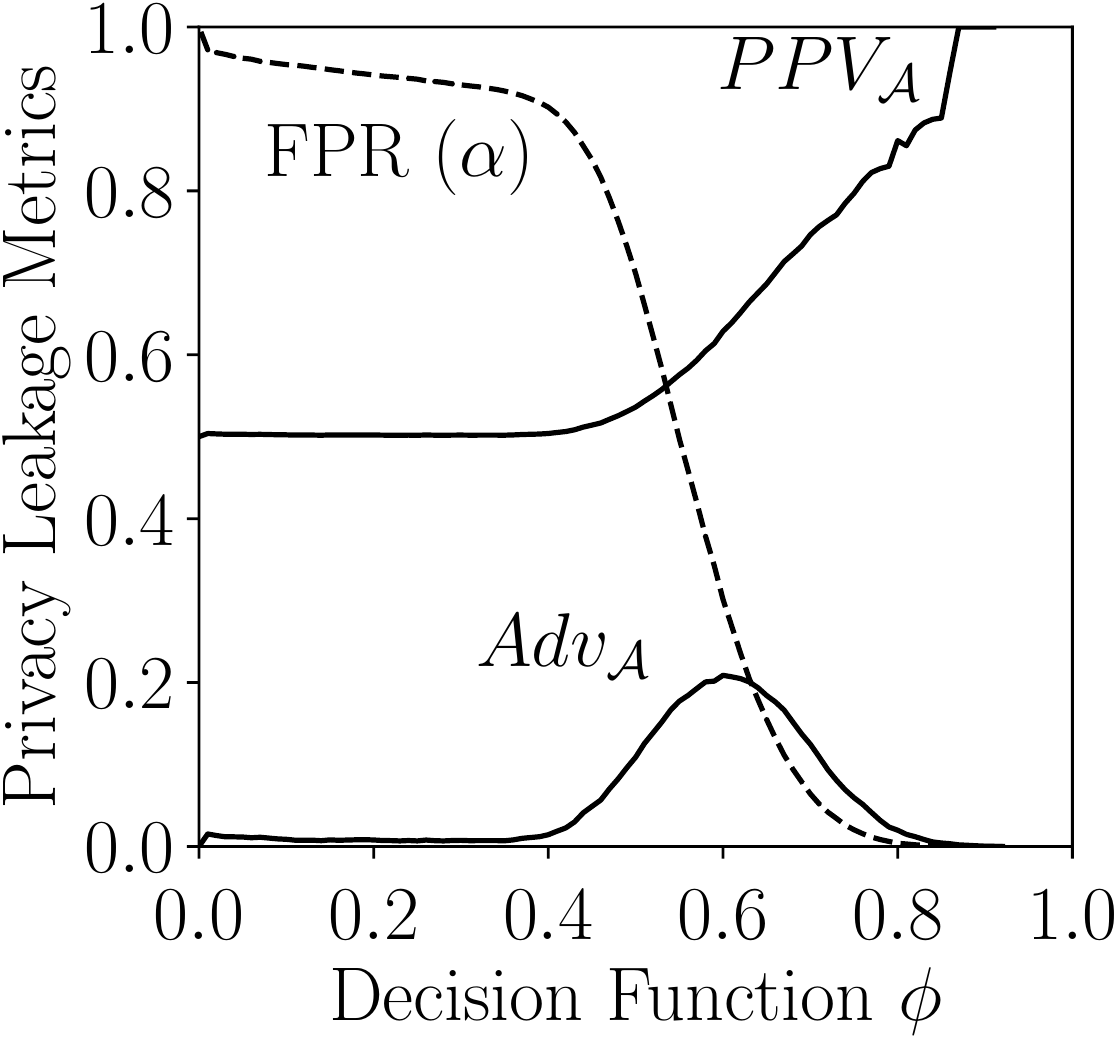}
    \caption{\bfattack{Merlin} Performance.}
    \label{fig:purchase_100_mi_2_analysis_b}
    \end{subfigure}
    \caption{Analysis of \bfattack{Merlin} on non-private model trained on \bfdataset{Purchase-100X} with balanced prior.}
    \label{fig:purchase_100_mi_2_analysis}
\end{figure}

\subsection{Merlin Attack}\label{sec:evalmerlin}
Next, we perform inference attacks using the \bfattack{Merlin} (Algorithm \ref{algo:proposed_mi_2}) where the attacker perturbs a record with random Gaussian noise of magnitude $\sigma = 0.01$ and notes the direction of change in loss. This process is repeated $T = 100$ times and the attacker counts the number of times the loss increases out of $T$ trials to find the \bfattack{Merlin} ratio, $count / T$. If the \bfattack{Merlin} ratio exceeds a threshold, then the record is classified as a member. As with the \bfattack{Yeom} and \bfattack{Shokri} experiments, we use Procedure~\ref{algo:inference_training} to select a suitable threshold. 

\shortsection{Results on \bfdataset{Purchase-100X}}
Figure~\ref{fig:purchase_100_mi_2_analysis_a} shows the distribution of \bfattack{Merlin} ratio for member and non-member records for a non-private model trained on the \bfdataset{Purchase-100X} data set. The average \bfattack{Merlin} ratio is $0.57 \pm 0.17$ for member records, whereas for the non-member records it is $0.52 \pm 0.16$. A peculiar observation is that the \bfattack{Merlin} ratio is zero for a considerable fraction of members and non-members. For these non-member records, the loss is very high to begin with and hence it never increases for the nearby noise points. Whereas for the member records, the loss value does not change even with addition of noise.  As mentioned in step 5 of Algorithm~\ref{algo:proposed_mi_2}, we only check if the loss increases upon perturbation since we believe that equality is not a strong indicator of membership. Hence these outliers indicate regions where the loss doesn't change, not points where it always decreases.

Figure~\ref{fig:purchase_100_mi_2_analysis_b} shows the attack performance with varying thresholds. \bfattack{Merlin} can achieve much higher PPV than \bfattack{Yeom} and \bfattack{Shokri}.  Table~\ref{tab:purchase_100_loss_selection_results_uniform_prior} 
summarizes the thresholds selected by \bfattack{Merlin} with different attack goals and compares the performance with \bfattack{Yeom} and \bfattack{Shokri}. While \bfattack{Yeom} can only achieve a minimum false positive rate of 10\% on this data set, \bfattack{Merlin} can achieve false positive rate as low as 0.01\%. Thus \bfattack{Merlin} is successful at a fixed false positive rate of 1\% where \bfattack{Yeom} fails. Another notable observation is that \bfattack{Merlin} can achieve close to 93\% PPV, while the maximum possible PPV achievable via \bfattack{Yeom} and \bfattack{Shokri} (including their CBT versions) is under 74\%. Thus, this attack is more suitable for scenarios where attack precision is preferred.

\begin{table*}[tb]
    \centering
    \small
    \begin{tabular}{cS[table-format=3.0]S[table-format=2.2]S[table-format=2.2,separate-uncertainty,table-figures-uncertainty=1]S[table-format=2.1,separate-uncertainty,table-figures-uncertainty=1]}
        & {$\epsilon$} & {$\alpha$} & {$\phi$} & {Max $PPV_\cA$} \\ \hline
        \multirow{3}{*}{\bfattack{Yeom}} & 1 & 0.03 & \num{9.6 \pm 1.1 e-2} & 71.4 \pm 25.7 \\
        & 10 & 0.90 & \num{3.7 \pm 1.5 e-5} & 59.4 \pm 2.4 \\
        & 100 & 24.00 & \num{1.1 \pm 0.2 e-2} & 60.5 \pm 0.2 \\ \hline
        \multirow{3}{*}{\bfattack{Shokri}} & 1 & 1.00 & \num{0.80 \pm 0.2} & 50.2 \pm 9.5 \\
        & 10 & 30.00 & \num{0.64 \pm 0.01} & 60.1 \pm 1.6 \\
        & 100 & 5.00 & \num{0.74 \pm 0.01} & 62.4 \pm 1.1 \\ \hline
        \multirow{3}{*}{\bfattack{Merlin}} & 1 & 0.12 & \num{0.87 \pm 0.00} & 67.2 \pm 12.8 \\
        & 10 & 0.02 & \num{0.88 \pm 0.01} & 79.7 \pm 17.9 \\
        & 100 & 0.03 & \num{0.88 \pm 0.01} & 80.3 \pm 24.8 \\ \hline
        \multirow{3}{*}{\bfattack{Morgan}} & 1 & {-} & \rm \num{2.1}, \num{4.3}, \num{0.87} & 71.4 \pm 17.2 \\
        & 10 & {-} & \rm \num{4.5e-5}, \num{1.1e-2}, 
        \num{0.88} & 95.0 \pm 10.0 \\
        & 100 & {-} & \rm \num{1.8e-4}, \num{6.6e-3}, 
        \num{0.87} & 93.3 \pm 13.3 \\
    \end{tabular}
    \caption{Attacks against private models (\bfdataset{Purchase-100X}, balanced prior). \rm $\alpha$ and PPV values are percentages.}
    \label{tab:purchase_100_gamma_1_mi_dp}
\end{table*}

\shortsection{Results on Other Data Sets}
Table~\ref{tab:leakage_comparison_balanced_prior} compares the membership inference attack performance against non-private models across the other data sets. The \bfattack{Merlin} attack consistently achieves higher PPV than \bfattack{Yeom} and \bfattack{Shokri} across all the data sets. \bfattack{Merlin} is more successful on \bfdataset{Texas-100} compared to \bfdataset{Purchase-100X}, as the gap between \bfattack{Merlin} ratio distribution of member records and non-member records is high for \bfdataset{Texas-100} (see Appendix~\ref{appendix:non_private_results} for more analysis). More surprisingly, while \bfattack{Yeom} is less successful on \bfdataset{RCV1X}, we find that \bfattack{Merlin} still manages to achieve a very high PPV that even exceeds the PPV of \bfattack{Shokri} (see Table~\ref{tab:leakage_comparison_balanced_prior}). Thus, \bfattack{Merlin} poses a credible privacy threat even in scenarios where \bfattack{Yeom} fails. However, \bfattack{Merlin} does not perform significantly better than \bfattack{Yeom} and \bfattack{Shokri} on \bfdataset{CIFAR-100} since the per-instance loss of members is high on this data set and hence the members are not at local minimum. Appendix~\ref{appendix:non_private_results} provides more details on all these results.

\begin{figure*}[tb]
    \centering
\begin{subfigure}[b]{0.33\textwidth}
    \centering
    \includegraphics[width=\linewidth]{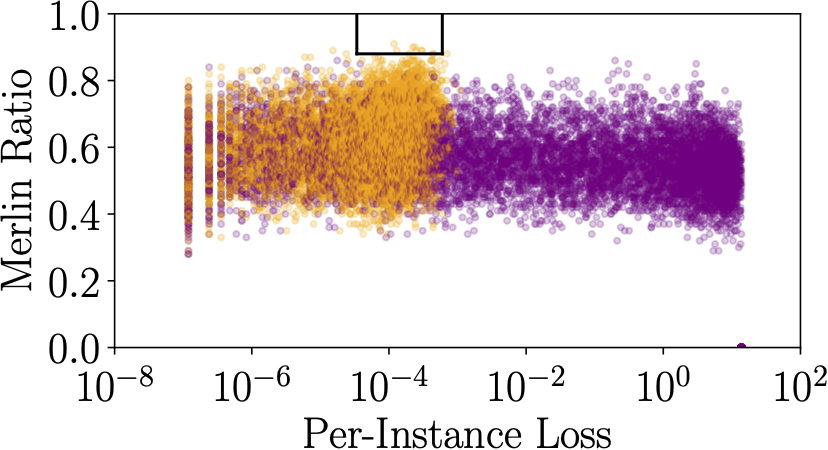}
    \caption{Non-private model at $\gamma = 1$}\label{fig:morgan_purchase}
    \end{subfigure}
    \begin{subfigure}[b]{0.32\textwidth}
    \centering
    \includegraphics[width=\linewidth]{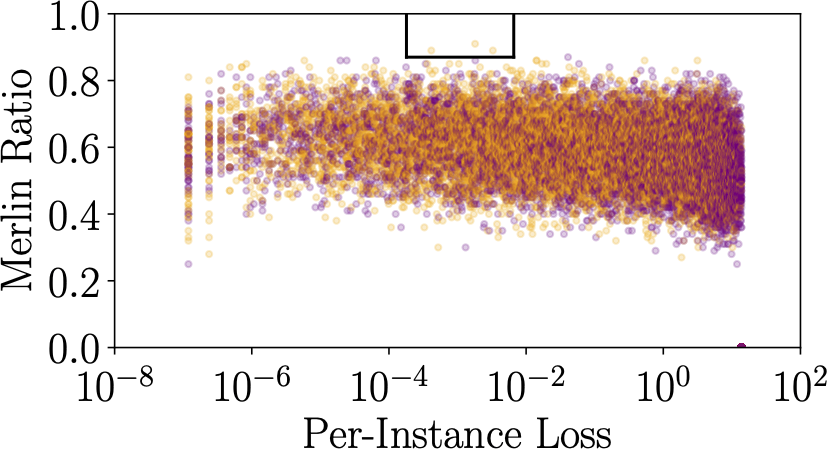}
    \caption{Private model at $\gamma = 1, \epsilon = 100$}\label{fig:morgan_purchase_eps_100}
    \end{subfigure}
    \begin{subfigure}[b]{0.33\textwidth}
    \centering
    \includegraphics[width=\linewidth]{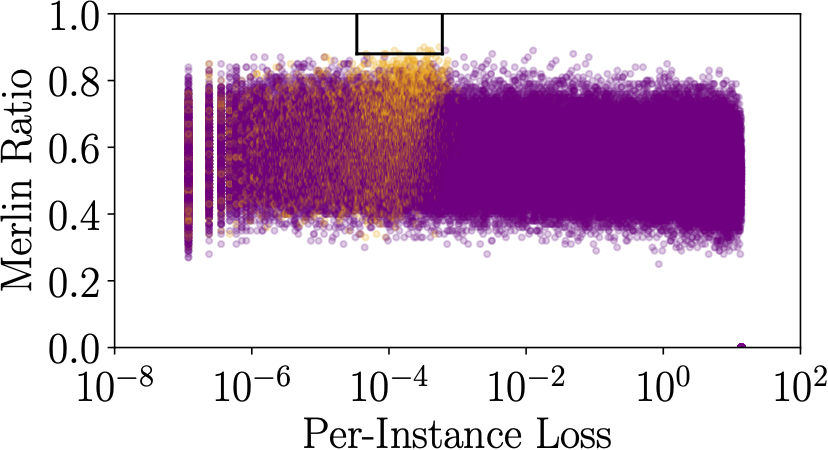}
    \caption{Non-private model at $\gamma = 10$}\label{fig:morgan_purchase_gamma_10}
    \end{subfigure}
    \caption{Comparing loss and \bfattack{Merlin} ratio side-by-side on \bfdataset{Purchase-100X}. \rm Members and non-members are denoted by orange and purple points respectively. The boxes show the thresholds found by the threshold selection process (without access to the training data, but with the same data distribution), and illustrate the regions where members are identified by \bfattack{Morgan} with high confidence.}
    \label{fig:morgan_purchase_plots}
\end{figure*}

\shortsection{Using Class-Based Thresholds}
We also tried class-based thresholds for \bfattack{Merlin}, like we did for \bfattack{Yeom} and \bfattack{Shokri}. However, we found that this approach does not benefit \bfattack{Merlin} as the individual classes do not have enough records to provide meaningful thresholds. Using class-based thresholds for \bfattack{Merlin} increases the advantage metric from 0.1\% to 2.8\%, but decreases the maximum achievable PPV from around 93.4\% to 83.1\%. We observed similar behavior across different thresholds.

\begin{figure*}[ptb]
    \centering
    \begin{subfigure}[b]{0.35\textwidth}
    \centering
    \includegraphics[width=.95\linewidth]{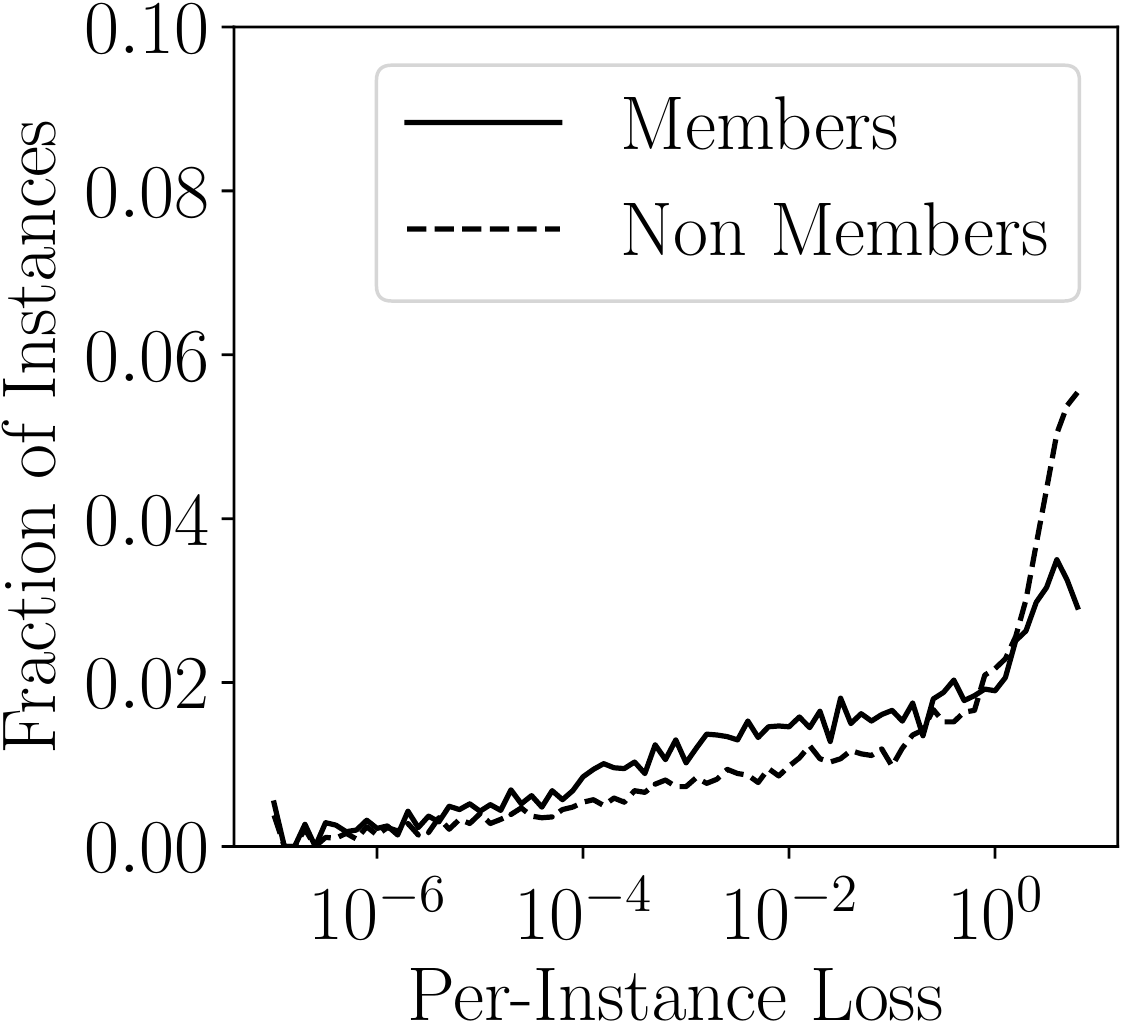}
    \caption{Loss distribution.}
    \label{fig:purchase_100_mi_1_analysis_dp_a}
    \end{subfigure} \qquad \qquad
    \begin{subfigure}[b]{0.35\textwidth}
    \centering
    \includegraphics[width=.95\linewidth]{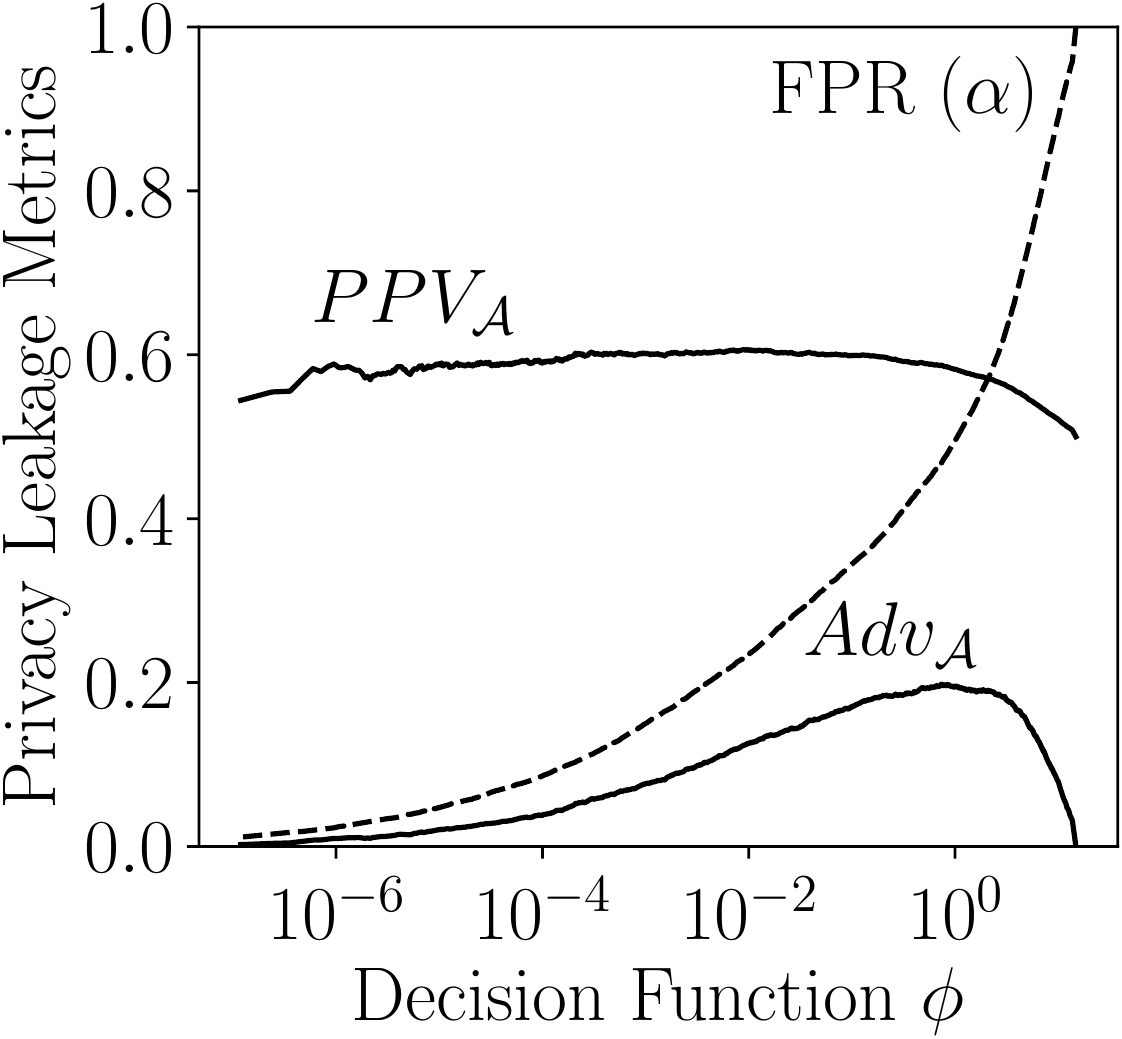}
    \caption{\bfattack{Yeom} performance.}
    \label{fig:purchase_100_mi_1_analysis_dp_b}
    \end{subfigure}
    \begin{subfigure}[b]{0.35\textwidth}
    \centering
    \includegraphics[width=.95\linewidth]{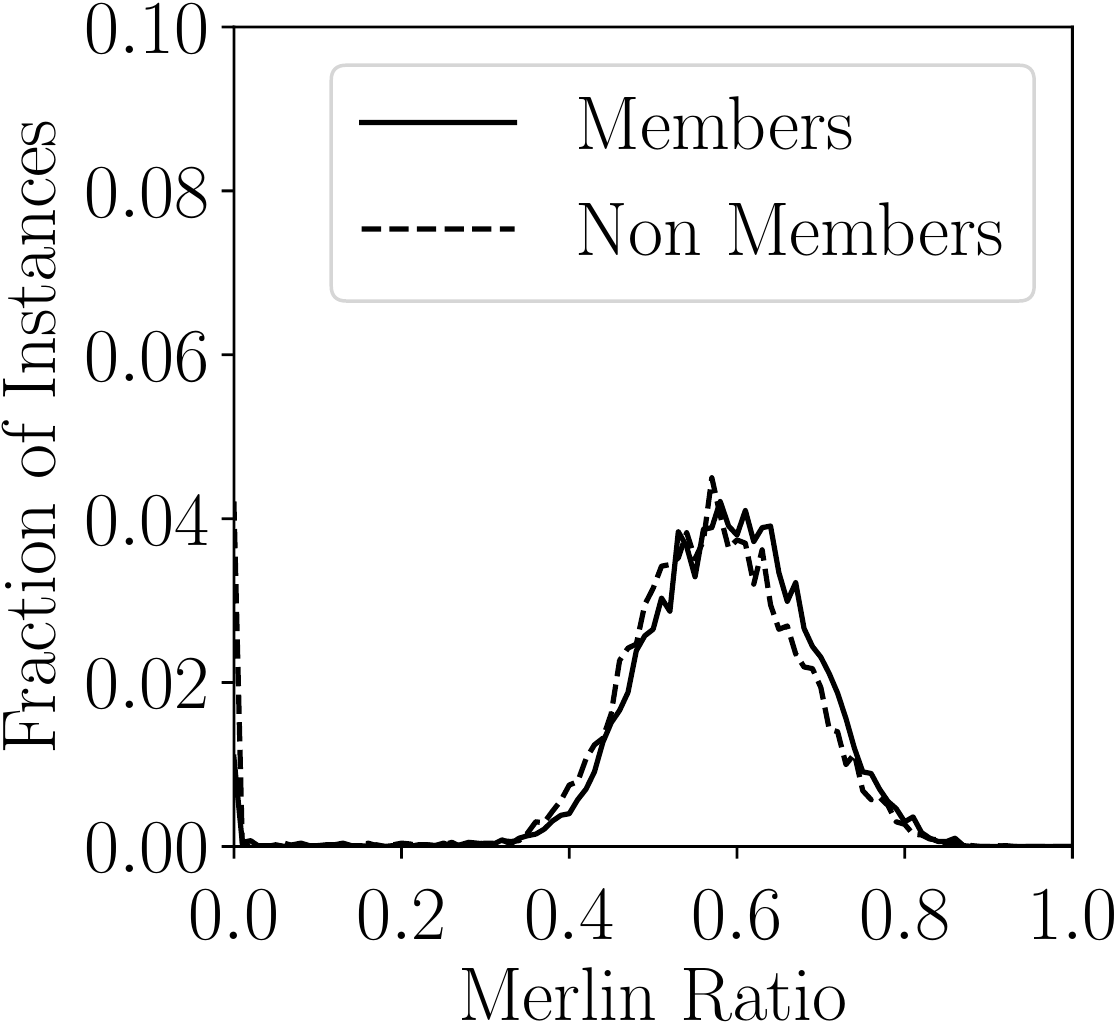}
    \caption{Merlin ratio distribution.}
    \label{fig:purchase_100_mi_2_analysis_dp_a}
    \end{subfigure} \qquad \qquad
    \begin{subfigure}[b]{0.35\textwidth}
    \centering
    \includegraphics[width=.95\linewidth]{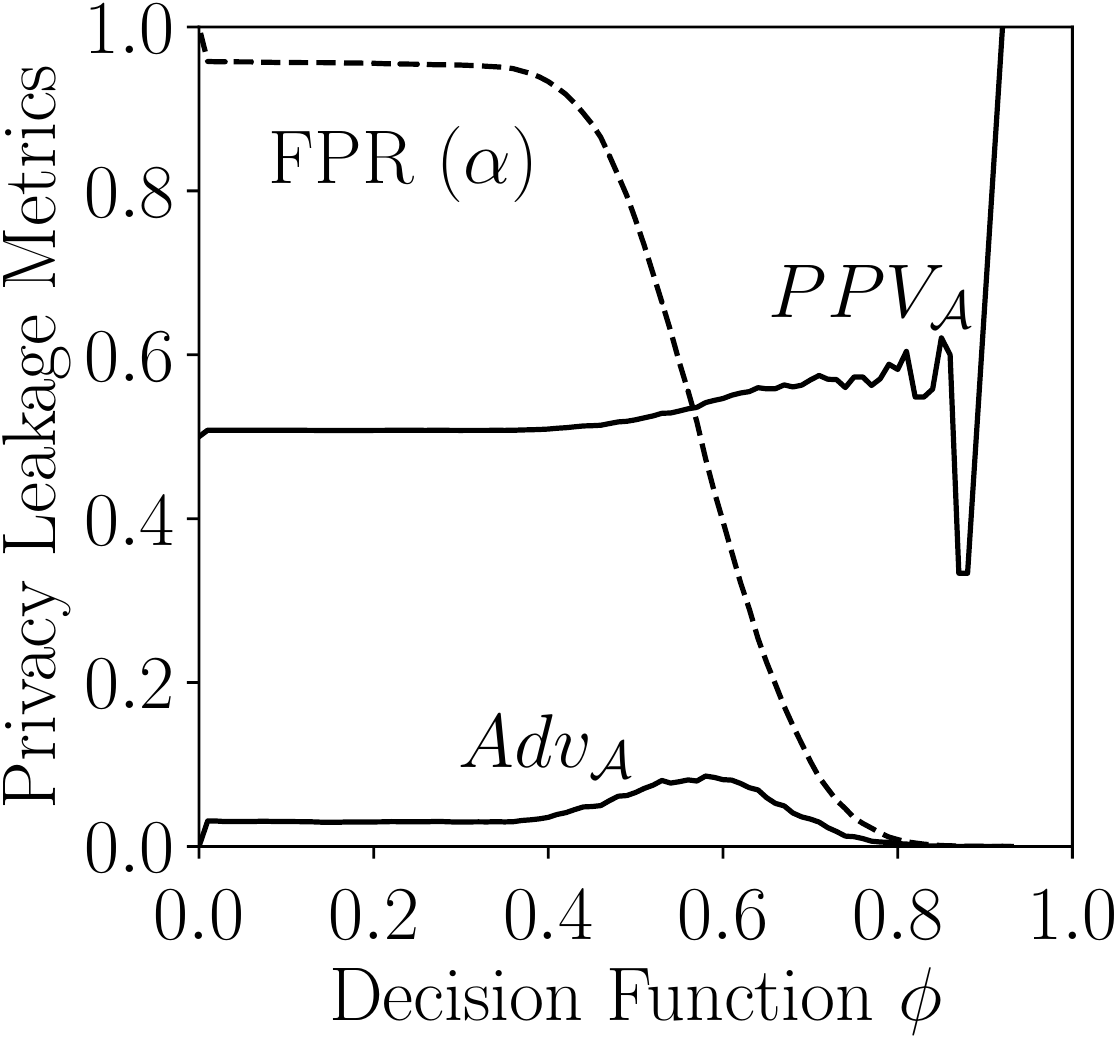}
    \caption{\bfattack{Merlin} performance.}
    \label{fig:purchase_100_mi_2_analysis_dp_b}
    \end{subfigure}
    \caption{Analysis of \bfattack{Yeom} and \bfattack{Merlin} on private model trained with $\epsilon = 100$ at $\gamma = 1$ (\bfdataset{Purchase-100X}).}
    \label{fig:purchase_100_analysis_dp}
\end{figure*}

\subsection{Morgan Attack}\label{sec:evalmorgan}

The \bfattack{Morgan} attack (Section~\ref{sec:morgan}) combines both \bfattack{Yeom} and \bfattack{Merlin} attacks to identify the most vulnerable members. Recall that \bfattack{Morgan} classifies a record as member if its per-instance loss is between $\phi_L$ and $\phi_U$ and if the \bfattack{Merlin} ratio is at least $\phi_M$.

\shortsection{Results on \bfdataset{Purchase-100X}}
Figure~\ref{fig:morgan_purchase} shows the loss and \bfattack{Merlin} ratio for members and non-members for one run of non-private model training in balanced prior. As shown, a fraction of members are clustered between \num{3.4e-5} and \num{6.0e-4} loss and with \bfattack{Merlin} ratio at least 0.88, and in this region there are very few non-members. Thus, \bfattack{Morgan} can target these vulnerable members whereas \bfattack{Yeom} and \bfattack{Merlin} fail to do, being restricted to a single threshold. As reported in Table~\ref{tab:purchase_100_loss_selection_results_uniform_prior}, \bfattack{Morgan} succeeds at achieving around 98\% PPV while \bfattack{Yeom} and \bfattack{Merlin} only achieve 73\% and 93\% PPV respectively at maximum on \bfdataset{Purchase-100X}.

\shortsection{Results on Other Data Sets}
\bfattack{Morgan} exposes members with 100\% PPV in our experiments against non-private models for the \bfdataset{RCV1X} and \bfdataset{CIFAR-100} data sets, and exceeds 95\% PPV for \bfdataset{Texas-100} (see Table~\ref{tab:leakage_comparison_balanced_prior}, and Appendix~\ref{appendix:non_private_results}). \bfattack{Morgan} benefits by using multiple thresholds and is able to identify the most vulnerable members with close to 100\% confidence. Further discussion on these results can be found in Appendix~\ref{appendix:non_private_results}.

\subsection{Impact of Privacy Noise}\label{sec:threshprivacy}

So far, all results we have reported are for inference attacks on models trained without any privacy protections. We also evaluated membership inference attacks against the private models and found the models to be vulnerable to \bfattack{Merlin} and \bfattack{Morgan} at privacy loss budgets high enough to train useful models. Like the experiments with non-private models, here also we repeat the experiments five times and report average results and standard error. In each run, we train a private model from scratch and perform the attack procedure on it. 

Table~\ref{tab:purchase_100_gamma_1_mi_dp} compares the maximum PPV achieved by \bfattack{Yeom}, \bfattack{Shokri}, \bfattack{Merlin} and \bfattack{Morgan} against private models trained on \bfdataset{Purchase-100X} with varying privacy loss budgets. As expected, the privacy leakage increases with the privacy loss budget. \bfattack{Merlin} and \bfattack{Morgan} both achieve high PPV for privacy loss budgets, $\epsilon \ge 10$ (large enough to offer no meaningful privacy guarantee, but this is still smaller than needed to train useful models). \bfattack{Morgan} has higher PPV on average and less deviation than \bfattack{Merlin}.

\shortsection{\bfattack{Yeom} and \bfattack{Shokri} Attacks}
To understand how the privacy noise influences \bfattack{Yeom} attack success, we plot the loss distribution of member and non-member records for a private model trained with $\epsilon = 100$ in Figure~\ref{fig:purchase_100_mi_1_analysis_dp_a}. The figure shows that the noise reduces the gap between the two distributions when compared to Figure~\ref{fig:purchase_100_mi_1_analysis_a} with no privacy. Hence differential privacy limits the success of \bfattack{Yeom} by spreading out the loss values for both member and non-member distributions. This has the counter-productive impact of reducing the number of non-member records with zero loss from $959.2 \pm 23.5$ (in non-private case) to $98.0 \pm 16.0$. This reduces the minimum achievable false positive rate to 1\%, and hence allows the attacker to set $\alpha$ thresholds smaller than 10\% against private models which wasn't possible in the non-private case. However, the PPV is still less than 60\% for these thresholds.

Figure~\ref{fig:purchase_100_mi_1_analysis_dp_b} shows the attack performance at different thresholds. Due to the reduced gap between the member and non-member loss distributions, the PPV is close to 60\% across all loss thresholds even if the maximum membership advantage is considerable (close to 20\% for $\epsilon = 100$). Thus even with minimal privacy noise, the privacy leakage risk to membership inference attacks is significantly mitigated. For $\epsilon = 1$, the minimum false positive rate goes to 0.01\%, allowing \bfattack{Yeom} to achieve high PPV but with high deviation. The average PPV is close to 50\%. We observe similar trends for other data sets and hence defer these results to Appendix~\ref{appendix:private_results}. Similar to \bfattack{Yeom}, \bfattack{Shokri} attack also achieves only around 60\% PPV even for $\epsilon = 100$, and hence does not pose significant privacy threat.

\shortsection{\bfattack{Merlin} and \bfattack{Morgan} Attacks}
Figure~\ref{fig:purchase_100_mi_2_analysis_dp_a} shows the distribution of \bfattack{Merlin} ratio for member and non-member records on a private model trained with $\epsilon = 100$.  When compared to the corresponding distribution for a non-private model (see Figure~\ref{fig:purchase_100_mi_2_analysis_a}), the gap between the distributions is greatly reduced. This restricts the privacy leakage across all thresholds, as shown in Figure~\ref{fig:purchase_100_mi_2_analysis_dp_b}. Though the maximum PPV can still be high enough to pose an exposure risk at higher privacy loss budgets. We observe similar trends for \bfattack{Merlin} on the other data sets (see Appendix~\ref{appendix:private_results}). Unlike for non-private models, \bfattack{Morgan} does not achieve close to 100\% PPV as the members and non-members are not easily distinguishable due to the added privacy noise (see Figure~\ref{fig:morgan_purchase_eps_100}), but it does better than \bfattack{Merlin}. Regardless, models trained with high privacy loss budgets can still be vulnerable to \bfattack{Merlin} and \bfattack{Morgan} even if \bfattack{Yeom} and \bfattack{Shokri} do not succeed. This shows the importance of choosing appropriate privacy loss budgets for differential privacy mechanisms.

\subsection{Imbalanced Scenarios}
\label{sec:imbalanced_prior_case}

\begin{table}[tb]
    \centering
    \begin{tabular}{cS[table-format=2.1]S[table-format=2.1,separate-uncertainty,table-figures-uncertainty=1]S[table-format=2.1,separate-uncertainty,table-figures-uncertainty=1]S[table-format=2.1,separate-uncertainty,table-figures-uncertainty=1]S[table-format=3.1,separate-uncertainty,table-figures-uncertainty=1]}
        & {$\gamma$} & \bfattack{Yeom} & \bfattack{Shokri} & \bfattack{Merlin} & \bfattack{Morgan} \\ \hline
        \multirow{5}{*}{\parbox{2.2cm}{\bfdataset{Purchase-100X}}} & 0.1 & 96.5 \pm 0.1 & 94.9 \pm 0.1 & 99.3 \pm 0.7 & 100.0 \pm 0.0 \\
        & 0.5 & 84.5 \pm 0.1 & 81.9 \pm 0.7 & 97.2 \pm 2.8 & 100.0 \pm 0.0 \\
        & 1.0 & 73.0 \pm 0.2 & 73.4 \pm 1.6 & 93.4 \pm 6.3 & 98.0 \pm 4.0 \\
        & 2.0 & 57.6 \pm 0.3 & 62.3 \pm 7.7 & 84.0 \pm 5.6 & 99.1 \pm 1.7 \\
        & 10.0 & 21.2 \pm 0.1 & 33.1 \pm 4.5 & 69.7 \pm 13.8 & 97.5 \pm 5.0 \\ \hline
        \multirow{5}{*}{\parbox{2.2cm}{\bfdataset{Texas-100}}} & 0.1 & 97.0 \pm 0.1 & 97.3 \pm 0.3 & 99.2 \pm 0.7 & 100.0 \pm 0.0 \\
        & 0.5 & 86.4 \pm 1.1 & 92.0 \pm 1.5 & 95.0 \pm 3.6 & 98.4 \pm 0.5 \\
        & 1.0 & 76.1 \pm 1.6 & 89.4 \pm 1.5 & 92.0 \pm 4.5 & 95.7 \pm 4.6 \\ 
        & 2.0 & 62.4 \pm 0.4 & 84.4 \pm 3.7 & 87.7 \pm 11.1 & 97.4 \pm 2.7 \\
        & 10.0 & {-} & {-} & {-} & {-} \\ \hline
        \multirow{5}{*}{\parbox{2.2cm}{\bfdataset{RCV1X}}} & 0.1 & 93.3 \pm 0.5 &95.5 \pm 1.6  & 99.8 \pm 0.3 & 100.0 \pm 0.0 \\
        & 0.5 & 72.5 \pm 0.9 & 92.5 \pm 3.3 & 94.3 \pm 6.5 & 99.5 \pm 1.0 \\
        & 1.0 & 57.9 \pm 1.0 & 91.7 \pm 4.2 & 98.8 \pm 2.4 & 100.0 \pm 0.0 \\
        & 2.0 & 40.5 \pm 1.5 & 89.1 \pm 1.8 & 98.8 \pm 2.4 & 98.8 \pm 2.4 \\
        & 10.0 & 12.2 \pm 0.3 & 67.3 \pm 5.1 & 74.3 \pm 15.9 & 93.0 \pm 9.8 \\ \hline
        \multirow{5}{*}{\parbox{2.2cm}{\bfdataset{CIFAR-100}}} & 0.1 & 96.0 \pm 0.2 & 91.6 \pm 0.0 & 97.9 \pm 1.9 & 100.0 \pm 0.0 \\
        & 0.5 & 84.6 \pm 0.5 & 75.7 \pm 0.6 & 86.4 \pm 1.7 & 100.0 \pm 0.0 \\
        & 1.0 & 72.7 \pm 0.8 & 64.9 \pm 0.7 & 75.0 \pm 2.6 & 100.0 \pm 0.0 \\
        & 2.0 & 56.7 \pm 0.6 & 55.3 \pm 2.2 & 74.0 \pm 8.1 & 100.0 \pm 0.0 \\
        & 10.0 & {-} & {-} & {-} & {-} \\
    \end{tabular}
    \caption{Effect of varying $\gamma$ on maximum PPV achieved by attacks against non-private models. {\rm All values are in percentage.}}
    \label{tab:leakage_comparison_imbalanced_prior}
\end{table}

\begin{figure*}[ptb]
    \centering
    \begin{subfigure}[b]{0.35\textwidth}
    \centering
    \includegraphics[width=.95\linewidth]{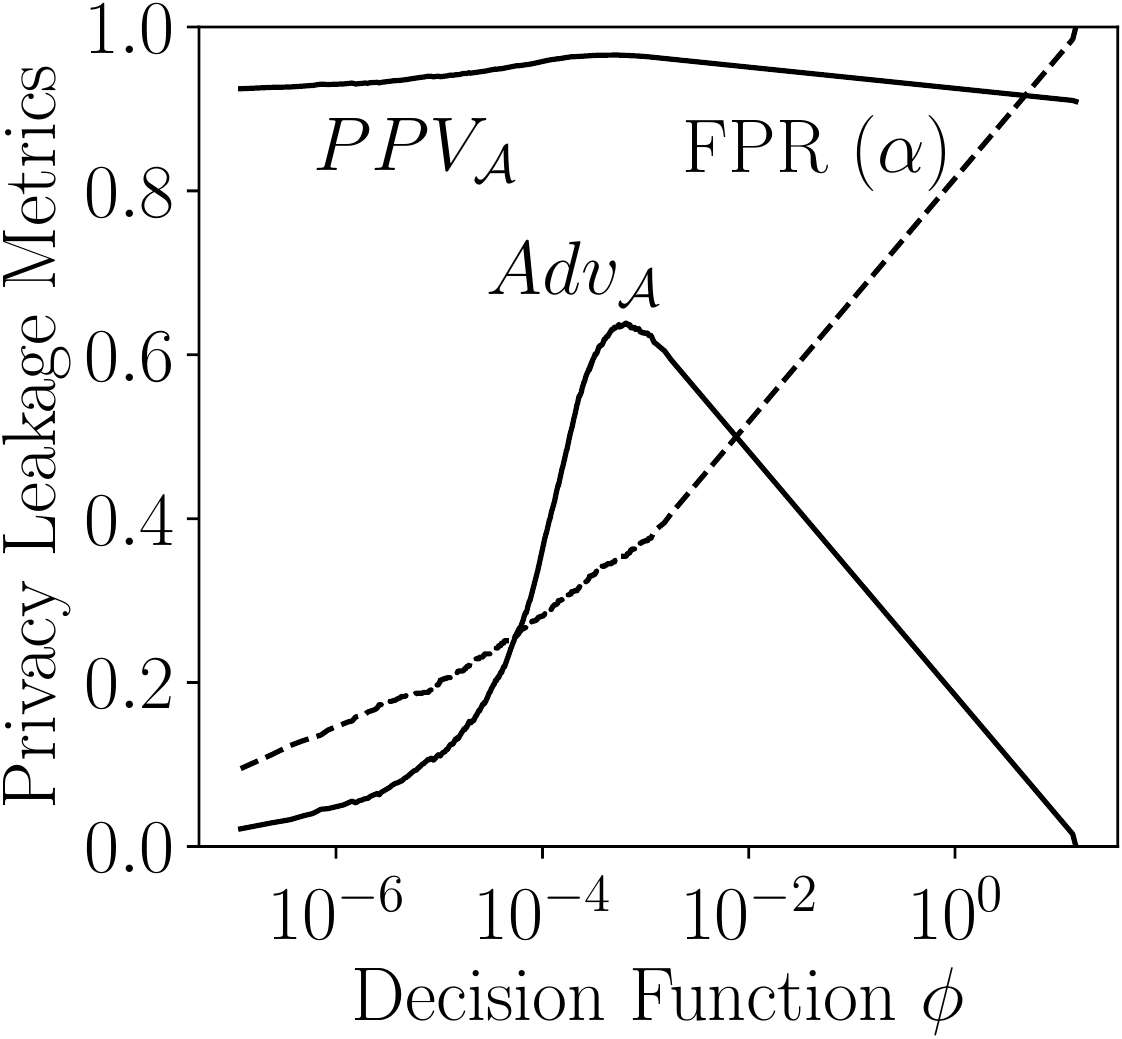}
    \caption{\bfattack{Yeom} performance at $\gamma = 0.1$.}
    \label{fig:purchase_100_gamma_0_1_mi_1_analysis}
    \end{subfigure} \qquad \qquad
    \begin{subfigure}[b]{0.35\textwidth}
    \centering
    \includegraphics[width=.95\linewidth]{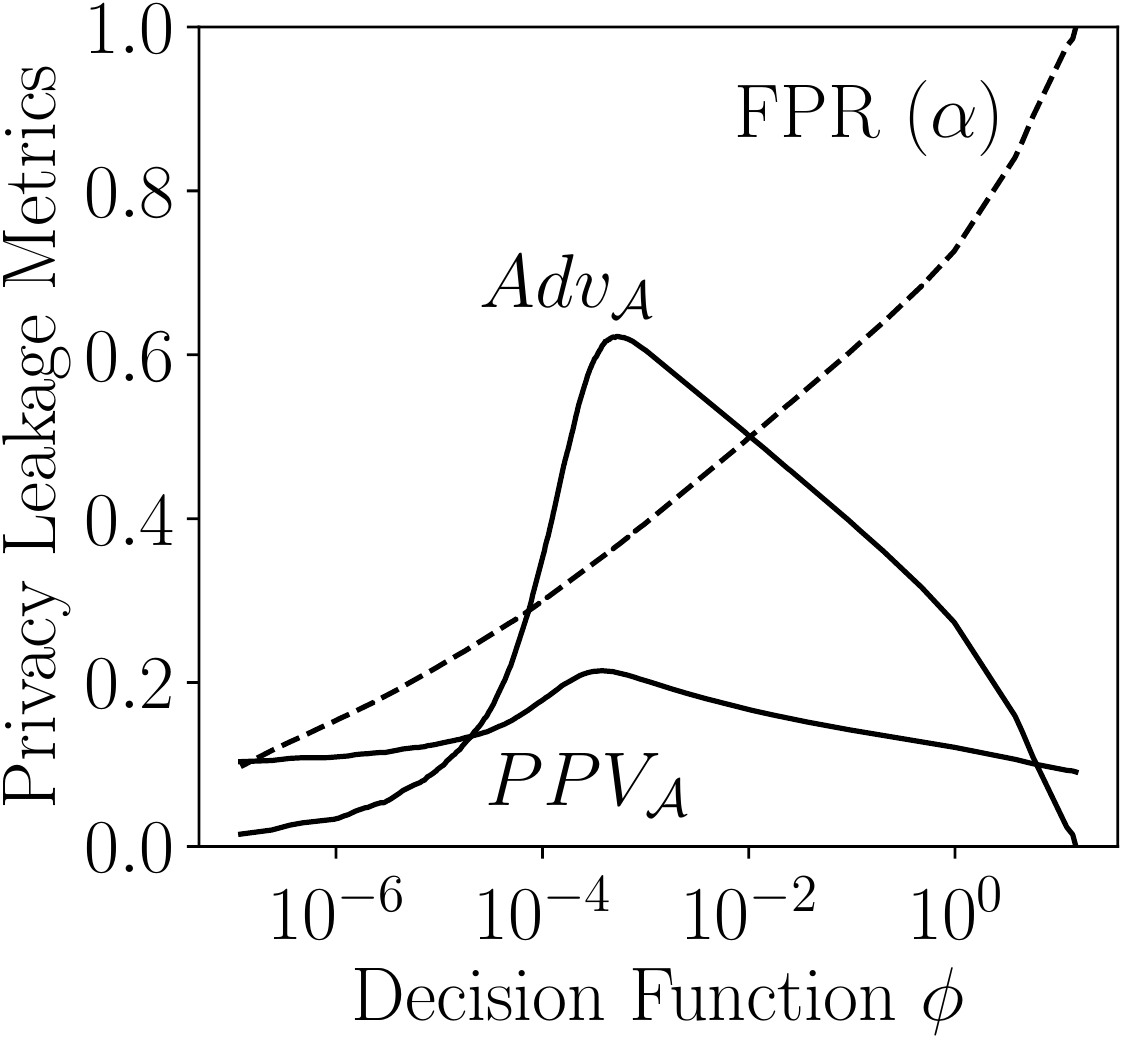}
    \label{fig:purchase_100_gamma_10_mi_1_analysis}
    \caption{\bfattack{Yeom} performance at $\gamma = 10$.}
    \end{subfigure}
    \begin{subfigure}[b]{0.35\textwidth}
    \centering
    \includegraphics[width=.95\linewidth]{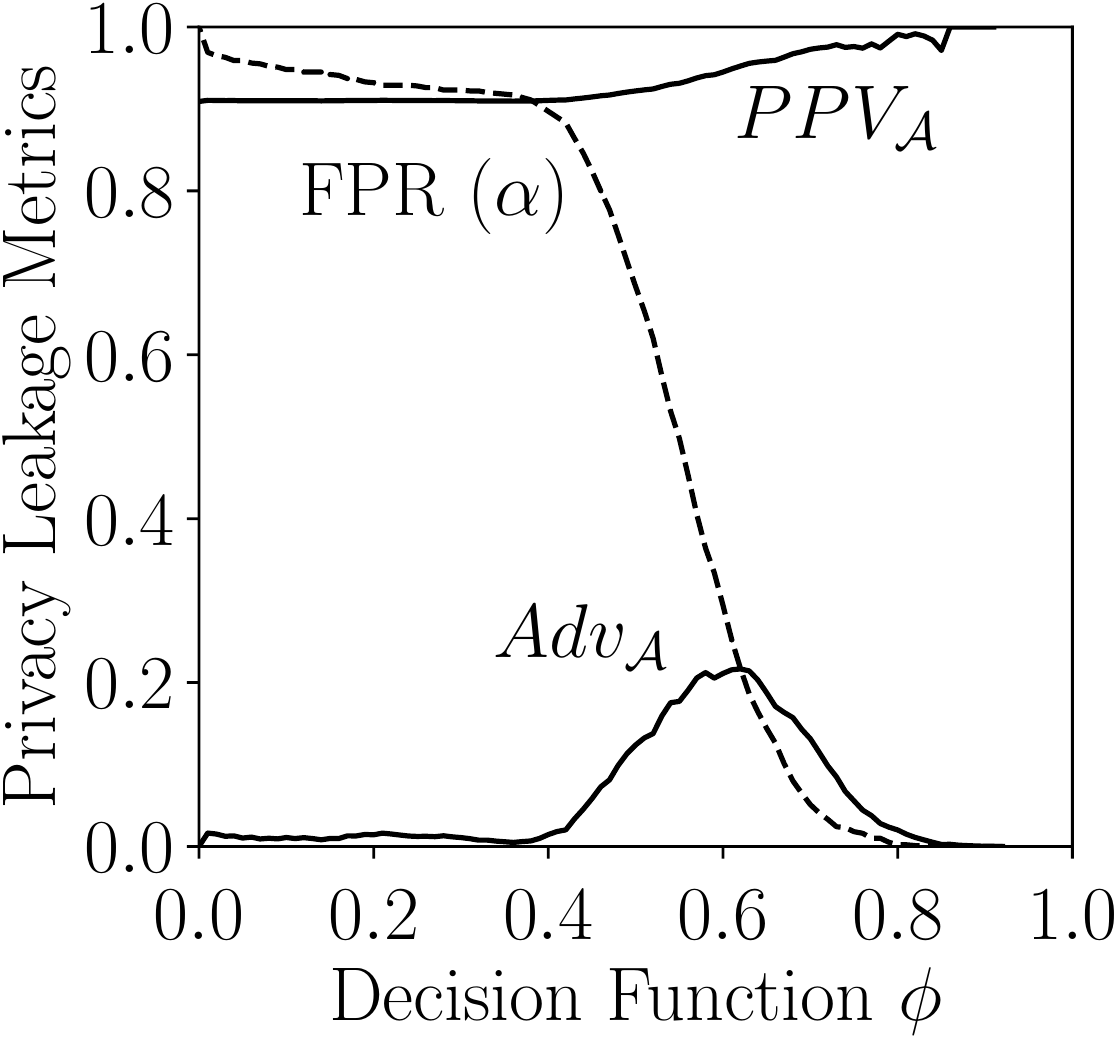}
    \caption{\bfattack{Merlin} performance at $\gamma = 0.1$.}
    \label{fig:purchase_100_gamma_0_1_mi_2_analysis}
    \end{subfigure} \qquad \qquad
    \begin{subfigure}[b]{0.35\textwidth}
    \centering
    \includegraphics[width=.95\linewidth]{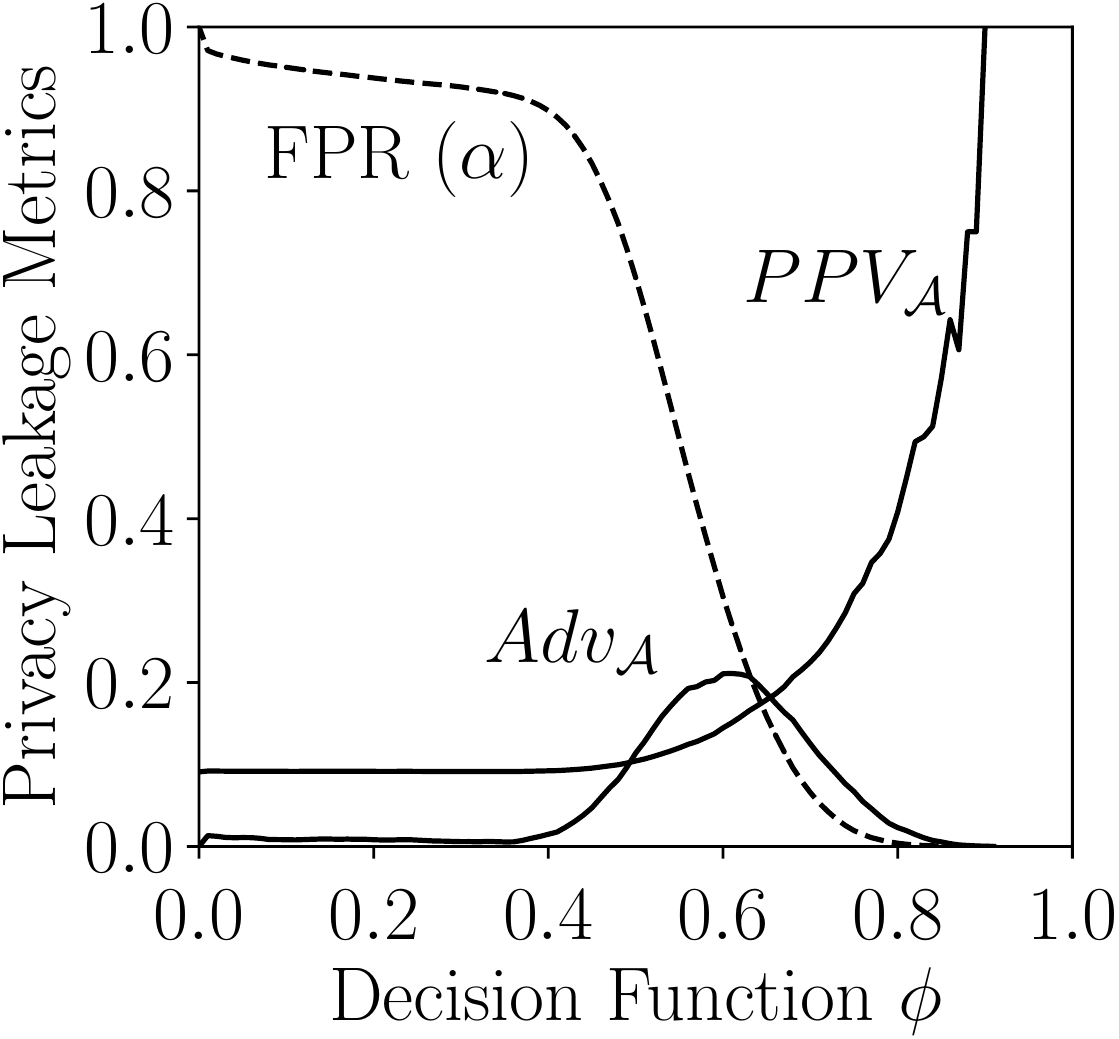}
    \caption{\bfattack{Merlin} performance at $\gamma = 10$.}
    \label{fig:purchase_100_gamma_10_mi_2_analysis}
    \end{subfigure}
    \caption{Attack performance on \bfdataset{Purchase-100X} for imbalanced prior setting. {\rm $PPV_\cA$ varies with $\gamma$, while $Adv_\cA$ remains almost same.}}
    \label{fig:purhcase_100_attack_performance_imbalanced_prior}
\end{figure*}

As discussed in Section~\ref{sec:privacy_metrics}, the membership advantage metric does not consider the prior distribution probability and hence does not capture the true privacy risk for imbalanced prior settings. In this section, we provide empirical evidence that the PPV metric captures  privacy leakage more naturally in imbalanced prior settings, and hence is a more reliable metric for evaluating the privacy leakage.

In imbalanced prior settings, the candidate pool from which the attacker samples records for inference testing has $\gamma$ times more non-member records than members. In other words, a randomly selected candidate is $\gamma$ times more likely to be a non-member than a member record. We keep the training set size fixed to 10,000 records as in our previous experiments, so need a test set size that is $\gamma$ times the training set size. For each data set, we set $\gamma$ as high as possible given the available data.  As mentioned in Section~\ref{sec:exp_setup}, we constructed expanded versions of the \bfdataset{Purchase-100} and \bfdataset{RCV1} data sets to enable these experiments. Both the \bfdataset{Purchase-100X} and \bfdataset{RCV1X} data sets have more than 200,000 records, and hence are large enough to allow setting $\gamma = 10$. We did not have source data to expand \bfdataset{Texas-100}, so are left with a data set with only 67,000 records and hence only have results for $\gamma = 2$. The threshold selection procedure (Procedure~\ref{algo:inference_training}) uses holdout training and test sets that are disjoint from the target training and test sets mentioned above, so the data set needs at least $(\gamma + 1) \times 20,000$ records to run the experiments. 

Table~\ref{tab:leakage_comparison_imbalanced_prior} shows the effect of varying $\gamma$ on the maximum PPV of inference attacks against non-private models trained on different data sets. We can see a clear drop in PPV values across all data sets with increasing $\gamma$ values for \bfattack{Yeom}, \bfattack{Shokri} and \bfattack{Merlin}. Although, \bfattack{Merlin} consistently outperforms \bfattack{Yeom} and \bfattack{Shokri} across all settings. At $\gamma = 0.1$, the base rate for PPV is 90\%. While \bfattack{Yeom} and \bfattack{Shokri} achieve between 90\% and 97\% PPV, \bfattack{Merlin} achieves close to 100\% PPV across all data sets. For $\gamma = 2$, the maximum PPV of \bfattack{Yeom} is close to 60\%, whereas \bfattack{Merlin} still achieves high enough PPV to pose some privacy threat. Both the \bfattack{Yeom} and \bfattack{Shokri} attacks, and our \bfattack{Merlin} attack are less successful as the $\gamma$ value increases to 10. However, \bfattack{Morgan} consistently achieves close to 100\% PPV across all settings, thereby showing the vulnerability of non-private models even in the skewed prior settings. This is graphically shown in Figure~\ref{fig:morgan_purchase_gamma_10} where \bfattack{Morgan} is able to identify the most vulnerable members on \bfdataset{Purchase-100X} even at $\gamma = 10$. The advantage values remain more or less the same across different $\gamma$ values for both \bfattack{Yeom} and \bfattack{Merlin} on \bfdataset{Purchase-100X}, as shown in Figure~\ref{fig:purhcase_100_attack_performance_imbalanced_prior}. These results support our claim that PPV is a more reliable metric in skewed prior scenarios. We observe the same trend for the other data sets, and hence do not include their plots.

While \bfattack{Yeom}, \bfattack{Shokri} and \bfattack{Merlin} do not pose an exposure threat in the imbalanced prior settings where $\gamma$ values are higher than 10, \bfattack{Morgan} still exposes some vulnerable members with close to 100\% PPV. Thus, our proposed attacks pose significant threat even in more realistic settings of skewed priors, where the existing attacks fail. We observe that the private models are not vulnerable to any of our inference attacks in the imbalanced prior setting where $\gamma > 1$. At $\gamma = 2$, the best attack achieves maximum PPV close to 48\% across all data sets, whereas at $\gamma = 10$, this further drops to around 17\%. Hence we do not show the membership inference attack results against private models for these settings.

\section{Conclusion}
Understanding the privacy risks posed by machine learning involves considerable challenges, and there remains a large gap between achievable privacy guarantees, and what can be inferred using known attacks in practice. While membership inference has previously been evaluated in balanced prior settings, we consider scenarios with imbalanced priors and show that there are attacks which pose serious privacy threats even in such settings where previous attacks fail. 

We introduce a novel threshold selection procedure that allows adversaries to choose inference thresholds specific to their attack goals, and propose two new membership inference attacks, \bfattack{Merlin} and \bfattack{Morgan}, that outperform previous attacks in the settings that concern us most: being able to identify members, with very high confidence, even from candidate pools where most candidates are not members. From experiments on four data sets under different prior distribution settings, we find that the non-private models are highly vulnerable to such attacks, and the models trained with high privacy loss budgets can still be vulnerable. 

\section*{Availability}

All of our code and data for our experiments is available at  \url{https://github.com/bargavj/EvaluatingDPML}.

\section*{Acknowledgments}

This work was partially supported by grants from the National Science Foundation (\#1717950 and \#1915813).

\bibliographystyle{plainnat}
\bibliography{ref}
\appendix
\section{Hyperparameters}\label{appendix:hyperparameters}
For each data set, the training set is fixed to 10,000 randomly sampled records and the test set size is varied to reflect different prior probability distributions. We sample $\gamma$ times the number of training set records to create the test set. For both \bfdataset{Purchase-100X} and \bfdataset{RCV1X}, we use $\gamma = \{0.1, 0.5, 1, 2, 10\}$; for \bfdataset{Texas-100} and \bfdataset{CIFAR-100}, we use $\gamma = \{0.1, 0.5, 1, 2\}$ since they are too small for experiments with larger $\gamma$. For training the models, we minimize the cross-entropy loss using the Adam optimizer and perform grid search to find the best values for hyperparameters such as batch size, learning rate, $\ell_2$ penalty, clipping threshold and number of iterations. We find a batch size of 200, clipping threshold of 4, and $\ell_2$ penalty of $10^{-8}$ work best across all the data sets, except for \bfdataset{CIFAR-100} where we use $\ell_2$ penalty of $10^{-4}$, and \bfdataset{RCV1X} where we use clipping threshold of 1. We use a learning rate of 0.005 for \bfdataset{Purchase-100X} and \bfdataset{Texas-100}, 0.003 for \bfdataset{RCV1X}, and 0.001 for \bfdataset{CIFAR-100}. We set the training epochs to 100 for \bfdataset{Purchase-100X} and \bfdataset{CIFAR-100}, 30 for \bfdataset{Texas-100}, and 80 for \bfdataset{RCV1X}. We fix the differential privacy failure parameter $\delta$ as $10^{-5}$ to keep it smaller than the inverse of the training set size, generally considered the maximum acceptable value for $\delta$. For an in-depth analysis of hyperparameter tuning on private learning, see \cite{abadi2016deep}.

\section{Additional Results for Non-Private Models}\label{appendix:non_private_results}

\begin{figure*}[ptb]
    \centering
    \begin{subfigure}[b]{0.35\textwidth}
    \centering
    \includegraphics[width=0.95\linewidth]{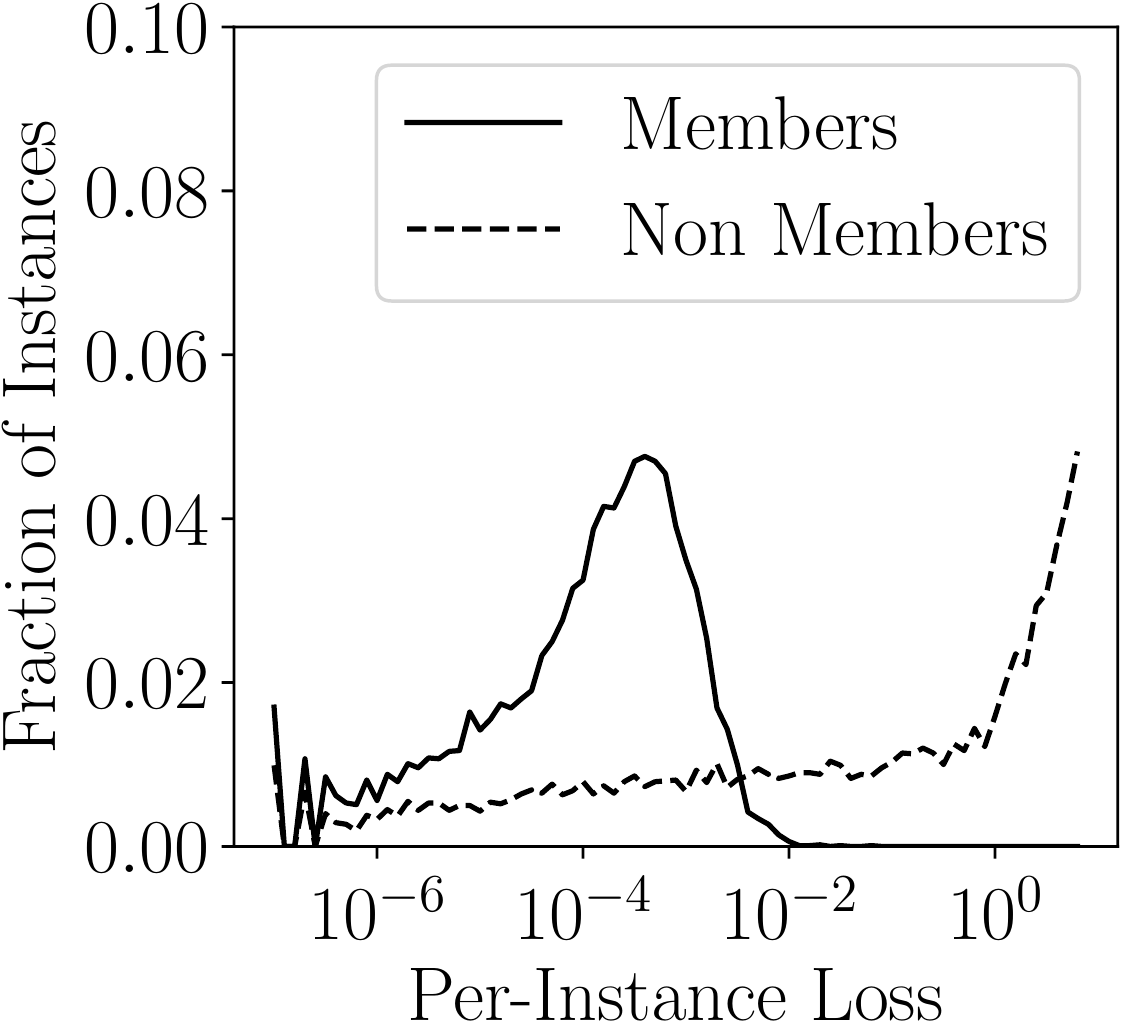}
    \caption{Loss distribution for \bfdataset{Texas-100}}
    \label{fig:texas_100_mi_1_analysis_a}
    \end{subfigure}\qquad \qquad
    \begin{subfigure}[b]{0.35\textwidth}
    \centering
    \includegraphics[width=0.95\linewidth]{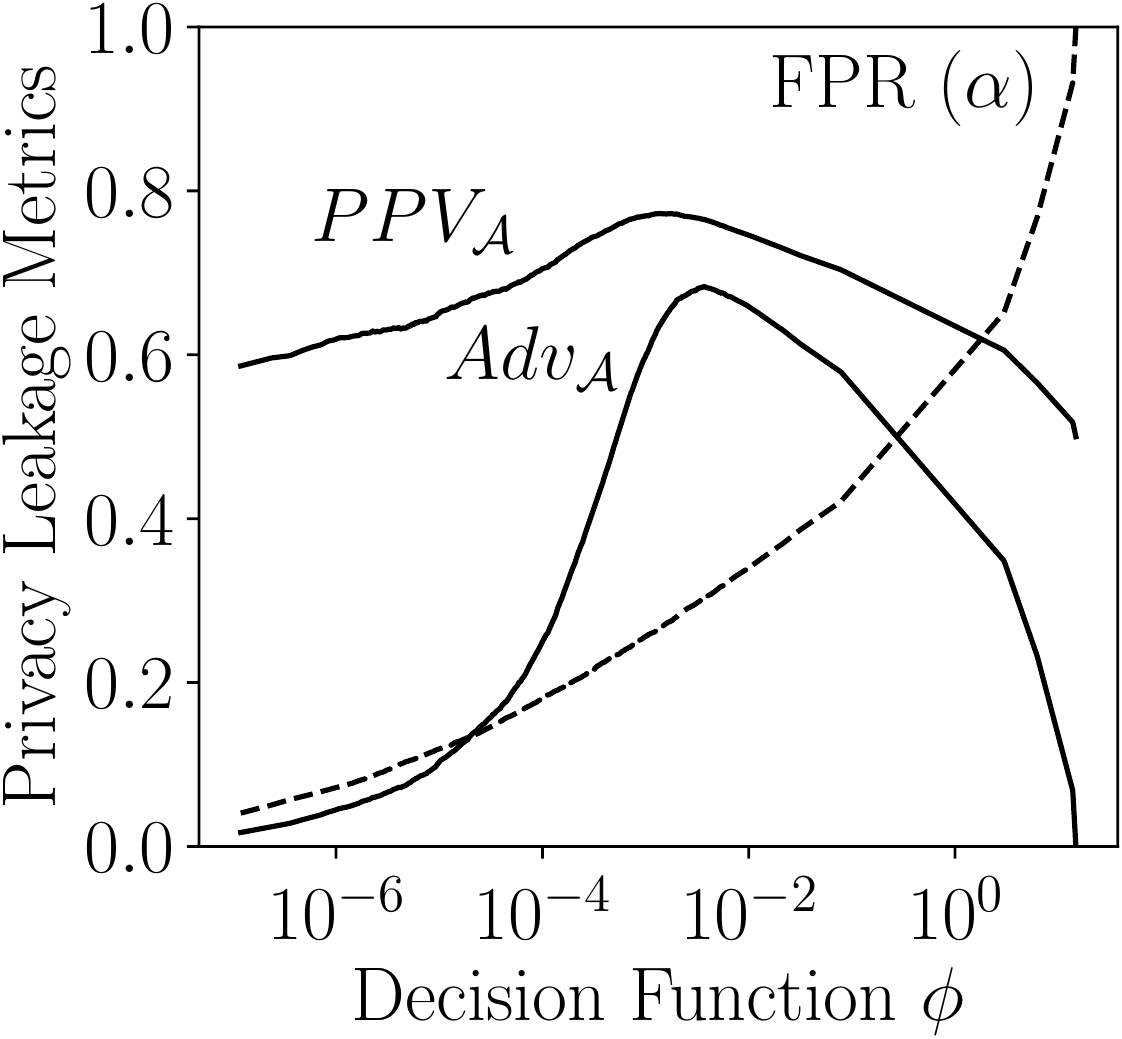}
    \caption{\bfattack{Yeom} on \bfdataset{Texas-100}}
    \label{fig:texas_100_mi_1_analysis_b}
    \end{subfigure}
    \begin{subfigure}[b]{0.35\textwidth}
    \centering
    \includegraphics[width=0.95\linewidth]{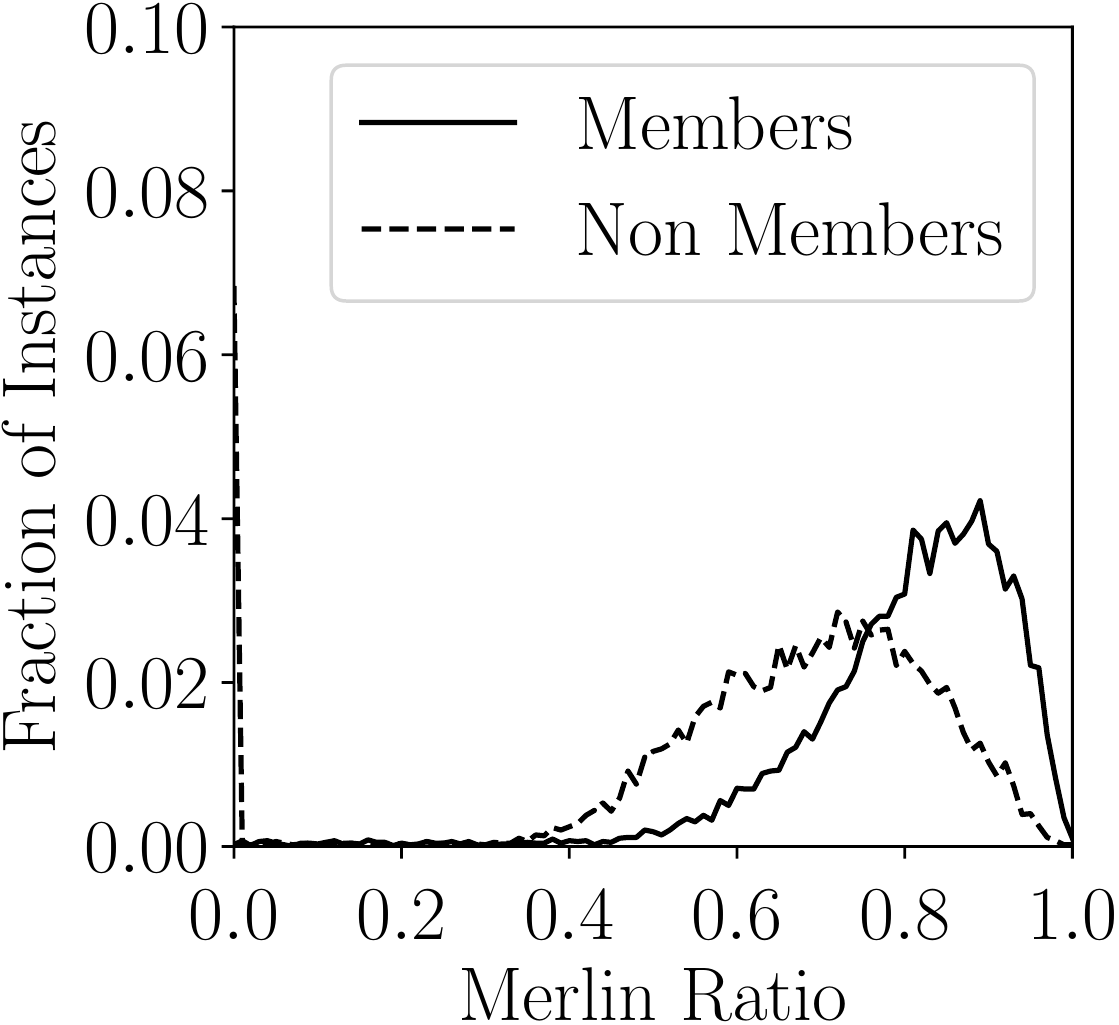}
    \caption{Merlin ratio for \bfdataset{Texas-100}}
    \label{fig:texas_100_mi_2_analysis_a}
    \end{subfigure}\qquad \qquad
    \begin{subfigure}[b]{0.35\textwidth}
    \centering
    \includegraphics[width=0.95\linewidth]{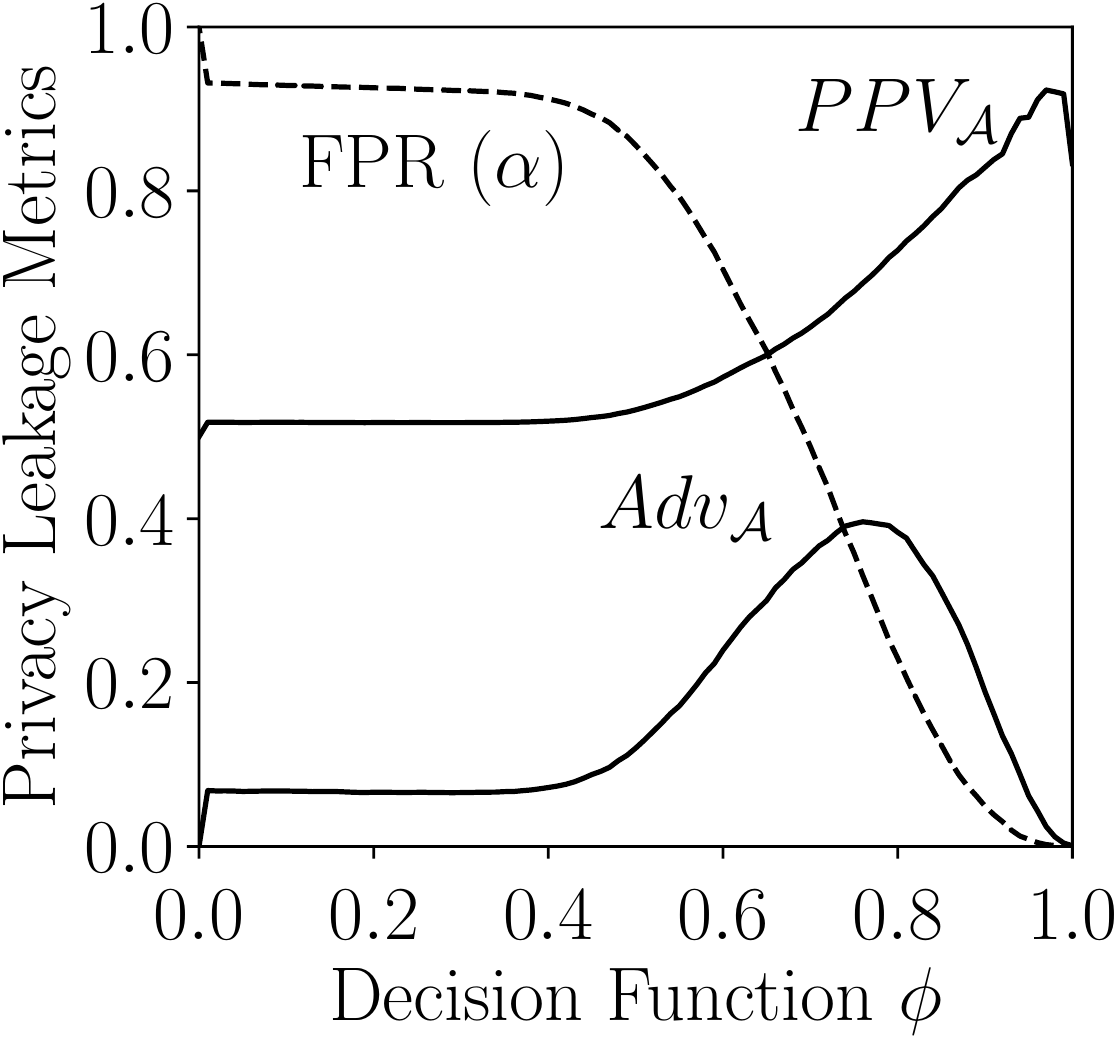}
    \caption{\bfattack{Merlin} on \bfdataset{Texas-100}}
    \label{fig:texas_100_mi_2_analysis_b}
    \end{subfigure}
    \caption{Analysis of \bfattack{Yeom} and \bfattack{Merlin} against non-private models trained on \bfdataset{Texas-100} in the balanced prior setting.}
    \label{fig:texas_100_analysis}
\end{figure*}

\begin{table*}[ptb]
    \centering
    \small
    \begin{tabular}{llS[table-format=2.2]cS[table-format=2.1,separate-uncertainty,table-figures-uncertainty=1]S[table-format=2.1,separate-uncertainty,table-figures-uncertainty=1]}
        & & {$\alpha$} & $\phi$ & {$Adv_\cA$} & {$PPV_\cA$} \\ \hline
        \multirow{5}{*}{\parbox{1.5cm}{\bfattack{Yeom}}} & Fixed FPR & 1.00 & - & {-} & {-} \\
        & Min FPR & 3.00 & \num{0} & 0.4 \pm 0.9 & 12.0 \pm 24.0 \\
        & Fixed $\phi$ & {-} & \num{1.1 \pm 2.0 e-2} & 51.3 \pm 2.6 & 75.0 \pm 1.6 \\
        & Max $PPV_\cA$ & 26.00 & \num{1.8 \pm 0.4 e-3} & 59.2 \pm 11.7 & 76.1 \pm 1.6 \\
        & Max $Adv_\cA$ & 31.00 & \num{6.6 \pm 1.3 e-3} & 62.9 \pm 7.7 & 75.0 \pm 0.6 \\ \hline
        \multirow{5}{*}{\parbox{1.7cm}{\bfattack{Yeom CBT}}} & Min FPR & 0.01 & \rm \num{0}, \num{3.4e-6}, \num{9.2e-2} & 10.2 \pm 2.6 & 92.0 \pm 2.3 \\
        & Max $PPV_\cA$ & 0.01 & \rm \num{0}, \num{3.4e-6}, \num{9.2e-2} & 10.2 \pm 2.6 & 92.0 \pm 2.3 \\
        & Fixed FPR & 1.00 & \rm \num{0}, \num{4.8e-6}, \num{9.2e-2} & 11.2 \pm 2.9 & 91.2 \pm 2.1 \\
        & Fixed $\phi$ & {-} & \rm(\num{0.1}, \num{8.3}, \num{554.8})$\times$\num{e-4} & 49.2 \pm 12.1 & 76.3 \pm 1.4 \\
        & Max $Adv_\cA$ & 52.00 & \rm \num{1.2e-7}, \num{7.3e-3}, \num{4.3} & 61.5 \pm 8.0 & 75.8 \pm 1.2 \\ \hline
        \multirow{5}{*}{\parbox{1.5cm}{\bfattack{Shokri}}} & Min FPR & 0.01 & \num{0.99 \pm 0.00} & 1.0 \pm 0.5 & 72.6 \pm 8.6 \\
        & Max $PPV_\cA$ & 0.70 & \num{0.92 \pm 0.01} & 13.8 \pm 1.1 & 89.4 \pm 1.5 \\
        & Fixed FPR & 1.00 & \num{0.91 \pm 0.01} & 16.0 \pm 1.3 & 88.9 \pm 1.7 \\
        & Max $Adv_\cA$ & 31.00 & \num{0.51 \pm 0.01} & 64.1 \pm 1.2 & 74.7 \pm 0.9 \\ & Fixed $\phi$ & {-} & \num{0.50 \pm 0.00} & 64.0 \pm 1.4 & 74.1 \pm 1.3 \\
        \hline
        \multirow{5}{*}{\parbox{1.6cm}{\bfattack{Shokri CBT}}} & Min FPR & 0.01 & \num{0.6}, \num{0.9}, \num{1.8} & 3.0 \pm 0.4 & 77.5 \pm 4.6 \\
        & Fixed FPR & 1.00 & \num{0.5}, \num{0.9}, \num{1.5} & 9.8 \pm 1.1 & 84.7 \pm 2.5 \\
        & Max $PPV_\cA$ & 1.60 & \num{0.5}, \num{0.8}, \num{1.5} & 13.6 \pm 1.0 & 85.8 \pm 2.0 \\
        & Fixed $\phi$ & {-} & \num{0.5}, \num{0.5}, \num{0.5} & 64.0 \pm 1.4 & 74.1 \pm 1.3 \\
        & Max $Adv_\cA$ & 38.00 & \num{0.0}, \num{0.3}, \num{1.5} & 58.6 \pm 1.3 & 75.5 \pm 1.2 \\ \hline
        \multirow{4}{*}{\parbox{1.7cm}{\bfattack{Merlin}}} & Min FPR & 0.01 & \num{1.00 \pm 0.01} & 0.1 \pm 0.1 & 51.9 \pm 42.4 \\
        & Max $PPV_\cA$ & 0.06 & \num{0.99 \pm 0.00} & 0.3 \pm 0.2 & 92.0 \pm 4.5 \\
        & Fixed FPR & 1.00 & \num{0.95 \pm 0.00} & 4.9 \pm 1.3 & 87.8 \pm 2.7 \\
        & Max $Adv_\cA$ & 36.00 & \num{0.76 \pm 0.00} & 37.8 \pm 1.5 & 68.0 \pm 0.8 \\ \hline
        \bfattack{Morgan} & Max $PPV_\cA$ & {-} & \rm \num{1.2e-4}, \num{5.1e-3}, \num{0.98} & 0.5 \pm 0.2 & 95.7 \pm 4.6 \\
    \end{tabular}
    \caption{Thresholds selected against non-private models trained on \bfdataset{Texas-100} with balanced prior. 
    \rm 
    The results are averaged over five runs such that the target model is trained from the scratch for each run. \bfattack{Yeom CBT} and \bfattack{Shokri CBT} use class-based thresholds, where $\phi$ shows the triplet of minimum, median and maximum thresholds across all classes. All values, except $\phi$, are reported in percentage.}
    \label{tab:texas_100_loss_selection_results_uniform_prior}
\end{table*}

\begin{figure*}[ptb]
    \centering
    \begin{subfigure}[b]{0.35\textwidth}
    \centering
    \includegraphics[width=0.95\linewidth]{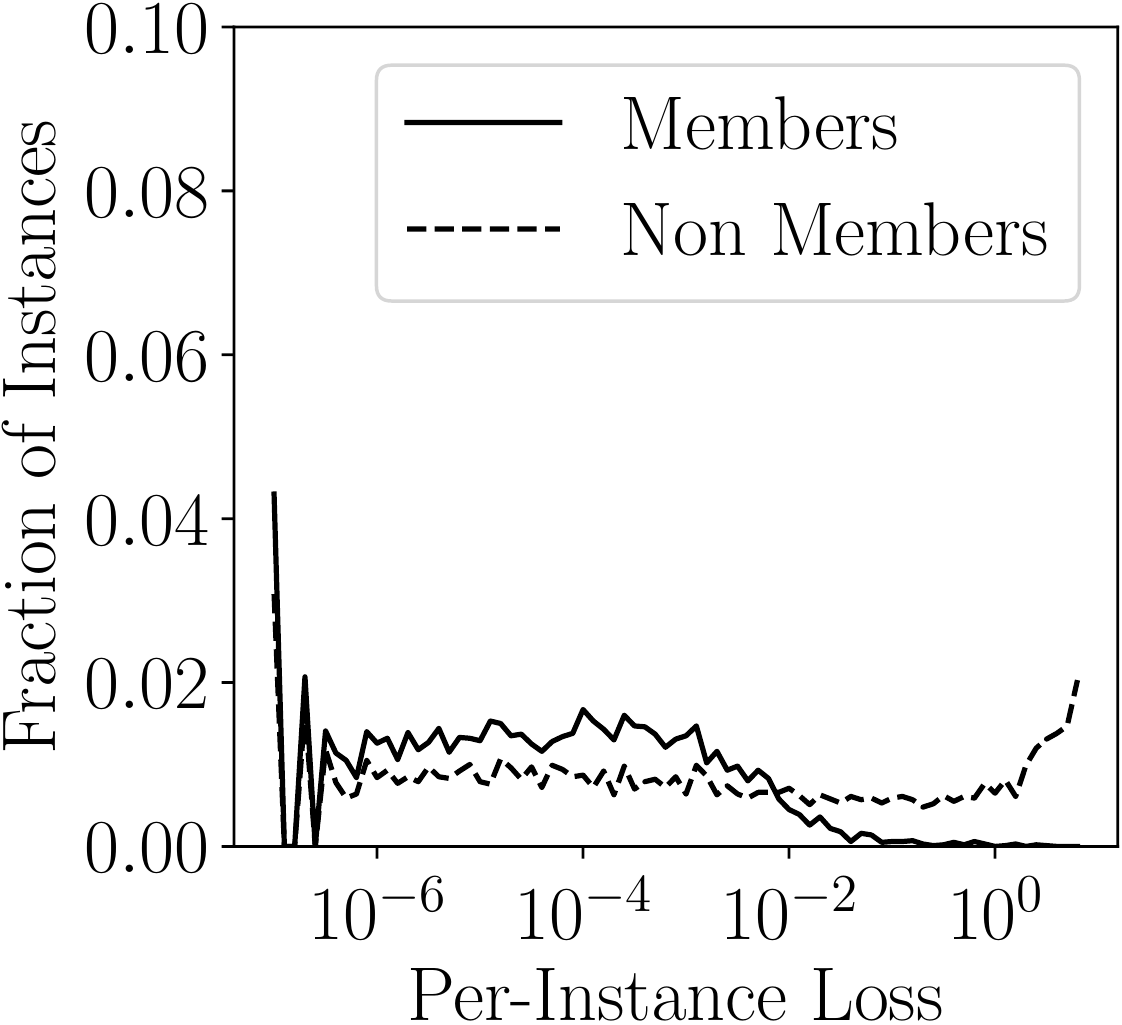}
    \caption{Loss distribution for \bfdataset{RCV1X}}
    \label{fig:rcv1_mi_1_analysis_a}
    \end{subfigure}\qquad \qquad
    \begin{subfigure}[b]{0.35\textwidth}
    \centering
    \includegraphics[width=0.95\linewidth]{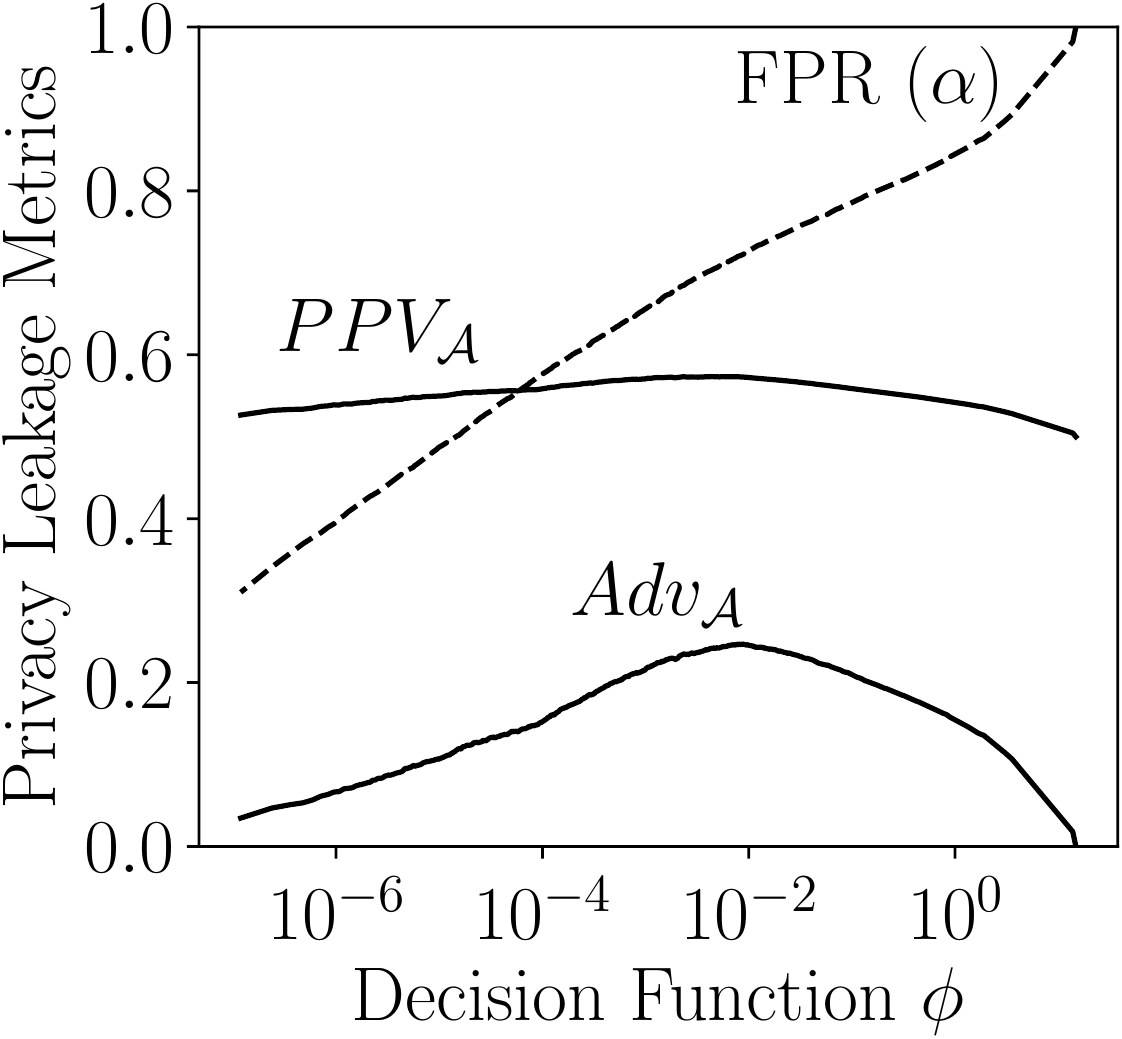}
    \caption{\bfattack{Yeom} on \bfdataset{RCV1X}}
    \label{fig:rcv1_mi_1_analysis_b}
    \end{subfigure}
    \begin{subfigure}[b]{0.35\textwidth}
    \centering
    \includegraphics[width=0.95\linewidth]{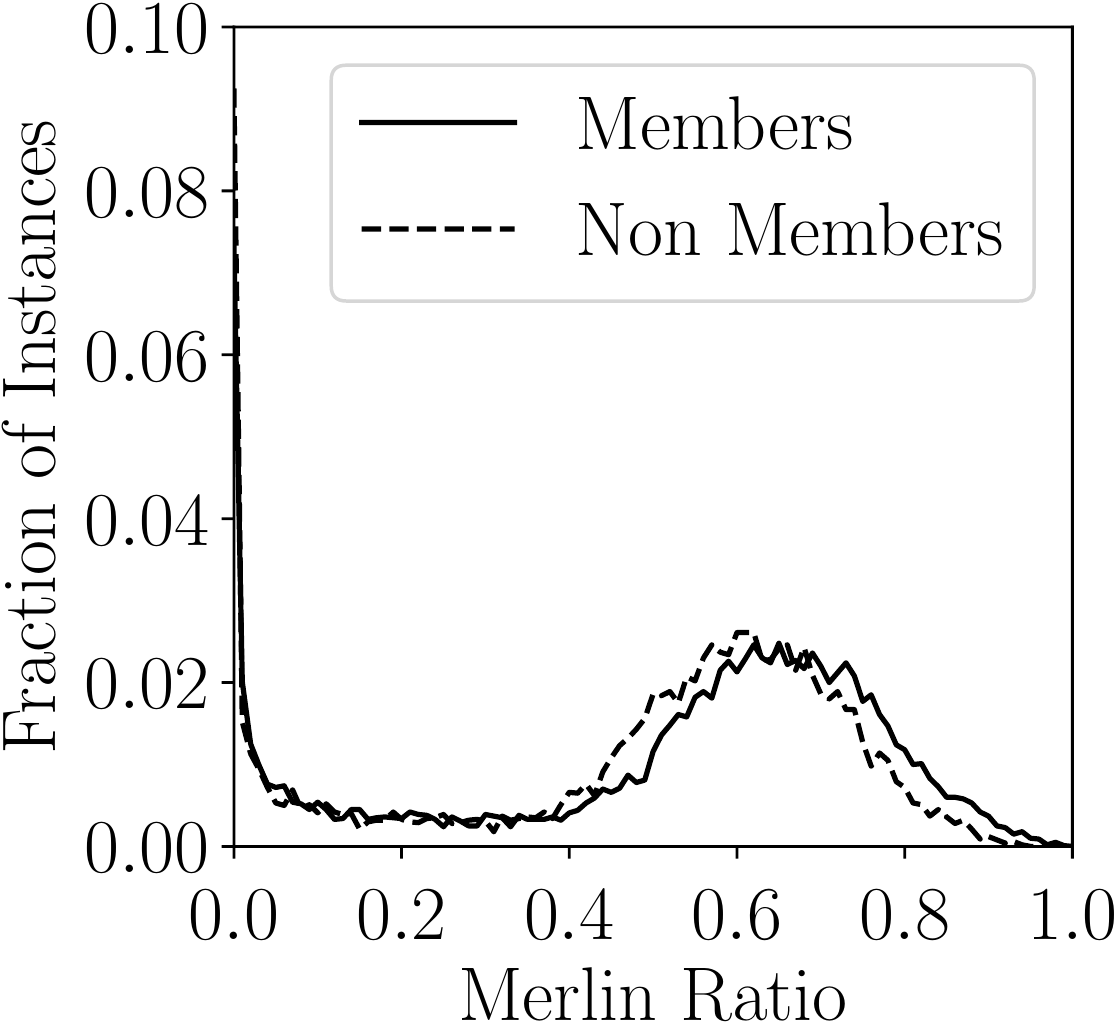}
    \caption{Merlin ratio for \bfdataset{RCV1X}}
    \label{fig:rcv1_mi_2_analysis_a}
    \end{subfigure}\qquad \qquad
    \begin{subfigure}[b]{0.35\textwidth}
    \centering
    \includegraphics[width=0.95\linewidth]{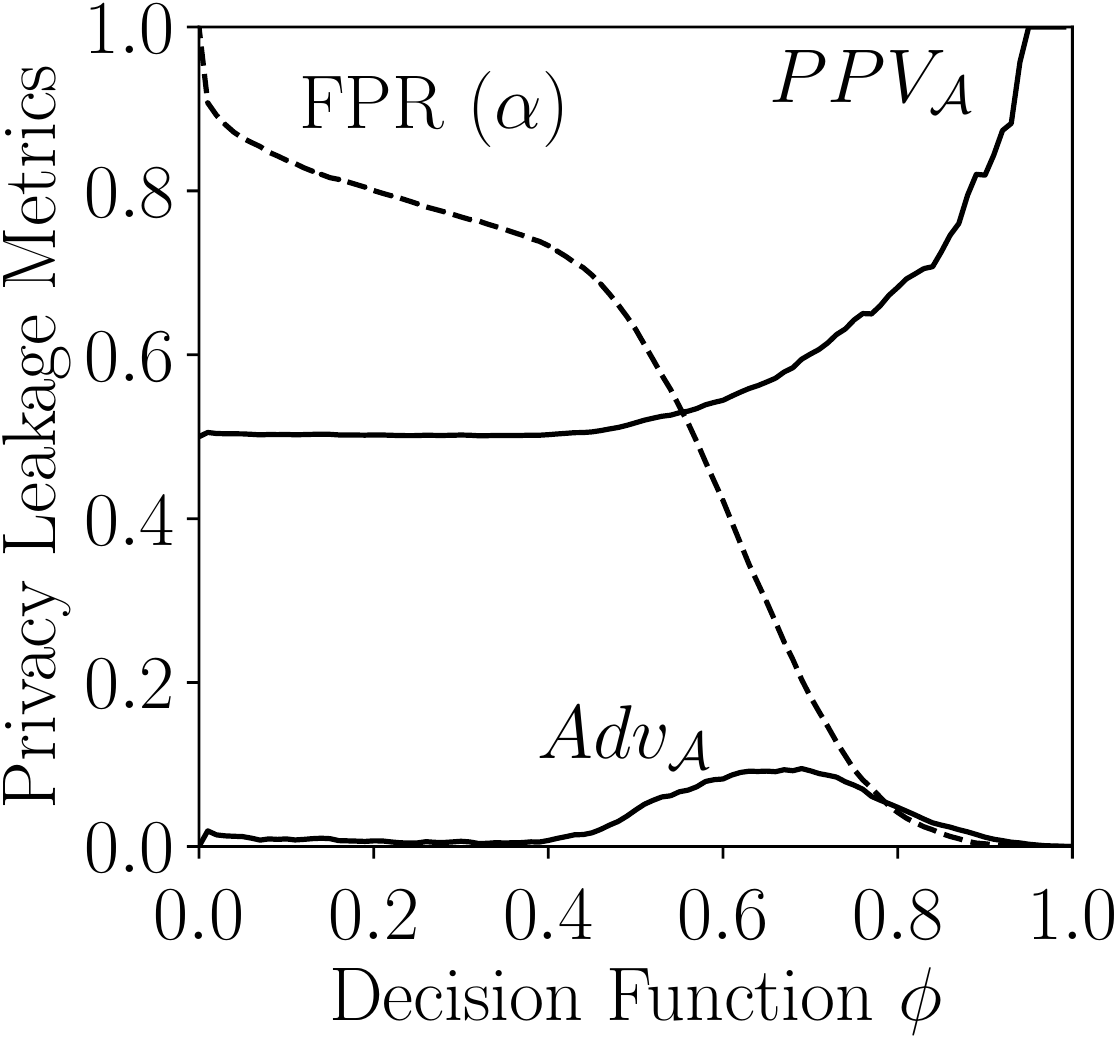}
    \caption{\bfattack{Merlin} on \bfdataset{RCV1X}}
    \label{fig:rcv1_mi_2_analysis_b}
    \end{subfigure}
    \caption{Analysis of \bfattack{Yeom} and \bfattack{Merlin} against non-private models trained on \bfdataset{RCV1X} in the balanced prior setting.}
    \label{fig:rcv1_analysis}
\end{figure*} 

\begin{table*}[tb]
    \centering
    \small
    \begin{tabular}{llS[table-format=2.2]cS[table-format=2.1,separate-uncertainty,table-figures-uncertainty=1]S[table-format=2.1,separate-uncertainty,table-figures-uncertainty=1]}
        & & {$\alpha$} & $\phi$ & {$Adv_\cA$} & {$PPV_\cA$} \\ \hline
        \multirow{5}{*}{\parbox{1.5cm}{\bfattack{Yeom}}} & Fixed FPR & 1.00 & - & {-} & {-} \\
        & Min FPR & 33.00 & \num{0} & 0.4 \pm 0.8 & 10.3 \pm 20.5 \\
        & Max $PPV_\cA$ & 67.00 & \num{0.5 \pm 0.2 e-3} & 24.8 \pm 5.0 & 57.9 \pm 1.0 \\
        & Max $Adv_\cA$ & 70.00 & \num{1.5 \pm 0.6 e-3} & 25.1 \pm 3.2 & 57.7 \pm 0.6 \\
        & Fixed $\phi$ & {-} & \num{3.2 \pm 3.2 e-3} & 26.9 \pm 2.7 & 58.0 \pm 0.8 \\ \hline
        \multirow{5}{*}{\parbox{1.7cm}{\bfattack{Yeom CBT}}} & Min FPR & 0.01 & \rm \num{0}, \num{0}, \num{3.8e-3} & 1.2 \pm 0.3 & 93.1 \pm 3.2 \\
        & Max $PPV_\cA$ & 0.01 & \rm \num{0}, \num{0}, \num{3.8e-3} & 1.2 \pm 0.3 & 93.1 \pm 3.2 \\
        & Fixed FPR & 1.00 & \rm \num{0}, \num{2.4e-8}, \num{3.8e-3} & 1.3 \pm 0.3 & 92.7 \pm 3.5 \\
        & Max $Adv_\cA$ & 70.00 & \rm \num{0}, \num{1.4e-3}, \num{9.0} & 22.4 \pm 3.2 & 59.0 \pm 0.5 \\
        & Fixed $\phi$ & {-} & \rm(\num{0.1}, \num{2.8}, \num{91.1})$\times$\num{e-4} & 21.6 \pm 5.6 & 57.3 \pm 1.2 \\ \hline
        \multirow{5}{*}{\parbox{1.5cm}{\bfattack{Shokri}}} & Min FPR & 0.01 &\num{0.97 \pm 0.01} & 0.6 \pm 0.2 & 91.7 \pm 4.2 \\
        & Max $PPV_\cA$ & 0.01 &\num{0.97 \pm 0.01} & 0.6 \pm 0.2 & 91.7 \pm 4.2 \\
        & Fixed FPR & 1.00 &\num{0.80 \pm 0.01} & 4.6 \pm 0.6 & 84.5 \pm 1.8 \\
        & Fixed $\phi$ & {-} & \num{0.50 \pm 0.00} & 24.0 \pm 0.8 & 57.3 \pm 0.4 \\
        & Max $Adv_\cA$ & 75.00 & \num{0.31 \pm 0.06} & 24.2 \pm 0.5 & 58.0 \pm 0.4 \\ \hline
        \multirow{5}{*}{\parbox{1.6cm}{\bfattack{Shokri CBT}}} & Min FPR & 0.01  & \num{0.5}, \num{1.0}, \num{1.9} & 0.8 \pm 0.2 & 90.9 \pm 3.7 \\
        & Max $PPV_\cA$ & 0.01  & \num{0.5}, \num{1.0}, \num{1.9} & 0.8 \pm 0.2 & 90.9 \pm 3.7 \\
        & Fixed FPR & 1.00 & \num{0.5}, \num{0.9}, \num{1.8} & 1.0 \pm 0.2 & 77.5 \pm 5.6 \\
        & Fixed $\phi$ & {-} & \num{0.5}, \num{0.5}, \num{0.5} & 24.0 \pm 0.8 & 57.3 \pm 0.4 \\
        & Max $Adv_\cA$ & 70.00 & \num{0}, \num{0.6}, \num{1.4} & 20.9 \pm 0.7 & 60.3 \pm 0.8 \\ \hline
        \multirow{4}{*}{\parbox{1.7cm}{\bfattack{Merlin}}} & Min FPR & 0.01 & \num{0.97 \pm 0.01} & 0.2 \pm 0.0 & 98.8 \pm 2.4 \\
        & Max $PPV_\cA$ & 0.01 & \num{0.97 \pm 0.01} & 0.2 \pm 0.0 & 98.8 \pm 2.4 \\
        & Fixed FPR & 1.00 & \num{0.88 \pm 0.00} & 2.6 \pm 0.7 & 81.7 \pm 4.3 \\
        & Max $Adv_\cA$ & 26.00 & \num{0.66 \pm 0.00} & 11.6 \pm 2.3 & 59.5 \pm 2.0 \\ \hline
        \bfattack{Morgan} & Max $PPV_\cA$ & {-} & \rm \num{1.0e-4}, \num{1.5e-3}, \num{0.95} & 0.4 \pm 0.3 & 100.0 \pm 0.0 \\
    \end{tabular}
    \caption{Thresholds selected against non-private models trained on \bfdataset{RCV1X} with balanced prior. 
    \rm 
    The results are averaged over five runs such that the target model is trained from the scratch for each run. \bfattack{Yeom CBT} and \bfattack{Shokri CBT} use class-based thresholds, where $\phi$ shows the triplet of minimum, median and maximum thresholds across all classes. All values, except $\phi$, are reported in percentage.}
    \label{tab:RCV1X_loss_selection_results_uniform_prior}
\end{table*}

\begin{figure*}[ptb]
    \centering
    \begin{subfigure}[b]{0.35\textwidth}
    \centering
    \includegraphics[width=0.95\linewidth]{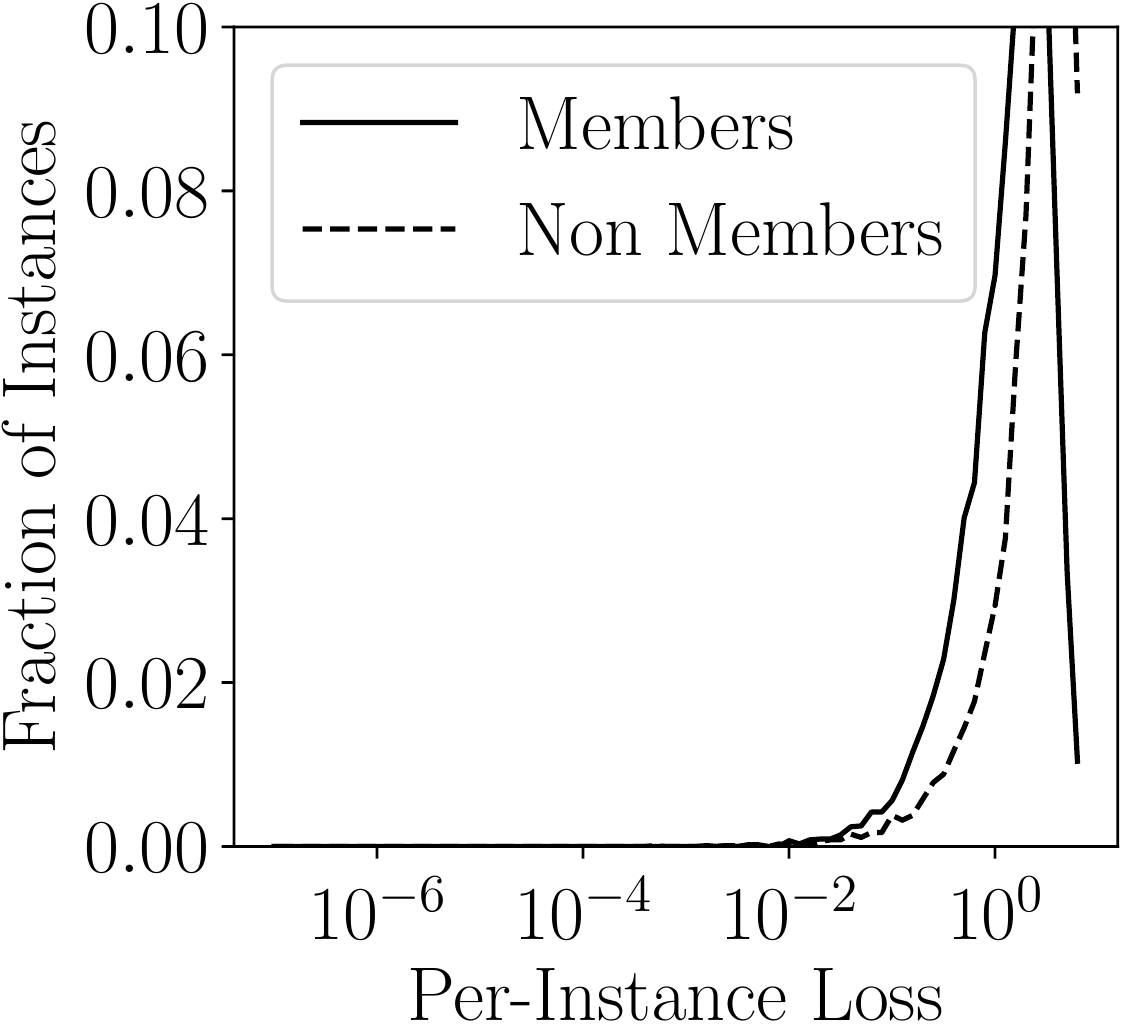}
    \caption{Loss distribution for \bfdataset{CIFAR-100}}
    \label{fig:cifar_100_mi_1_analysis_a}
    \end{subfigure}\qquad \qquad
    \begin{subfigure}[b]{0.35\textwidth}
    \centering
    \includegraphics[width=0.95\linewidth]{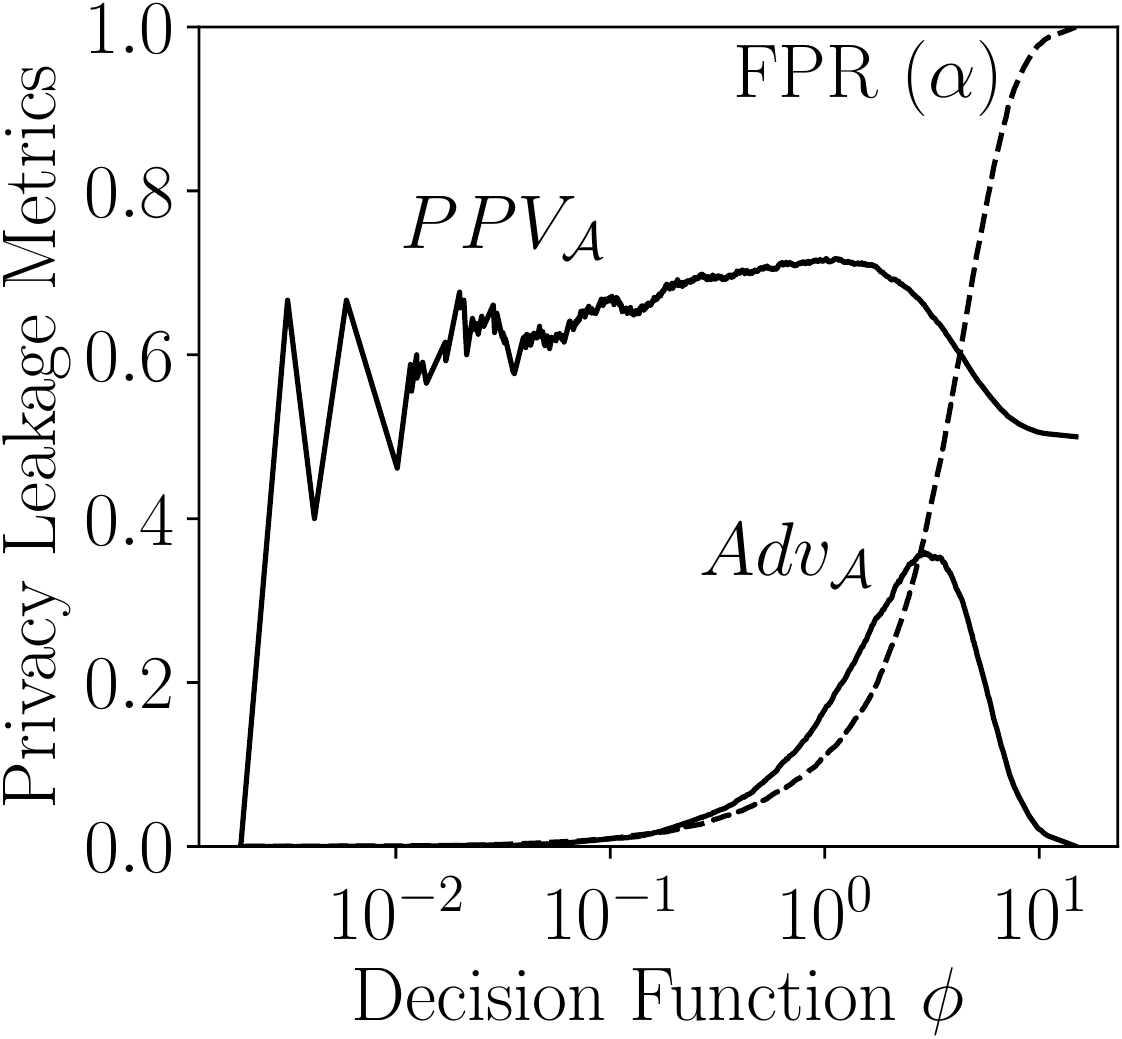}
    \caption{\bfattack{Yeom} on \bfdataset{CIFAR-100}}
    \label{fig:cifar_100_mi_1_analysis_b}
    \end{subfigure}
    \begin{subfigure}[b]{0.35\textwidth}
    \centering
    \includegraphics[width=0.95\linewidth]{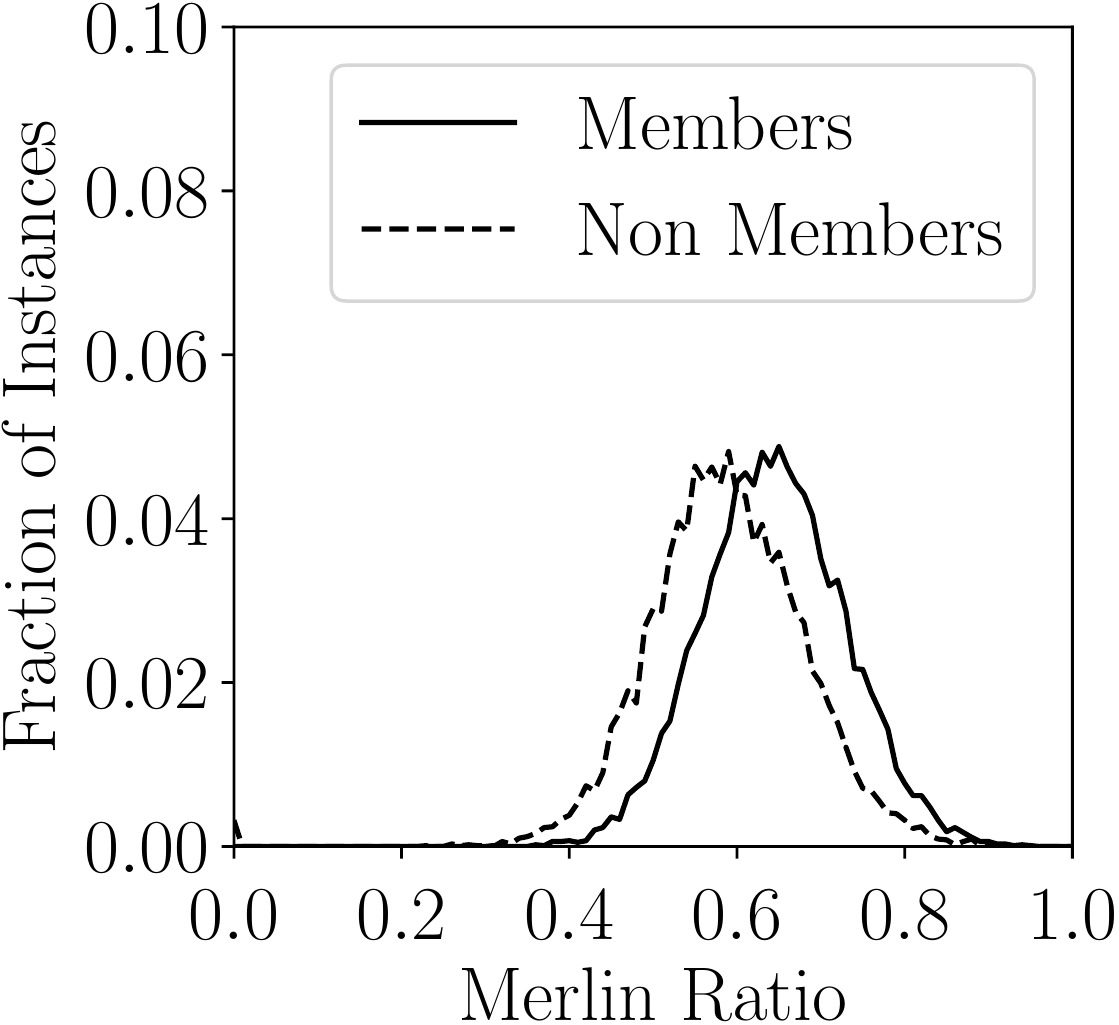}
    \caption{Merlin ratio for \bfdataset{CIFAR-100}}
    \label{fig:cifar_100_mi_2_analysis_a}
    \end{subfigure}\qquad \qquad
    \begin{subfigure}[b]{0.35\textwidth}
    \centering
    \includegraphics[width=0.95\linewidth]{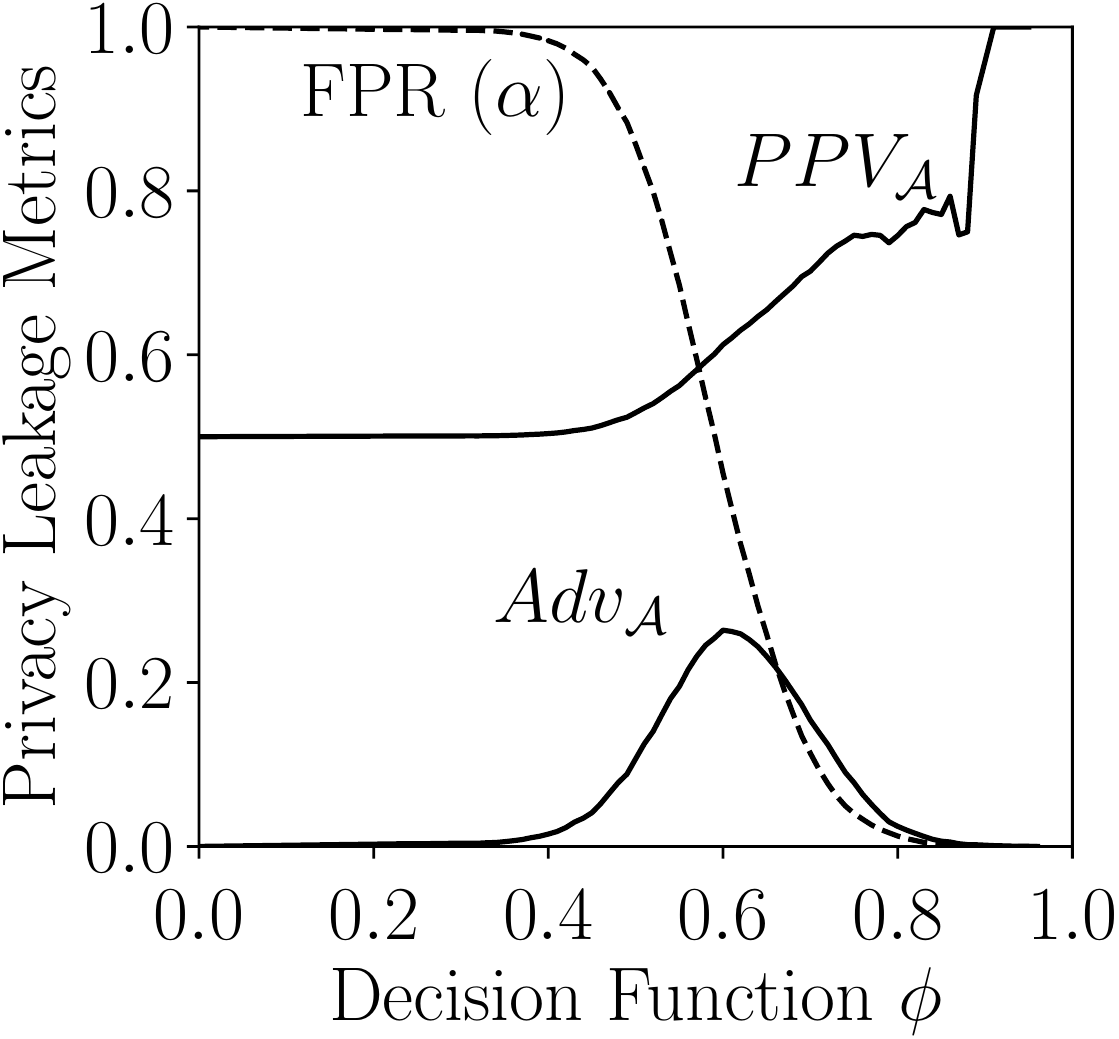}
    \caption{\bfattack{Merlin} on \bfdataset{CIFAR-100}}
    \label{fig:cifar_100_mi_2_analysis_b}
    \end{subfigure}
    \caption{Analysis of \bfattack{Yeom} and \bfattack{Merlin} against non-private models trained on \bfdataset{CIFAR-100} in the balanced prior setting.}
    \label{fig:cifar_100_analysis}
\end{figure*}

\begin{table*}[tb]
    \centering
    \small
    \begin{tabular}{llS[table-format=2.2]cS[table-format=2.1,separate-uncertainty,table-figures-uncertainty=1]S[table-format=2.1,separate-uncertainty,table-figures-uncertainty=1]}
        & & {$\alpha$} & $\phi$ & {$Adv_\cA$} & {$PPV_\cA$} \\ \hline
        \multirow{5}{*}{\parbox{1.5cm}{\bfattack{Yeom}}} & Min FPR & 0.01 & \num{4.3 \pm 2.1 e-3} & 0.0 \pm 0.0 & 33.3 \pm 27.9 \\
        & Fixed FPR & 1.00 & \num{7.2 \pm 0.8 e-2} & 0.7 \pm 0.3 & 65.1 \pm 2.6 \\
        & Max $PPV_\cA$ & 12.00 & \num{1.0 \pm 0.0} & 19.0 \pm 1.6 & 72.7 \pm 0.8 \\
        & Fixed $\phi$ & {-} &\num{2.0 \pm 0.1} & 33.0 \pm 1.6 & 70.3 \pm 1.1 \\
        & Max $Adv_\cA$ & 39.00 & \num{2.9 \pm 0.0} & 37.2 \pm 1.8 & 66.5 \pm 0.6 \\ \hline
        \multirow{5}{*}{\parbox{1.7cm}{\bfattack{Yeom CBT}}} & Min FPR & 0.01 & \rm \num{0}, \num{0.1}, \num{1.8} & 3.4 \pm 0.9 & 81.2 \pm 2.3 \\
        & Max $PPV_\cA$ & 0.01 & \rm \num{0}, \num{0.1}, \num{1.8} & 3.4 \pm 0.9 & 81.2 \pm 2.3 \\
        & Fixed FPR & 1.00 & \rm \num{0}, \num{0.2}, \num{2.0} & 5.1 \pm 0.8 & 81.0 \pm 1.4 \\
        & Fixed $\phi$ & {-} & \rm \num{0.7}, \num{2.0}, \num{3.2} & 34.0 \pm 2.7 & 71.5 \pm 0.7 \\
        & Max $Adv_\cA$ & 40.00 & \rm \num{0.5}, \num{3.0}, \num{4.6} & 37.5 \pm 1.5 & 66.6 \pm 0.6 \\ \hline
        \multirow{5}{*}{\parbox{1.5cm}{\bfattack{Shokri}}} & Min FPR & 0.01 & \num{0.99 \pm 0.00} & 16.5 \pm 1.9 & 64.9 \pm 0.7 \\
        & Max $PPV_\cA$ & 0.01 & \num{0.99 \pm 0.00} & 16.5 \pm 1.9 & 64.9 \pm 0.7 \\
        & Fixed FPR & 1.00 & \num{0.72 \pm 0.03} & 24.6 \pm 0.8 & 63.0 \pm 0.4 \\
        & Fixed $\phi$ & {-} & \num{0.50 \pm 0.00} & 26.0 \pm 0.8 & 62.5 \pm 0.4 \\
        & Max $Adv_\cA$ & 8.00 & \num{0.09 \pm 0.01} & 26.9 \pm 0.9 & 61.3 \pm 0.4 \\ \hline
        \multirow{5}{*}{\parbox{1.6cm}{\bfattack{Shokri CBT}}} & Min FPR & 0.01 & \num{0.3}, \num{0.9}, \num{1.0} & 22.1 \pm 0.6 & 63.7 \pm 0.4 \\
        & Max $PPV_\cA$ & 0.01 & \num{0.3}, \num{0.9}, \num{1.0} & 22.1 \pm 0.6 & 63.7 \pm 0.4 \\
        & Fixed $\phi$ & {-} & \num{0.5}, \num{0.5}, \num{0.5} & 26.0 \pm 0.8 & 62.5 \pm 0.4 \\
        & Fixed FPR & 1.00 & \num{0.1}, \num{0.8}, \num{1.0} & 24.2 \pm 0.6 & 63.4 \pm 0.3 \\
        & Max $Adv_\cA$ & 31.00 & \num{0.0}, \num{0.4}, \num{1.0} & 26.3 \pm 0.9 & 62.4 \pm 0.4 \\ \hline
        \multirow{4}{*}{\parbox{1.7cm}{\bfattack{Merlin}}} & Min FPR & 0.01 & \num{0.92 \pm 0.02} & 0.0 \pm 0.0 & 51.4 \pm 32.0 \\
        & Max $PPV_\cA$ & 0.90 & \num{0.82 \pm 0.00} & 1.6 \pm 0.5 & 75.0 \pm 2.6 \\
        & Fixed FPR & 1.00 & \num{0.82 \pm 0.00} & 1.8 \pm .5 & 74.7 \pm 1.7 \\
        & Max $Adv_\cA$ & 39.00 & \num{0.62 \pm 0.00} & 27.7 \pm 1.3 & 63.3 \pm 0.3 \\ \hline
        \bfattack{Morgan} & Max $PPV_\cA$ & {-} & \rm \num{2.7}, \num{3.7}, \num{0.87} & 0.0 \pm 0.0 & 100.0 \pm 0.0 \\
    \end{tabular}
    \caption{Thresholds selected against non-private models trained on \bfdataset{CIFAR-100} with balanced prior. 
    \rm 
    The results are averaged over five runs such that the target model is trained from the scratch for each run. \bfattack{Yeom CBT} and \bfattack{Shokri CBT} use class-based thresholds, where $\phi$ shows the triplet of minimum, median and maximum thresholds across all classes. All values, except $\phi$, are in percentage.}
    \label{tab:cifar_100_loss_selection_results_uniform_prior}
\end{table*}

\begin{figure*}[tb]
    \centering
\begin{subfigure}[b]{0.3\textwidth}
    \centering
    \includegraphics[width=\linewidth]{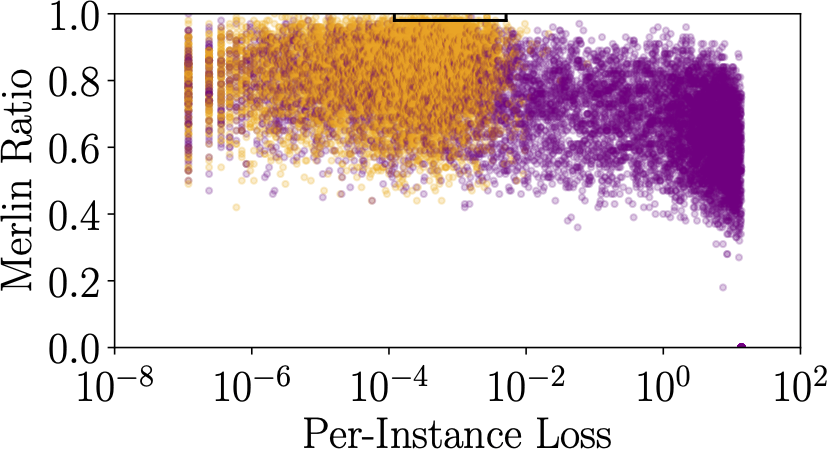}
    \caption{\bfdataset{Texas-100}}\label{fig:morgan_texas_100}
    \end{subfigure}
    \begin{subfigure}[b]{0.3\textwidth}
    \centering
    \includegraphics[width=\linewidth]{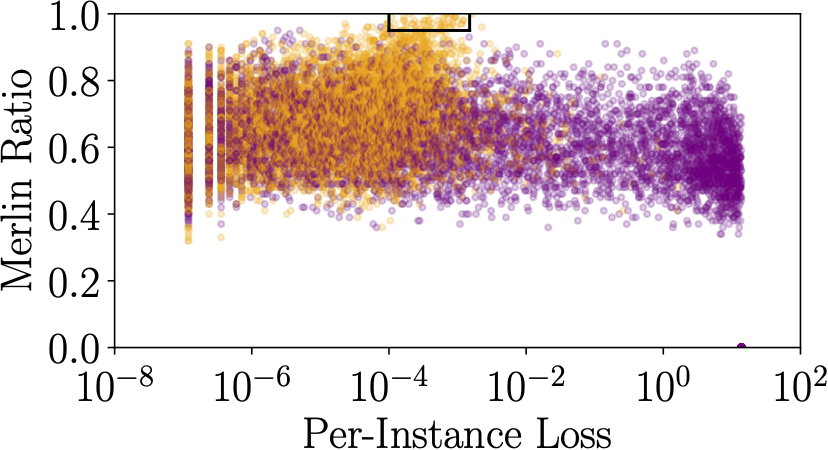}
    \caption{\bfdataset{RCV1X}}\label{fig:morgan_rcv1}
    \end{subfigure}
    \begin{subfigure}[b]{0.3\textwidth}
    \centering
    \includegraphics[width=\linewidth]{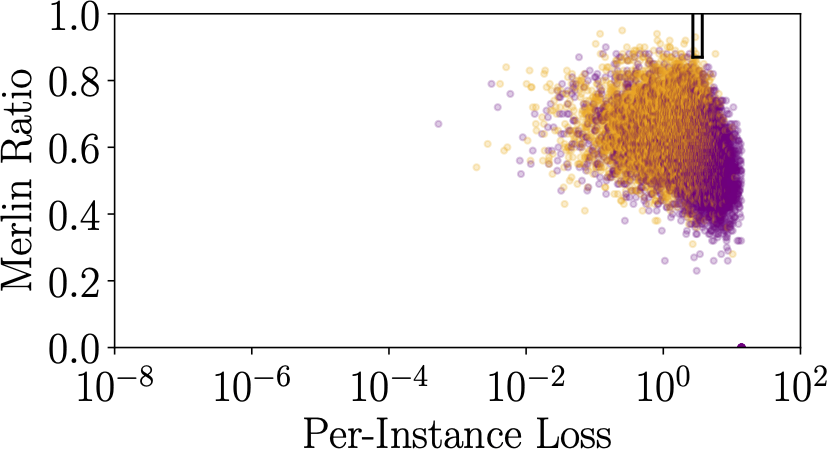}
    \caption{\bfdataset{CIFAR-100}}\label{fig:morgan_cifar_100}
    \end{subfigure}
    \caption{Comparing loss and \bfattack{Merlin} ratio for non-private models trained on different data sets at $\gamma = 1$. \rm Members and non-members are denoted by orange and purple points respectively. Highlighted boxes denote the members identified by \bfattack{Morgan} at max PPV.}
    \label{fig:morgan_plots}
\end{figure*}

\shortsection{Results on \bfdataset{Texas-100}}
We plot the distribution of per-instance loss for a non-private model trained on \bfdataset{Texas-100} in  Figure~\ref{fig:texas_100_mi_1_analysis_a}. A notable difference is that the number of non-members having zero loss is lower than that of \bfdataset{Purchase-100X}. As a result, the false positive rate can be as low as 3\% for this data set. This is depicted in Figure~\ref{fig:texas_100_mi_1_analysis_b} which shows the performance of \bfattack{Yeom} against a non-private model at different thresholds. The trend is similar to what we observe for \bfdataset{Purchase-100X}.

Figure~\ref{fig:texas_100_mi_2_analysis_a} shows the distribution of \bfattack{Merlin} ratio against a non-private model trained on \bfdataset{Texas-100}. The gap between the member and non-member distributions is greater than that of \bfdataset{Purchase-100X} and hence this attack is more effective on this data set. An important indicator is that all members have non-zero \bfattack{Merlin} ratio. The average \bfattack{Merlin} ratio is $0.81 \pm 0.12$ for members whereas it is $0.65 \pm 0.22$ for non-members. Figure~\ref{fig:texas_100_mi_2_analysis_b} shows the performance of \bfattack{Merlin} on non-private model at different count thresholds. These results further validate the effectiveness of selecting a good threshold based on our proposed procedure.  Figure~\ref{fig:morgan_texas_100} shows the scatter plot of per-instance loss and \bfattack{Merlin} ratio for all records. Similar to the case of \bfdataset{Purchase-100X}, more fraction of members are concentrated between \num{1.2e-4} and \num{5.1e-3} loss and have \bfattack{Merlin} ratio greater than \num{0.90}. Table~\ref{tab:texas_100_loss_selection_results_uniform_prior} compares the membership inference attacks across different attack settings on \bfdataset{Texas-100}. As shown, \bfattack{Merlin} achieves higher PPV values than both \bfattack{Yeom} and \bfattack{Shokri}. Using class based thresholds drastically improves PPV for \bfattack{Yeom} such that \bfattack{Yeom CBT} achieves maximum PPV comparable to \bfattack{Merlin}. As with \bfdataset{Purchase-100X}, we observe no benefit of using CBT for \bfattack{Merlin}. While \bfattack{Shokri} achieves 89\% PPV, slightly less than \bfattack{Merlin}, on this data set, using CBT decreases the PPV to 85\%. \bfattack{Morgan} achieves highest PPV among all attacks.

\shortsection{Results on \bfdataset{RCV1X}}
We plot the distribution of per-instance loss of members and non-members for a non-private model trained on \bfdataset{RCV1X} in Figure~\ref{fig:rcv1_mi_1_analysis_a}. While more members are concentrated closer to zero loss than the non-members, we observe that the gap between the two distributions is not as large as with the other data sets. Moreover, $3504 \pm 444$ non-members have zero loss, and hence the minimum possible false positive rate for \bfattack{Yeom} is around 33\%.  Figure~\ref{fig:rcv1_mi_1_analysis_b} shows the performance of \bfattack{Yeom} for different loss thresholds. The maximum PPV that can be achieved using this attack is only around 58\%, at which point the advantage is close to 27\%. Thus while the advantage metric would suggest that there is privacy risk, \bfattack{Yeom} does not pose significant risk as PPV is close to 50\% for balanced prior. \bfattack{Shokri}, on the other hand, achieves a PPV of 91\% (see Table~\ref{tab:leakage_comparison_balanced_prior}) and poses a significant privacy risk.

Figure~\ref{fig:rcv1_mi_2_analysis_a} shows the distribution of \bfattack{Merlin} ratio for a non-private model trained on \bfdataset{RCV1X}. While the gap between distributions is small, the PPV can still be high as depicted in Figure~\ref{fig:rcv1_mi_2_analysis_b}. \bfattack{Merlin} achieves a maximum PPV of around 99\% on an average for threshold values close to 0.97, and hence poses privacy threat even when \bfattack{Yeom} fails. Table~\ref{tab:RCV1X_loss_selection_results_uniform_prior} compares the attacks on \bfdataset{RCV1X} for different attack goals. \bfattack{Yeom} is benefited from using class based thresholds, as the maximum PPV jumps from 58\% to 93\%. However, \bfattack{Merlin} still outperforms \bfattack{Yeom CBT} at maximum PPV setting. As with other data sets, \bfattack{Shokri} does not benefit from CBT technique. Figure~\ref{fig:morgan_rcv1} shows the loss and \bfattack{Merlin} ratio scatter plot on \bfdataset{RCV1X}. Though the members and non-members are less differentiated, \bfattack{Morgan} is still able to identify the most vulnerable members with 100\% confidence (see Table~\ref{tab:RCV1X_loss_selection_results_uniform_prior}).

\begin{table*}[ptb]
    \centering
    \small
    \begin{tabular}{c|cS[table-format=3.0]S[table-format=2.2]S[table-format=2.2,separate-uncertainty,table-figures-uncertainty=1]S[table-format=2.1,separate-uncertainty,table-figures-uncertainty=1]}
        & & {$\epsilon$} & {$\alpha$} & {$\phi$} & {Max $PPV_\cA$} \\ \hline
        \multirow{12}{*}{\rm\bfdataset{Texas-100}} & \multirow{3}{*}{\bfattack{Yeom}} & 1 & 0.05 & \num{1.0 \pm 0.5 e-2} & 58.5 \pm 4.6 \\
        & & 10 & 0.06 & \num{1.0 \pm 0.4 e-5} & 65.5 \pm 19.8 \\
        & & 100 & 0.02 & \num{0.2 \pm 0.2 e-5} & 58.8 \pm 24.0 \\ \cline{2-6}
        & \multirow{3}{*}{\bfattack{Shokri}} & 1 & 52.00 & \num{0.00 \pm 0.00} & 51.7 \pm 1.0 \\
        & & 10 & 22.00 & \num{0.60 \pm 0.01} & 51.5 \pm 1.7 \\
        & & 100 & 18.00 & \num{0.64 \pm 0.00} & 55.4 \pm 1.2 \\ \cline{2-6}
        & \multirow{3}{*}{\bfattack{Merlin}} & 1 & 0.13 & \num{0.92 \pm 0.00} & 61.0 \pm 5.4 \\
        & & 10 & 0.05 & \num{0.94 \pm 0.00} & 67.2 \pm 19.3 \\
        & & 100 & 0.45 & \num{0.92 \pm 0.00} & 59.5 \pm 3.6 \\ \cline{2-6}
        & \multirow{3}{*}{\bfattack{Morgan}} & 1 & {-} & \rm \num{0.3}, \num{1.4}, \num{0.93} & 85.6 \pm 19.8 \\
        & & 10 & {-} & \rm \num{5.4e-2}, \num{0.2}, \num{0.93} & 76.7 \pm 26.1 \\
        & & 100 & {-} & \rm \num{0}, \num{8.4e-5}, \num{0.92} & 68.2 \pm 17.9 \\ \hline
        
        \multirow{12}{*}{\rm\bfdataset{RCV1X}} & \multirow{3}{*}{\bfattack{Yeom}} & 1 & 13.00 & \num{6.0 \pm 1.1 e-4} & 51.7 \pm 0.5 \\
        & & 10 & 60.00 & \num{2.4 \pm 0.3 e-2} & 51.8 \pm 0.2 \\
        & & 100 & 70.00 & \num{3.7 \pm 0.3 e-2} & 53.0 \pm 0.1 \\ \cline{2-6}
        & \multirow{3}{*}{\bfattack{Shokri}} & 1 & 60.00 & \num{0.51 \pm 0.01} & \num{50.6 \pm 0.3} \\
        & & 10 & 10.00 &  \num{0.61 \pm 0.01} &   \num{55.0 \pm 1.8}\\
        & & 100 & 2.20 &  \num{0.72 \pm 0.01} &  \num{60.6 \pm 3.8} \\ \cline{2-6}
        & \multirow{3}{*}{\bfattack{Merlin}} & 1 & 3.00 & \num{0.80 \pm 0.01} & 52.6 \pm 2.0 \\
        & & 10 & 0.10 & \num{0.89 \pm 0.01} & 70.9 \pm 12.9 \\
        & & 100 & 0.04 & \num{0.92 \pm 0.01} & 86.9 \pm 11.6 \\ \cline{2-6}
        & \multirow{3}{*}{\bfattack{Morgan}} & 1 & {-} & \rm \num{0}, \num{5.0e-6}, \num{0.80} & 75.0 \pm 21.1 \\
        & & 10 & {-} & \rm \num{4.7e-5}, \num{3.6}, \num{0.89} & 77.4 \pm 11.0 \\
        & & 100 & {-} & \rm \num{2.7e-5}, \num{12}, \num{0.92} & 89.1 \pm 11.3 \\ \hline
        
        \multirow{12}{*}{\rm\bfdataset{CIFAR-100}} & \multirow{3}{*}{\bfattack{Yeom}} & 1 & 0.11 & \num{4.3 \pm 0.0} & 68.4 \pm 27.1 \\
        & & 10 & 0.11 & \num{1.2 \pm 0.1} & 64.2 \pm 29.4 \\
        & & 100 & 0.80 & \num{1.3 \pm 0.0} & 56.3 \pm 3.2 \\ \cline{2-6}
        & \multirow{3}{*}{\bfattack{Shokri}} & 1 & 0.08 & \num{0.95 \pm 0.02} & 50.1 \pm 0.1 \\
        & & 10 & 0.20 & \num{0.89 \pm 0.02} & 50.1 \pm 0.1 \\
        & & 100 & 0.06 & \num{0.96 \pm 0.02} & 50.2 \pm 0.3 \\ \cline{2-6}
        & \multirow{3}{*}{\bfattack{Merlin}} & 1 & 0.11 & \num{0.85 \pm 0.01} & 57.3 \pm 10.1 \\
        & & 10 & 0.07 & \num{0.79 \pm 0.01} & 66.6 \pm 17.6 \\
        & & 100 & 0.12 & \num{0.77 \pm 0.01} & 62.0 \pm 11.5 \\ \cline{2-6}
        & \multirow{3}{*}{\bfattack{Morgan}} & 1 & {-} & \rm \num{4.5}, \num{4.8}, \num{0.85} & 62.7 \pm 7.7 \\
        & & 10 & {-} & \rm \num{0.7}, \num{3.1}, \num{0.79} & 77.3 \pm 12.6 \\
        & & 100 & {-} & \rm \num{1.4}, \num{2.2}, \num{0.77} & 89.7 \pm 13.5 \\
    \end{tabular}
    \caption{MI attacks against private models trained on different data sets in the balanced prior setting. \rm $\alpha$ and PPV values are in percentage.}
    \label{tab:gamma_1_mi_dp_comparison}
\end{table*}

\shortsection{Results on \bfdataset{CIFAR-100}}
Figure~\ref{fig:cifar_100_mi_1_analysis_a} shows the distribution of per-instance loss for a non-private model trained on \bfdataset{CIFAR-100}. The loss of both members and non-members is high, since the model does not completely overfit on this data set.  Figure~\ref{fig:cifar_100_mi_1_analysis_b} shows the performance of \bfattack{Yeom} for different loss thresholds. Figures~\ref{fig:cifar_100_mi_2_analysis_a} and \ref{fig:cifar_100_mi_2_analysis_b} show the distribution of \bfattack{Merlin} ratio and leakage metrics for different thresholds. Using CBT on \bfattack{Yeom} increases the maximum PPV from 73\% to 81\% (Table~\ref{tab:cifar_100_loss_selection_results_uniform_prior}), exceeding that of \bfattack{Merlin}. \bfattack{Shokri} is less successful on this data set, achieving only 65\% PPV, and does not benefit from the CBT technique. Figure~\ref{fig:morgan_cifar_100} shows the loss and \bfattack{Merlin} ratio of all records on \bfdataset{CIFAR-100}. As shown, members with high \bfattack{Merlin} ratio are distinguishable from non-members. \bfattack{Morgan} is able to identify certain members with 100\% PPV (see Table~\ref{tab:cifar_100_loss_selection_results_uniform_prior}).

\section{Additional Results for Private Models}\label{appendix:private_results}
The plots for private models on all three data sets are similar to that of \bfdataset{Purchase-100X} and do not convey any additional information, hence we do not include them. They only corroborate the fact that adding privacy noise reduces the gap between member and non-member distributions and in turn limits the attack success. Instead, we directly compare the maximum PPV of the attacks against private models trained with varying privacy loss budgets across all three data sets in Table~\ref{tab:gamma_1_mi_dp_comparison}. 
As with \bfdataset{Purchase-100X}, adding noise allows \bfattack{Yeom} to set much smaller thresholds on \bfdataset{Texas-100}. For higher $\epsilon$ values, \bfattack{Yeom} poses some privacy threat but the PPV deviation is high. On \bfdataset{RCV1X}, the $\alpha$ values are still high and hence \bfattack{Yeom} is not successful even for $\epsilon = 100$. On \bfdataset{CIFAR-100}, \bfattack{Yeom} is able to achieve considerably high PPV values for $\epsilon = 1$ and $\epsilon = 10$, but the deviation is very high, indicating that the attack is only successful in some runs. At $\epsilon = 100$, \bfattack{Yeom} fails to pose any threat. \bfattack{Shokri} achieves close to 50\% PPV across all data sets, and hence fails to pose any privacy threat even against models trained with large privacy loss budgets. \bfattack{Merlin} achieves higher PPV than both \bfattack{Yeom} and \bfattack{Shokri} on an average across all data sets. Similar to \bfdataset{Purchase-100X}, \bfattack{Merlin} achieves high PPV values for $\epsilon = 10$ and $\epsilon = 100$ on \bfdataset{RCV1X}. However, it does not achieve high enough PPV on \bfdataset{Texas-100} and \bfdataset{CIFAR-100} to pose serious privacy threat, even for $\epsilon = 100$. On the other hand, \bfattack{Morgan} poses serious privacy threat against models trained with high privacy loss budgets across all the tested data sets.

\end{document}